\newif\ifappendix
\newif\ifnoappendix
\appendixtrue
\noappendixfalse

\ifappendix
\documentclass[sigconf]{acmart}
\else 
\documentclass[acmsmall]{acmart}
\fi
\AtBeginDocument{%
  }

\setcopyright{cc}
\setcctype{by-nc-nd}
\acmJournal{PACMMOD}
\acmYear{2026} \acmVolume{4} \acmNumber{3 (SIGMOD)} \acmArticle{239}
\acmMonth{6} \acmDOI{10.1145/3802116}

\usepackage{array}
\usepackage{breqn}
\usepackage{booktabs,makecell,multirow}
\usepackage[makeroom]{cancel}
\usepackage{caption}
\usepackage{enumitem}
\usepackage{wrapfig}
\usepackage{tikz,pgfplots}
\usetikzlibrary{positioning}
\usetikzlibrary{arrows}
\usetikzlibrary{patterns}
\usetikzlibrary{shapes.geometric}
\usetikzlibrary{decorations.pathreplacing}
\usetikzlibrary{calligraphy}
\usetikzlibrary{hobby}
\usetikzlibrary{calc}

\usepackage{pifont}

\usepackage{thm-restate}
\newtheorem{theorem}{Theorem}[section]

\newtheorem{lemma}[theorem]{Lemma}
\newtheorem{claim}[theorem]{Claim}

\usepackage{cleveref}

\usepackage{xspace}
\newcommand{\sketch}{Sublime}
\newcommand{\sketchcms}{\sketch\textsubscript{CMS}\xspace}
\newcommand{\sketchcs}{\sketch\textsubscript{CS}\xspace}
\newcommand{\sketchmg}{\sketch\textsubscript{MG}\xspace}
\newcommand{\counterencoding}{\underline{va}riable-\underline{l}ength count\underline{e}r}
\newcommand{\counterencodingabbrv}{VALE}
\newcommand{\parentheses}[1]{\left(#1\right)}
\newcommand{\expectation}[1]{\mathbb{E}\left[ #1 \right]}
\newcommand{\variance}[1]{\mathrm{Var}\left( #1 \right)}

\graphicspath{{./figures/}}

\ifappendix
\usepackage{draftwatermark}
\SetWatermarkScale{0.2}
\SetWatermarkAngle{270}
\SetWatermarkText{Accepted to SIGMOD 2026}
\SetWatermarkHorCenter{0.95\paperwidth}
\SetWatermarkVerCenter{0.5\paperheight}
\fi

\begin{document}

\title{\sketch: Sublinear~Error~\&~Space~for~Unbounded~Skewed~Streams}

\author{Navid Eslami}
\email{navideslami@cs.toronto.edu}
\affiliation{%
  \institution{University of Toronto}
  \city{Toronto}
  \country{Canada}
}

\author{Ioana O. Bercea}
\email{bercea@kth.se}
\affiliation{%
  \institution{KTH Royal Institute of Technology}
  \city{Stockholm}
  \country{Sweden}
}

\author{Rasmus Pagh}
\email{pagh@di.ku.dk}
\affiliation{%
  \institution{BARC, University of Copenhagen}
  \city{Copenhagen}
  \country{Denmark}
}

\author{Niv Dayan}
\email{nivdayan@cs.toronto.edu}
\affiliation{%
  \institution{University of Toronto}
  \city{Toronto}
  \country{Canada}
}


\begin{abstract}
    Modern stream processing systems often need to track the frequency of
    distinct keys in a data stream in real-time. Since maintaining exact counts
    can require a prohibitive amount of memory, many applications rely on
    compact, probabilistic data structures known as frequency estimation
    sketches to approximate them. However, mainstream frequency estimation
    sketches fall short in two critical aspects. First, they are
    memory-inefficient under skewed workloads because they use uniformly-sized
    counters to count the keys, thus wasting memory on storing the leading
    zeros of many small counts. Second, their estimation error deteriorates at
    least linearly with the length of the stream{\textemdash}which may grow
    indefinitely{\textemdash}because they rely on a fixed number of counters.

    We present \sketch, a framework that generalizes frequency estimation
    sketches to address these challenges. To reduce memory footprint under
    skew, \sketch\ begins with short counters and dynamically elongates them as
    they overflow, storing their extensions within the same cache line. It
    employs efficient bit manipulation routines to quickly locate and access a
    counter's extensions. To maintain accuracy as the stream grows, \sketch\
    also expands its number of counters at a configurable rate, exposing a new
    spectrum of accuracy-memory tradeoffs that applications can tune to their
    needs. We apply \sketch\ to both Count-Min Sketch and Count Sketch. Through
    theoretical analysis and empirical evaluation, we show that \sketch\
    significantly improves accuracy and memory over the state of the art while
    maintaining competitive or \mbox{superior~performance}.
\end{abstract}



\received{17 October 2025}
\received[revised]{5 February 2026}
\received[accepted]{11 March 2026}

\maketitle

\section{Introduction}~\label{sec:introduction}
\textbf{Data Streams.}
With data volumes growing exponentially, it has become infeasible for many
applications to compute statistics by scanning entire datasets. At the same
time, many modern systems must monitor and analyze data streams that grow at
high speeds in real time, where storing the full stream (e.g., all packets
transmitted over a network) is impractical. Nevertheless, users rely on these
statistics to make timely decisions and extract insights from the data (e.g.,
to detect anomalies in network
traffic~\cite{NetworkAnomalyDetection1,NetworkAnomalyDetection2}, identify
trending content on social media networks~\cite{ScoutSketch}, or monitor system
performance~\cite{OctoSketch,Canopy}). One way to address this challenge is to
incrementally maintain statistics as new data arrives or old data expires. This
approach works well for simple metrics such as sums or means, which can be
updated exactly using basic arithmetic operations~\cite{Canopy}. However, many
other important statistics, such as cardinalities, quantiles, or item
frequencies, cannot be maintained so easily. To address this challenge, data
sketches have emerged as a popular approach. A data sketch is a compact data
structure maintaining a small amount of information about the data that
approximates a given statistic while providing accuracy guarantees.

\textbf{Frequency Estimation Sketches.}
One of the most fundamental statistics in data stream analysis is
frequency{\textemdash}the number of times each key appears in the stream.
Tracking these counts exactly requires maintaining a counter for every distinct
key, which leads to a substantial memory overhead if the stream contains many
unique keys. To overcome this issue, a \emph{Frequency Estimation Sketch}
approximates the counts using a compact, fixed-size structure. At a high level,
this structure maps the key domain to a smaller codomain of counters (e.g., via
hashing) and uses the counter a key maps to as an estimate of that key's
frequency. In this way, a frequency estimation sketch significantly reduces
memory usage, though it introduces estimation errors due to the possibility of
multiple keys mapping to the same counter. Frequency estimation sketches have
become key components across many systems, including query
optimizers~\cite{Compass,QueryOptimizationSketches,ConvolutionSketch},
network~monitors~\cite{SketchBasedChangeDetection,WhatsNew,SketchingStreamsNet,TrafficManagementSurvey,RandomizedAdmissionEstimation,Blink,LOFT},
streaming SQL engines~\cite{Flink}, sensor data
analyzers~\cite{SensorDataAnalytics,RobustAggregationSensorNetworks}, and
genome analysis pipelines~\cite{CMSkmer,SetMinSketch}.

\textbf{Different Types of Error.}
Frequency estimation sketches differ in the nature of the errors they
introduce. The Count-Min Sketch~\cite{CountMin} either matches or overestimates
a key's true count. This one-sided error makes the Count-Min Sketch suitable
for security-critical applications such as detecting
denial-of-service~\cite{CoDDoS} and hardware row hammer~\cite{CoMeT} attacks,
where misclassifying malicious behavior is unacceptable. The Count
Sketch~\cite{CountSketch} can both overestimate and underestimate a key's
count. Yet, these errors cancel out on average, yielding unbiased estimates.
This property prevents the error from compounding when combining or multiplying
estimates, making the Count Sketch useful for tasks such as query
optimization~\cite{Compass,QueryOptimizationSketches,ConvolutionSketch} and
covariance estimation~\cite{ASCS}. The Misra-Gries structure~\cite{MisraGries}
either matches or underestimates a key's true count. This makes Misra-Gries
effective in identifying the most frequent keys (i.e., the heavy
hitters)~\cite{NetworkAnomalyDetection1,NetworkAnomalyDetection2}, since the
keys it tracks are guaranteed to have a relatively high frequency. Many
variants of these data structures have been proposed, each offering distinct
tradeoffs between accuracy, space, and
performance~\cite{CountMeanMin,SALSA,BitSense,PyramidSketch,StingySketch,AdaptiveCounterSplicing,CounterTree,TreeSensing,TailoredSketch,AugmentedSketch,HeavyGuardian,HeavyKeeper,WavingSketch,ElasticSketch,MVSketch,SpaceSaving,BatchUpdateMisraGries,BitMatcher,JigsawSketch,MicroscopeSketch,PSketch,TightSketch}.

\textbf{Goals.}
Ideally, a frequency estimation sketch should provide high accuracy using a
modest memory footprint. At the same time, it should attain consistently high
performance as a data stream grows. This paper shows that all existing
frequency estimation sketches fall short of at least one of these goals due to
two~core~problems:

\textbf{Problem~1: Skew.}
In many real-world data streams, key frequencies are highly skewed, with a
small number of heavy hitters appearing far more than other keys (e.g.,
following a Zipfian or, more generally, a heavy-tailed
distribution)~\cite{Memento,AdditiveErrorCounters,StingySketch}. As a result,
most counters within a frequency estimation sketch hold small values while only
a few hold large values. Existing frequency estimation sketches use
uniformly-sized counters, which must be long enough to represent the maximum
counter value. Consequently, they waste a significant amount of memory storing
the unused higher-order bits of small counter values. Several prior works
attempt to address this problem by utilizing smaller counters and merging or
sharing those that
overflow~\cite{SALSA,BitSense,PyramidSketch,StingySketch,AdaptiveCounterSplicing,CounterTree,TreeSensing,TailoredSketch}.
We show that these approaches compromise either accuracy or performance.
Therefore, minimizing memory wastage under skew remains a challenge.

\textbf{Problem~2: Unbounded Data Growth.}
In most applications, data streams continuously grow over
time~\cite{Compass,QueryOptimizationSketches,ConvolutionSketch,SketchBasedChangeDetection,WhatsNew,TrafficManagementSurvey}.
Yet, all existing frequency estimation sketches are allocated with a fixed size
from the get-go. This causes their estimation error to scale at least linearly
with the stream's unknown length. One cannot construct a larger frequency
estimation sketch to control the error without losing all prior information, as
the stream is too large to be stored, rescanned, and reinserted into a new
sketch. Thus, applications are unable to maintain tight~error~bounds~as the
stream grows. This problem is critical when deploying frequency estimation
sketches at scale, as recently reported by the creators of the Apache Data
Sketches library~\cite{SimonsYahoo,DataSketches}. Recent works have tackled the
challenge of accommodating unbounded data growth for filter data
structures~\cite{PaghExpandability,TaffyFilters,InfiniFilter,AlephFilter,Zeno,BeyondBloomTutorial,Memento,Diva,Aeris}.
This paper is the first to address this problem for frequency estimation
sketches. In fact, to the best of our knowledge, this problem has not
previously been studied in the context of any type of sketch. We expect this
work to open up interesting directions for sketches supporting other
statistics, such as cardinalities~\cite{HLL,HLLL,ULL},
quantiles~\cite{GK,KLL,SplineSketch}, graph~analytics~\cite{LpSamplers,AGM}.

\textbf{Solution: \sketch.}
We introduce \sketch\footnote{\sketch\ is a play on words referring to the
\textbf{subli}n\textbf{e}arly scaling error and memory footprint of our
framework. It also describes, in the philosophical tradition of Immanuel Kant,
the feeling of awe experienced when contemplating something vast or infinite.
This notion of infinity resonates with the theme of our work, which addresses
the challenge of accommodating indefinite data growth.}, a framework that
generalizes frequency estimation sketches to address the above two challenges.
To address Problem~1 (Skew), \sketch\ allocates short counters upfront and
dynamically extends them as they overflow. It colocates each counter and its
extension within the same cache line and employs specialized bit manipulation
routines to decode them in constant time. To resolve Problem~2 (Unbounded Data
Growth), \sketch\ expands the overall number of counters as the stream grows.
In doing so, it achieves sublinear error bounds and memory footprint with
respect to the stream's length. It further introduces a new Pareto frontier
between these metrics, allowing applications to pick the tradeoff that best
suits their needs.

\textbf{Additional Contributions.}
\begin{enumerate}
    \item We show that while traditional Count-Min and Count Sketches exhibit
        an estimation error that grows linearly with the stream size, variants
        designed to handle
        skew~\cite{BitSense,PyramidSketch,StingySketch,AdaptiveCounterSplicing,CounterTree,TreeSensing,TailoredSketch,SALSA}
        cause the error to grow super-linearly, raising the question of how to
        address Problem~1 without exacerbating Problem~2
        (\Cref{sec:problem_analysis_skew}).

    \item We are the first to identify the problem of designing frequency
        estimation sketches that maintain both error and memory usage sublinear
        in the stream length
        (\Cref{sec:problem_analysis_unbounded_growth}).

    \item We introduce a constant-time \underline{va}riable-\underline{l}ength
        data \underline{e}ncoding scheme~(\counterencodingabbrv), applicable
        beyond frequency estimation sketches, and use it to encode
        variable-length counters
        (\mbox{\Cref{sec:accommodating_skew}}).

    \item We describe how to expand and contract a frequency estimation sketch
        based on the stream's length while supporting constant-time queries and
        enabling fine-grained control over the accuracy-memory tradeoff. We also
        characterize this new tradeoff space
        (\Cref{sec:accommodating_unknown_stream_lengths}).

    \item To demonstrate \sketch's generality, we apply it to three
        representative frequency estimation sketches: Count-Min Sketch
        (Sections~\ref{sec:accommodating_skew}
        and
        \ref{sec:accommodating_unknown_stream_lengths}),
        Count Sketch
        (\Cref{sec:cs}),
        and Misra-Gries
    \ifappendix
        (\Cref{sec:mg}).
    \else
        (the Appendix~\cite{SublimeArxiv}).
    \fi

    \item We prove a lower bound on the minimum space required for any
        expandable frequency estimation sketch to achieve a target error. We
        show that \sketch's memory footprint approximately meets this bound
        (\Cref{sec:memory_analysis}).

    \item We evaluate \sketch\ against state-of-the-art frequency estimation
        sketches and show that it simultaneously 1)~reduces the memory
        footprint under skew, 2)~tightly bounds error as the stream grows, and
        3)~maintains equal or better performance. We also apply \sketch\ to
        estimate the size of a join between two tables and show that it
        produces more accurate results than other frequency estimation sketches
        (\Cref{sec:evaluation}).
\end{enumerate}
\vspace{-2mm}

\begin{table}
    \centering
    \bgroup
    \def\arraystretch{1.025}
    \small
    \begin{tabular}{cc}
        \toprule
        \textbf{Symbol} & \textbf{Definition} \\
        \midrule
        $N$ & The sum of the counts of all keys in the stream. \\
        $x, y$ & Example keys from the stream. \\
        $f(x)$ & The ground-truth count of key~$x$. \\
        $\hat{f}(x)$ & The estimated count of key~$x$. \\
        $w$ & Number of counters/slots. \\
        $d$ & Number of independent arrays. \\
        $h_i(x)$ & Hash function mapping key~$x$ to a counter. \\
        $\sigma_i(x)$ & Hash function defining key~$x$'s update direction ($\pm 1$). \\
        %$N^{\text{res}(i)}$ & Total count of all but the~$i$ most frequent keys. \\
        \cmidrule{1-2}
        $c$ & Number of counters in each chunk. \\
        $s$ & Length of each stub in bits. \\
        $W(\cdot)$ & Size function. \\
        %\cmidrule{1-2}
        %AAE & Average Absolute Error, i.e., $\text{avg}_x|\hat{f}(x) - f(x)|$. \\
        %ARE & Average Relative Error, i.e., $\text{avg}_x\frac{|\hat{f}(x) - f(x)|}{f(x)}$. \\
        %\makecell[c]{Top-$w$ \\[-2pt] AAE} & Average Absolute Error among the~$w$ most frequent items. \\
        \bottomrule
    \end{tabular}
    \egroup
    \caption{Definitions of terms and symbols.}
    \ifnoappendix
      \vspace{-8mm}
    \fi
    \label{tab:term_and_symbol_definitions}
\end{table}

\section{Background}~\label{sec:background}
A stream is a sequence of key-value pairs~$\left\langle (x_1,\Delta_1),
(x_2,\Delta_2), \dots \right\rangle$, where~$x_i$ is the $i$-th key that
appears in the stream and~$\Delta_i = \pm1$ represents whether its count was
incremented due to an insertion ($\Delta_i=1$) or decremented due to a deletion
($\Delta_i=-1$). This is known as the \emph{General Turnstile Model} of data
streams~\cite{StreamModels}. We denote key~$x$'s current count as~$f(x)$,
accounting for its insertions minus its deletions. We represent the current
total count of the keys in the stream with~$N=\sum_x f(x)=\sum_i \Delta_i$.
\Cref{tab:term_and_symbol_definitions}
lists the terms and symbols used
throughout the paper.

\textbf{Exhaustive Approach.}
The simplest approach to tracking frequencies is to store each key and its
count in a hash table. Although this approach is fully accurate, it entails a
high memory overhead for three reasons: (i)~It stores all keys in the stream.
(ii)~It stores one counter for each unique key. (iii)~The hash table can be as
much as half empty to support collision resolution and expansions. These
different overheads multiply with the number of streams an application
monitors. These problems imply a memory footprint that is linear in the number
of unique keys in the stream, which can be of the same magnitude as the total
key count~$N$.

Frequency estimation (FE) sketches address these problems by trading accuracy
for space. That is, for any key~$x$, an FE sketch returns an
estimate~$\hat{f}(x)$ of key $x$'s count~$f(x)$. We focus on Count-Min
Sketch~\cite{CountMin} for now, as it is the simplest and most widely used FE
sketch~\cite{ADSMassiveDatasetBook,SketchesTrafficManagementSurvey}.
\ifappendix
Later, we describe Count Sketch~\cite{CountSketch}
(\Cref{sec:cs}) and
Misra-Gries~\cite{MisraGries}
(\Cref{sec:mg}) and apply our
techniques~to~them. 
\else
Later in \Cref{sec:cs}, we
describe Count Sketch~\cite{CountSketch} and apply our techniques to it. We do
the same for Misra-Gries~\cite{MisraGries} in the~Appendix~\cite{SublimeArxiv}. 
\fi

\textbf{Count-Min Sketch (CMS).}
The core design element of CMS is an array consisting of~$w$ counters, each
initialized to zero. We insert a key~$x$ into this array by hashing it to one
counter and incrementing it. We delete key~$x$ by decrementing its
corresponding counter.\footnote{Since CMS does not store the actual keys, users
must ensure that deletions only target previously inserted keys.} The counter
that key~$x$ hashes to provides an estimate of~$x$'s count. Due to hash
collisions, however, multiple keys can map to the same counter. Because of
this, CMS may overestimate a key's true count. This error is at most~$N/w$ in
expectation since the~$N$ keys are randomly hashed into~$w$ counters.
Increasing the array's size makes collisions less likely, thereby improving
accuracy. Even then, the error may deviate from its expectation. The
probability that errors do not exceed the expectation ``too much'' is known as
the sketch's \emph{Confidence}. An application of Markov's inequality shows
that for any key~$x$, the error does not exceed~$e \cdot N/w$ with a confidence
of at least~$1-e^{-1}$, i.e., $\Pr(|\hat{f}(x)-f(x)| \leq e \cdot N/w) \geq
1-e^{-1}$. Here, we use Euler's constant~$e$ due to convention. Any other
constant would also work and would replace~$e$ in both the error bound and the
confidence. This implies that the smaller the constant becomes, the more
accurate we expect the FE sketch to be and the lower confidence drops.

\begin{figure}
    \centering
    \begin{tikzpicture}
        \def\arraycount{3}
        \def\arrayw{5}
        \ifappendix
            \def\arrayh{0.25}
        \else
            \def\arrayh{0.3}
        \fi
        \def\arraylen{10}
        \ifappendix
            \def\arraysep{0.25}
        \else
            \def\arraysep{0.3}
        \fi
        \def\counterw{0.5}

        \foreach \i in {1, ..., \arraycount} {
            \draw[black] (0,-\i*\arrayh+1.5*\arrayh-\i*\arraysep+\arraysep) rectangle (\arrayw,-\i*\arrayh+0.5*\arrayh-\i*\arraysep+\arraysep);
            \foreach \j in {2, ..., \arraylen} {
                \draw (\j*\counterw-\counterw,-\i*\arrayh+1.5*\arrayh-\i*\arraysep+\arraysep) -- (\j*\counterw-\counterw,-\i*\arrayh+0.5*\arrayh-\i*\arraysep+\arraysep);
            }
        }
        \ifappendix
            \def\texty{4.5*\arrayh}
            \def\inserttextx{-2.25*\counterw}
        \else
            \def\texty{4.0*\arrayh}
            \def\inserttextx{-3*\counterw}
        \fi
        \node[inner sep=1pt,align=center] (key) at (\inserttextx,\texty) {\small Insert $x$};
        \draw (key.south) -- (\inserttextx,-1.5*\arrayh-1.5*\arraysep);

        \def\pa{3.0}
        \def\pb{1.0}
        \def\pc{2.0}
        \ifappendix
            \draw[-stealth] (\inserttextx,0.5*\arrayh+0.5*\arraysep) -| (\pa*\counterw+0.5*\counterw,0.5*\arrayh) node[pos=0.095,above=2pt,align=center,inner sep=0pt] {\footnotesize $h_1(x)$, $+1$};
            \draw[-stealth] (\inserttextx,-0.5*\arrayh-0.5*\arraysep) -| (\pb*\counterw+0.5*\counterw,-0.5*\arrayh-1.0*\arraysep) node[pos=0.148,above=2pt,align=center,inner sep=0pt] {\footnotesize $h_2(x)$, $+1$};
            \draw[-stealth] (\inserttextx,-1.5*\arrayh-1.5*\arraysep) -| (\pc*\counterw+0.5*\counterw,-1.5*\arrayh-2.0*\arraysep) node[pos=0.116,above=2pt,align=center,inner sep=0pt] {\footnotesize $h_3(x)$, $+1$};
        \else
            \draw[-stealth] (\inserttextx,0.5*\arrayh+0.5*\arraysep) -| (\pa*\counterw+0.5*\counterw,0.5*\arrayh) node[pos=0.1005,above=2pt,align=center,inner sep=0pt] {\small $h_1(x)$, $+1$};
            \draw[-stealth] (\inserttextx,-0.5*\arrayh-0.5*\arraysep) -| (\pb*\counterw+0.5*\counterw,-0.5*\arrayh-1.0*\arraysep) node[pos=0.145,above=2pt,align=center,inner sep=0pt] {\small $h_2(x)$, $+1$};
            \draw[-stealth] (\inserttextx,-1.5*\arrayh-1.5*\arraysep) -| (\pc*\counterw+0.5*\counterw,-1.5*\arrayh-2.0*\arraysep) node[pos=0.1185,above=2pt,align=center,inner sep=0pt] {\small $h_3(x)$, $+1$};
        \fi

        \ifappendix
            \def\querytextx{\arrayw+2.25*\counterw}
        \else
            \def\querytextx{\arrayw+3*\counterw}
        \fi
        \node[inner sep=1pt,align=center] (query) at (\querytextx,\texty) {\small Query $q$ \\[-2pt] \small Result $=1$};
        \draw[stealth-] (query.south) -- (\querytextx,-1.5*\arrayh-1.5*\arraysep) node[pos=0.5,below=2pt,rotate=90,inner sep=2pt] {\small Min};

        \def\pa{8.0}
        \def\pb{4.0}
        \def\pc{7.0}
        \ifappendix
            \draw (\querytextx,0.5*\arrayh+0.5*\arraysep) -| (\pa*\counterw+0.5*\counterw,0.5*\arrayh) node[pos=0.095,above=2pt,align=center,inner sep=0pt] {\footnotesize $h_1(q)$};
            \draw (\querytextx,-0.5*\arrayh-0.5*\arraysep) -| (\pb*\counterw+0.5*\counterw,-0.5*\arrayh-1.0*\arraysep) node[pos=0.0445,above=2pt,align=center,inner sep=0pt] {\footnotesize $h_2(q)$};
            \draw (\querytextx,-1.5*\arrayh-1.5*\arraysep) -| (\pc*\counterw+0.5*\counterw,-1.5*\arrayh-2.0*\arraysep) node[pos=0.074,above=2pt,align=center,inner sep=0pt] {\footnotesize $h_3(q)$};
        \else
            \draw (\querytextx,0.5*\arrayh+0.5*\arraysep) -| (\pa*\counterw+0.5*\counterw,0.5*\arrayh) node[pos=0.095,above=2pt,align=center,inner sep=0pt] {\small $h_1(q)$};
            \draw (\querytextx,-0.5*\arrayh-0.5*\arraysep) -| (\pb*\counterw+0.5*\counterw,-0.5*\arrayh-1.0*\arraysep) node[pos=0.0502,above=2pt,align=center,inner sep=0pt] {\small $h_2(q)$};
            \draw (\querytextx,-1.5*\arrayh-1.5*\arraysep) -| (\pc*\counterw+0.5*\counterw,-1.5*\arrayh-2.0*\arraysep) node[pos=0.0773,above=2pt,align=center,inner sep=0pt] {\small $h_3(q)$};
        \fi

        \draw[fill=gray!20] (\pa*\counterw,0.5*\arrayh) rectangle (\pa*\counterw+\counterw,-0.5*\arrayh);
        \draw[fill=gray!20] (\pb*\counterw,-0.5*\arrayh-1.0*\arraysep) rectangle (\pb*\counterw+\counterw,-1.5*\arrayh-1.0*\arraysep);
        \draw[fill=gray!20] (\pc*\counterw,-1.5*\arrayh-2.0*\arraysep) rectangle (\pc*\counterw+\counterw,-2.5*\arrayh-2.0*\arraysep);
        \ifappendix
            \node[inner sep=0pt] at (\pa*\counterw+0.5*\counterw,-0.0*\arrayh-0.0*\arraysep) {\footnotesize $1$};
            \node[inner sep=0pt] at (\pb*\counterw+0.5*\counterw,-1.0*\arrayh-1.0*\arraysep) {\footnotesize $4$};
            \node[inner sep=0pt] at (\pc*\counterw+0.5*\counterw,-2.0*\arrayh-2.0*\arraysep) {\footnotesize $3$};
        \else 
            \node[inner sep=0pt] at (\pa*\counterw+0.5*\counterw,-0.0*\arrayh-0.0*\arraysep) {\small $1$};
            \node[inner sep=0pt] at (\pb*\counterw+0.5*\counterw,-1.0*\arrayh-1.0*\arraysep) {\small $4$};
            \node[inner sep=0pt] at (\pc*\counterw+0.5*\counterw,-2.0*\arrayh-2.0*\arraysep) {\small $3$};
        \fi

        %\node[inner sep=0pt,anchor=west] at (-3.05*\counterw,3.0*\arrayh) {\textbf{A) CMS}};
    \end{tikzpicture}
    \ifnoappendix
      \vspace{-3mm}
    \fi
    \caption{CMS maintains~$d$ counter arrays of size~$w$ and hashes keys into
    them to estimate their frequencies. Insertions are illustrated on the left
    and queries on the right.}
    \ifnoappendix
      \vspace{-5mm}
    \fi
    \label{fig:cms_overview}
\end{figure}

Yet, with a single array, confidence is low. The reason is that all keys
hashing to the same counter as a frequent key will suffer from inflated
estimates. Moreover, since there are many keys with low or moderate counts,
there is a non-negligible chance that enough of them map to the same counter
and inflate estimates beyond the error bound. CMS addresses these issues by
allocating~$d$ independent counter arrays, each with a dedicated hash
function~$h_i$. CMS inserts or deletes a key by applying the operation to each
array, and it answers a query by querying all arrays and taking the minimum of
the estimates.\footnote{CMS uses pairwise-independent hash functions to bound
the probability of collisions and overestimation~\cite{CountMin}.} Returning
the minimum for queries boosts confidence exponentially with respect to the
number of arrays~$d$, since the minimum estimate only carries severe errors
when all estimates err significantly. Formally, with $d$ arrays of length $w$,
CMS guarantees an error of at most~$e \cdot N/w$ with a confidence
of~$1-e^{-d}$, i.e., $\Pr(|\hat{f}(x)-f(x)| \leq e \cdot N/w) \geq 1-e^{-d}$
for any key~$x$. \Cref{fig:cms_overview}
shows an example of insertions and queries applied to three~counter~arrays.

Since CMS uses a fixed number of tightly packed counters and does not store the
stream's keys, it significantly improves the memory footprint compared to the
exhaustive approach of tracking each key's exact count. Nevertheless, CMS still
suffers from several memory inefficiencies, which we explore in the following
section.

\section{Problem Analysis}\label{sec:problem_analysis}
We show that mainstream FE sketches such as CMS suffer from two core
limitations: 1)~they waste a significant amount of memory under
skew{\textemdash}memory that could otherwise be used to improve accuracy or
confidence{\textemdash}and 2)~their estimation error grows rapidly with the
stream length. These shortcomings prevent current FE sketches from achieving
accuracy at scale.

\subsection{Challenge 1: Skew}\label{sec:problem_analysis_skew}
Most real-world workloads are skewed, meaning that a few heavy hitters receives
the majority of insertions, while all other keys have low
counts~\cite{Memento,AdditiveErrorCounters,StingySketch}. In CMS, however, all
counters are uniformly sized, so each must be provisioned to accommodate the
largest possible value. This leads to substantial memory waste. Ideally, each
counter should be sized just large enough to fit its value, thereby saving
space or enabling more counters to improve accuracy or confidence. This issue
has been explored in recent work; we now review existing solutions and their
limitations.

\textbf{Hybrid Methods.}
A common strategy for tackling skew is to track the heavy
hitters separately in a hash table while leaving the rest of the keys to
CMS~\cite{AugmentedSketch,HeavyKeeper,HeavyGuardian,MVSketch,WavingSketch,ElasticSketch}.
Such hybrid structures aim to cap the maximum counter value in CMS, thereby
permitting shorter counters. Yet, this approach does not fundamentally resolve
the issue. The reasons are twofold: The keys tracked by CMS may still follow a
skewed distribution, leading to significant memory wastage. Moreover, the heavy
hitters can vary over time. Identifying the next heavy hitter and ``promoting''
it from CMS into the hash table is error-prone, since collisions in CMS can
cause many infrequent keys to be misclassified as heavy hitters.

\textbf{Compression-based Methods.}
Other approaches compress the counters using algorithms such as compressive
sensing~\cite{TreeSensing,BitSense} or hypergraph peeling~\cite{CodingSketch}.
Yet, these methods impose heavy decompression overheads on queries.

\begin{figure}
    \centering
    \begin{tikzpicture}
        \ifappendix
            \def\arrayw{3.5}
            \def\arrayh{0.25}
        \else 
            \def\arrayw{5.25}
            \def\arrayh{0.35}
        \fi
        \def\methodsep{1.2}
        \def\offdots{0.25}
        \def\countercount{4}
        \def\counterw{\arrayw/\countercount}

        % Counter Sharing
        \def\sharingarraywmult{0.8}
        \ifappendix
            \def\treelevelsep{0.2}
        \else
            \def\treelevelsep{0.15}
        \fi
        \def\fragmentw{\sharingarraywmult*\counterw/2}
        \def\offfragment{\fragmentw/3}
        \def\currentarrayw{\sharingarraywmult*\arrayw}
        \def\currentcounterw{\sharingarraywmult*\counterw}
        \draw (0,0) rectangle (\currentarrayw,\arrayh);
        \foreach \i in {2, ..., \countercount} {
            \draw (\i*\currentcounterw-\currentcounterw,0) -- (\i*\currentcounterw-\currentcounterw,\arrayh);
        }
        \ifappendix
            \node[inner sep=0pt,anchor=east] (shared_label) at (0.5*\currentcounterw,\treelevelsep+1.5*\arrayh) {\footnotesize Shared};
        \else
            \node[inner sep=0pt,anchor=east] (shared_label) at (0.65*\currentcounterw,\treelevelsep+1.5*\arrayh) {\small Shared};
        \fi
        \foreach \i in {1, ..., 2} {
            \def\countermidx{2*\i*\currentcounterw-\currentcounterw}
            \ifnum \i = 2 {
                \draw[fill=gray!20] (\countermidx-\fragmentw/2,\treelevelsep+\arrayh) rectangle (\countermidx+\fragmentw/2,\treelevelsep+2*\arrayh);
                \draw (\countermidx-\offfragment,\treelevelsep+\arrayh) -- (\countermidx-0.5*\currentcounterw,\arrayh);
                \draw (\countermidx+\offfragment,\treelevelsep+\arrayh) -- (\countermidx+0.5*\currentcounterw,\arrayh);
            }
            \else {
                \draw (\countermidx-\fragmentw/2,\treelevelsep+\arrayh) rectangle (\countermidx+\fragmentw/2,\treelevelsep+2*\arrayh);
                \draw[stealth-] (\countermidx-\offfragment,\treelevelsep+\arrayh) -- (\countermidx-0.5*\currentcounterw,\arrayh);
                \draw[stealth-] (\countermidx+\offfragment,\treelevelsep+\arrayh) -- (\countermidx+0.5*\currentcounterw,\arrayh);
            }
            \fi
        }
        \def\countermidx{\currentarrayw/2}
        \draw (\countermidx-\fragmentw/2,2*\treelevelsep+2*\arrayh) rectangle (\countermidx+\fragmentw/2,2*\treelevelsep+3*\arrayh);
        \draw[stealth-] (\countermidx-\offfragment,2*\treelevelsep+2*\arrayh) -- (\currentarrayw/4,\treelevelsep+2*\arrayh);
        \draw (\countermidx+\offfragment,2*\treelevelsep+2*\arrayh) -- (\currentarrayw/4*3,\treelevelsep+2*\arrayh);

        \node[inner sep=0pt] at (-\offdots,\arrayh/2) {\small $\dots$};
        \node[inner sep=0pt] at (\currentarrayw+\offdots,\arrayh/2) {\small $\dots$};
        \def\counterarray{{0,"001101","001000","001000","111000"}}
        \foreach \i in {1, ..., \countercount} {
            \def\countermidx{\i*\currentcounterw-\currentcounterw/2}
            \pgfmathparse{\counterarray[\i]}
            \let\countervalue\pgfmathresult
            \ifappendix
                \node[inner sep=0pt] at (\countermidx,\arrayh/2) {\scriptsize \countervalue};
            \else
                \node[inner sep=0pt] at (\countermidx,\arrayh/2) {\small \countervalue};
            \fi
        }
        \def\counterarray{{0,"00",0}}
        \foreach \i in {1, ..., 2} {
            \def\countermidx{2*\i*\currentcounterw-\currentcounterw}
            \pgfmathparse{\counterarray[\i]}
            \let\countervalue\pgfmathresult
            \ifnum \i = 2 {
            }
            \else {
              \ifappendix
                  \node[inner sep=0pt] at (\countermidx,\treelevelsep+1.5*\arrayh) {\scriptsize \countervalue};
              \else
                  \node[inner sep=0pt] at (\countermidx,\treelevelsep+1.5*\arrayh) {\small \countervalue};
              \fi
            }
            \fi
        }
        \ifappendix
            \node[inner sep=0pt] at (\currentarrayw/2,2*\treelevelsep+2.5*\arrayh) {\scriptsize 10};
        \else
            \node[inner sep=0pt] at (\currentarrayw/2,2*\treelevelsep+2.5*\arrayh) {\small 10};
        \fi

        \def\labely{3*\treelevelsep+3.25*\arrayh}
        \node[inner sep=0pt] (sharing_label) at (\currentarrayw/2,\labely) {\small A) Counter Sharing};

        % Counter Merging
        \def\mergingoffx{\sharingarraywmult*\arrayw+\methodsep}
        \def\mergingarraywmult{1}
        \def\currentarrayw{\mergingarraywmult*\arrayw}
        \def\currentcounterw{\mergingarraywmult*\counterw}
        \draw (\mergingoffx,0) rectangle (\mergingoffx+\currentarrayw,\arrayh);
        \foreach \i in {2, ..., \countercount} {
            \ifnum \i = 3 {
                \draw (\mergingoffx+\i*\currentcounterw-\currentcounterw,0) -- (\mergingoffx+\i*\currentcounterw-\currentcounterw,\arrayh);
            }
            \else \ifnum \i = 4 {
                \draw (\mergingoffx+\i*\currentcounterw-\currentcounterw,0) -- (\mergingoffx+\i*\currentcounterw-\currentcounterw,\arrayh);
            }
            \fi \fi
        }
        \ifappendix
            \draw[decorate,decoration={brace,raise=1pt,amplitude=2pt}] (\mergingoffx,\arrayh) -- (\mergingoffx+2*\currentcounterw,\arrayh)
                node[pos=0.5,above=4pt,inner sep=1pt] {\footnotesize Merged};
        \else
            \draw[decorate,decoration={brace,raise=1pt,amplitude=2pt}] (\mergingoffx,\arrayh) -- (\mergingoffx+2*\currentcounterw,\arrayh)
                node[pos=0.5,above=4pt,inner sep=1pt] {\small Merged};
        \fi
        \node[inner sep=0pt] at (\mergingoffx-\offdots,\arrayh/2) {\small $\dots$};
        \node[inner sep=0pt] at (\mergingoffx+\currentarrayw+\offdots,\arrayh/2) {\small $\dots$};
        \ifappendix
            \node[inner sep=0pt] at (\mergingoffx+\currentcounterw,\arrayh/2) {\scriptsize 0000110010000000};
            \node[inner sep=0pt] at (\mergingoffx+2.5*\currentcounterw,\arrayh/2) {\scriptsize 00100000};
            \node[inner sep=0pt] at (\mergingoffx+3.5*\currentcounterw,\arrayh/2) {\scriptsize 11100000};
        \else
            \node[inner sep=0pt] at (\mergingoffx+\currentcounterw,\arrayh/2) {\small 0000110010000000};
            \node[inner sep=0pt] at (\mergingoffx+2.5*\currentcounterw,\arrayh/2) {\small 00100000};
            \node[inner sep=0pt] at (\mergingoffx+3.5*\currentcounterw,\arrayh/2) {\small 11100000};
        \fi

        \node[inner sep=0pt] (merging_label) at (\mergingoffx+\currentarrayw/2,\labely) {\small B) Counter Merging};

        % CMS
        \def\cmsmidx{\sharingarraywmult*\arrayw/2+\methodsep/2+\currentarrayw/2}
        \def\cmsarraywmult{1.5}
        \def\cmsoffx{\cmsmidx-\cmsarraywmult*\arrayw/2}
        \def\cmsoffy{-0.75}
        \draw (\cmsoffx,\cmsoffy) rectangle (\cmsoffx+\cmsarraywmult*\arrayw,\cmsoffy+\arrayh);
        \foreach \i in {2, ..., \countercount} {
            \draw (\cmsoffx+\cmsarraywmult*\i*\counterw-\cmsarraywmult*\counterw,\cmsoffy) -- (\cmsoffx+\cmsarraywmult*\i*\counterw-\cmsarraywmult*\counterw,\cmsoffy+\arrayh);
        }
        \node[inner sep=0pt] at (\cmsoffx-\offdots,\cmsoffy+\arrayh/2) {\small $\dots$};
        \node[inner sep=0pt] at (\cmsoffx+\cmsarraywmult*\arrayw+\offdots,\cmsoffy+\arrayh/2) {\small $\dots$};
        \def\counterarray{{0,"001101001000","001000000000","001000000000","111000000000"}}
        \foreach \i in {1, ..., \countercount} {
            \pgfmathparse{\counterarray[\i]}
            \let\countervalue\pgfmathresult
            \ifappendix
                \node[inner sep=0pt] at (\cmsoffx+\cmsarraywmult*\i*\counterw-\cmsarraywmult*\counterw/2,\cmsoffy+\arrayh/2) {\scriptsize \countervalue};
            \else
                \node[inner sep=0pt] at (\cmsoffx+\cmsarraywmult*\i*\counterw-\cmsarraywmult*\counterw/2,\cmsoffy+\arrayh/2) {\small \countervalue};
            \fi
        }

        \def\labely{\cmsoffy-0.25}
        \node[inner sep=0pt] (cms_label) at (\cmsoffx+\cmsarraywmult*\arrayw/2,\labely) {\small CMS};
        \ifappendix
            \draw[-stealth] (\cmsoffx+\cmsarraywmult*\arrayw/4,\cmsoffy+1.5*\arrayh) -- (\sharingarraywmult*\arrayw/2,-0.5*\arrayh);
            \draw[-stealth] (\cmsoffx+\cmsarraywmult*\arrayw/4*2.72,\cmsoffy+1.5*\arrayh) -- (\mergingoffx+\mergingarraywmult*\arrayw/2,-0.5*\arrayh);
        \else
            \draw[-stealth] (\cmsoffx+\cmsarraywmult*\arrayw/4,\cmsoffy+1.3*\arrayh) -- (\sharingarraywmult*\arrayw/2,-0.3*\arrayh);
            \draw[-stealth] (\cmsoffx+\cmsarraywmult*\arrayw/4*2.72,\cmsoffy+1.3*\arrayh) -- (\mergingoffx+\mergingarraywmult*\arrayw/2,-0.3*\arrayh);
        \fi
    \end{tikzpicture}
    \ifappendix
        \vspace{-1.55mm}
    \fi
    \ifnoappendix
      \vspace{-3mm}
    \fi
    \caption{Counter sharing and counter merging encode CMS's counters in less
    space in exchange for blowing up some counter values. Here, the bits in
    each counter are in increasing order of significance from left to right.}
    \ifnoappendix
      \vspace{-8mm}
    \fi
    \label{fig:counter_sharing_merging}
\end{figure}

\textbf{Counter Sharing.}
Counter sharing approaches store the counter values in a hierarchy of
counters~\cite{CounterTree}. The bottom level employs short counters to store
the lower-order bits of each counter, while the upper levels handle overflows
using shared counters. For example,
\Cref{fig:counter_sharing_merging}{\nobreakdash-}A)
encodes four adjacent counters within CMS using 6{\nobreakdash-}bit counters in
the bottom level and 2{\nobreakdash-}bit counters in the upper levels. Whenever
a counter in this hierarchy overflows, it increments its parent and resets to
zero. A key's count is retrieved by concatenating its corresponding
bottom-level counter with the shared counters along its path to the root or
until reaching an internal node to which no overflow has occurred. Such an
internal node is identified using additional metadata. The core issue with this
approach is that the error rapidly deteriorates as the stream grows. In the
worst-case scenario, after~$N$ insertions into an array of size~$w$, each
counter overflows into~$\Theta(\log (N/w))$ levels above it. Since the counters
in the higher levels are shared, this blows up the values of~$\Theta(N/w)$
other counters, increasing the estimation error by the same factor. To mitigate
this issue, prior works experiment with varying the counter sizes and the
hierarchy's
fanout~\cite{BitSense,PyramidSketch,StingySketch,AdaptiveCounterSplicing,TreeSensing,TailoredSketch}.
Yet, they do not fundamentally resolve the problem.

\textbf{Counter Merging.}
Counter merging is an alternative approach that starts with 8-bit counters and
merges an overflowing counter with an adjacent counter into a longer counter
that sums their values~\cite{SALSA}.
\Cref{fig:counter_sharing_merging}{\nobreakdash-}B)
shows an example applied to CMS. After~$N$ insertions into an array of~$w$
counters, each counter will be merged with~$\Theta(\log (N/w))$ other counters,
increasing the error by the same factor. Merging only continues until each
counter becomes four bytes long, at which point error degradation stops.
Counter merging has a higher metadata overhead than counter sharing for
indicating which counters are merged, leaving less space for storing counters.
Parsing this metadata adds more complexity than alternative methods like
counter sharing,~reducing~performance.

\textbf{Deletions.}
Deletions deteriorate the accuracy of both counter sharing and counter merging
methods since these techniques ambiguate how counters should be ``separated''
as they decrease.

\ifnoappendix
\begin{wrapfigure}{r}{0.5\textwidth}
    \begin{tikzpicture}
        \def\w{2}
        \def\ox{-2}
        \def\oy{-2}
        \def\xaxislength{3.0}
        \def\yaxislength{2.00}
        \def\ooverlen{0.1}
        \def\xkk{0.9}
        \def\pkk{0.9}
        \def\labelx{\ox+\xaxislength+0.5}
        \def\labely{\oy+\pkk*\yaxislength-0.1}
        \def\labellinex{\labelx+2.4}
        \def\labelliney{\labely+0.22}
        \def\labellineydiffcountersharing{0.92}
        \def\labellineydiffcountermerging{2.04*\labellineydiffcountersharing}
        \def\labellinelen{0.5}
        \def\haveaxislabels{0}

        % Axes
        \draw[thick,-stealth] (\ox-\ooverlen,\oy) -- (\ox+\xaxislength,\oy);
        \draw[thick,-stealth] (\ox,\oy-\ooverlen) -- (\ox,\oy+\yaxislength);
        \node[inner sep=0pt,align=center] (x_label) at (\ox+\xaxislength/2,\oy-\haveaxislabels*0.9-0.3) {\small Total Key Count~$N$};
        \node[inner sep=0pt,rotate=90] (y_label) at (\ox-\haveaxislabels*0.7-0.3,\oy+\yaxislength/2) {\small Error};

        % CMS
        \node[inner sep=1pt,circle,fill=black] (start) at (\ox,\oy) {};
        \node[inner sep=1pt,circle,fill=black] (cms_end) at (\ox+\xkk*\xaxislength,\oy+\pkk*\yaxislength/\w) {};
        \draw (start) -- (cms_end);
        \node[inner sep=0pt,anchor=west,align=left] (cms_label) at (\labelx,\labely) {\footnotesize CMS \\[-1pt] \footnotesize Error~$={\scriptstyle \frac{N}{w}}$};

        \node[inner sep=1pt,circle,fill=black] (cms_legend_start) at (\labellinex,\labelliney) {};
        \node[inner sep=1pt,circle,fill=black] (cms_legend_end) at (\labellinex+\labellinelen,\labelliney) {};
        \draw (cms_legend_start) -- (cms_legend_end);

        % Counter Sharing
        \def\countersharingintersectionx{\ox+0.25*\xaxislength}
        \def\countersharingintersectiony{\oy+0.123*\yaxislength}
        \ifnum \haveaxislabels = 1 {
            \node[inner sep=1pt,circle,fill=black] (counter_sharing_intersection) at (\countersharingintersectionx,\countersharingintersectiony) {};
        }
        \else {
            \node[inner sep=0pt] (counter_sharing_intersection) at (\countersharingintersectionx,\countersharingintersectiony) {};
        }
        \fi

        \node[inner sep=1pt,circle,fill=black] (counter_sharing_end) at (\ox+\xkk*\xaxislength,\oy+\pkk*\yaxislength) {};
        \draw[dashed] plot[hobby] coordinates { (start) (\ox+0.13*\xaxislength,\oy+0.025*\yaxislength) (counter_sharing_intersection) (counter_sharing_end) };
        \node[inner sep=0pt,below=0.60 of cms_label.south west,anchor=west,align=left] (counter_sharing_label) {\footnotesize Counter Sharing \\[-1pt] \footnotesize Error~$={\scriptstyle \Theta\left( \min\left( \frac{N}{w}, w \right) \cdot \frac{N}{w} \right)}$};

        \node[inner sep=1pt,circle,fill=black] (counter_sharing_legend_start) at (\labellinex,\labelliney-\labellineydiffcountersharing) {};
        \node[inner sep=1pt,circle,fill=black] (counter_sharing_legend_end) at (\labellinex+\labellinelen,\labelliney-\labellineydiffcountersharing) {};
        \draw[dashed] (counter_sharing_legend_start) -- (counter_sharing_legend_end);

        % Counter Merging
        \def\countermergingintersectionx{\ox+0.612*\xaxislength}
        \def\countermergingintersectiony{\oy+0.305*\yaxislength}
        \ifnum \haveaxislabels = 1 {
            \node[inner sep=1pt,circle,fill=black] (counter_merging_intersection) at (\countermergingintersectionx,\countermergingintersectiony) {};
        }
        \else {
            \node[inner sep=0pt] (counter_merging_intersection) at (\countermergingintersectionx,\countermergingintersectiony) {};
        }
        \fi
        \def\longdashlen{3pt}
        \def\shortdashlen{0.5pt}
        \def\dashsep{1pt}
        \draw[dash pattern={on \longdashlen off \dashsep on \shortdashlen off \dashsep}] plot[hobby] coordinates { (start) (\ox+0.15*\xaxislength,\oy+0.050*\yaxislength) (counter_merging_intersection) };
        \node[inner sep=1pt,circle,fill=black] (counter_merging_end) at (\ox+\xkk*\xaxislength,\oy+\pkk*\yaxislength/\w*1.1) {};
        \draw[dash pattern={on \longdashlen off \dashsep on \shortdashlen off \dashsep}] (counter_merging_intersection) -- (counter_merging_end);
        \node[inner sep=0pt,below=0.60 of counter_sharing_label.south west,anchor=west,align=left] (counter_merging_label) {\footnotesize Counter Merging \\[-1pt] \footnotesize Error~$={\scriptstyle \Theta\left( \min\left( \log \frac{N}{w}, 32 \right) \cdot \frac{N}{w} \right)}$};

        \node[inner sep=1pt,circle,fill=black] (counter_merging_start) at (\labellinex,\labelliney-\labellineydiffcountermerging) {};
        \node[inner sep=1pt,circle,fill=black] (counter_merging_end) at (\labellinex+\labellinelen,\labelliney-\labellineydiffcountermerging) {};
        \draw[dash pattern={on \longdashlen off \dashsep on \shortdashlen off \dashsep}] (counter_merging_start) -- (counter_merging_end);

        % Focal Points and lines
        \ifnum \haveaxislabels = 1 {
            \draw[densely dotted] (\ox,\oy+\pkk*\yaxislength/\w) -- (cms_end);
            \node[inner sep=2pt,anchor=east] (y_cms_end_label) at (\ox,\oy+\pkk*\yaxislength/\w) {\footnotesize $\scriptstyle 2^{32}$};

            \draw[densely dotted] (\ox+\pkk*\xaxislength,\oy) -- (counter_sharing_end);
            \node[inner sep=3pt,anchor=north] (x_end_label) at (\ox+\pkk*\xaxislength,\oy) {\footnotesize $\scriptstyle w \cdot 2^{32}$};

            \draw[densely dotted] (\ox,\oy+\pkk*\yaxislength) -- (counter_sharing_end);
            \node[inner sep=2pt,anchor=east] (y_counter_sharing_label) at (\ox,\oy+\pkk*\yaxislength) {\footnotesize $\scriptstyle \approx w \cdot 2^{32}$};

            \draw[densely dotted] (\countersharingintersectionx,\oy) -- (counter_sharing_intersection);
            \node[inner sep=4.25pt,anchor=north] (x_counter_sharing_intersection_label) at (\countersharingintersectionx,\oy) {\footnotesize $\scriptstyle \Theta(w)$};

            \draw[densely dotted] (\ox,\countersharingintersectiony) -- (counter_sharing_intersection);
            \node[inner sep=2pt,anchor=east] (y_counter_sharing_intersection_label) at (\ox,\countersharingintersectiony) {\footnotesize $\scriptstyle \Theta(1)$};

            \draw[densely dotted] (\countermergingintersectionx,\oy) -- (counter_merging_intersection);
            \node[inner sep=3pt,anchor=north] (x_counter_merging_intersection_label) at (\countermergingintersectionx,\oy) {\footnotesize $\scriptstyle \approx w \cdot 2^{24}$};

            \draw[densely dotted] (\ox,\countermergingintersectiony) -- (counter_merging_intersection);
            \node[inner sep=2pt,anchor=east] (y_counter_merging_intersection_label) at (\ox,\countermergingintersectiony) {\footnotesize $\scriptstyle \approx 2^{24}$};
        }
        \fi
    \end{tikzpicture}
    \vspace{-8mm}
    \caption{Under the same memory budget as a CMS instance with~32-bit
    counters, counter sharing and counter merging can lead to lower accuracy
    when processing a growing stream.}
    \label{fig:counter_sharing_deterioration}
\end{wrapfigure}
\fi

\textbf{Comparison.}
Prior work has shown counter sharing and counter merging to dominate hybrid and
compression-based methods in terms of both accuracy and
performance~\cite{SALSA,TailoredSketch}. Thus, in
\Cref{fig:counter_sharing_deterioration},
we qualitatively compare counter sharing and counter merging to a CMS with the
same memory footprint to illustrate their error degradation with data growth.
While these methods perform well for short streams, their error exceeds that of
CMS as the stream grows. We empirically verify this behavior in
Experiment~\hyperlink{experiment:expansion}{4} of \Cref{sec:evaluation}.

\subsection{Challenge 2: Unbounded Growth}\label{sec:problem_analysis_unbounded_growth}
In many streaming applications, the data stream grows continuously, and it is
imperative to estimate the count of each key from the stream's beginning. For
instance, query optimizers often track the number of times each key occurs in a
table column to estimate the size of a join on that
column~\cite{Compass,QueryOptimizationSketches,ConvolutionSketch}. As new data
is inserted, the total number of keys~$N$ increases. Yet, the full count
history of each key is still needed, since every occurrence contributes equally
to the join size regardless of when it arrived.

\ifappendix
\begin{figure}
    \centering
    \begin{tikzpicture}
        \def\w{2}
        \def\ox{-2}
        \def\oy{-2}
        \def\xaxislength{3.2}
        \def\yaxislength{2.05}
        \def\ooverlen{0.1}
        \def\xkk{0.9}
        \def\pkk{0.9}
        \def\labelx{\ox+\xaxislength+1.00}
        \def\labely{\oy+\pkk*\yaxislength-0.2}
        \def\labellinex{\labelx+2.15}
        \def\labelliney{\labely+0.22}
        \def\labellineydiffcountersharing{0.89}
        \def\labellineydiffcountermerging{2.04*\labellineydiffcountersharing}
        \def\labellinelen{0.5}
        \def\haveaxislabels{0}

        % Axes
        \draw[thick,-stealth] (\ox-\ooverlen,\oy) -- (\ox+\xaxislength,\oy);
        \draw[thick,-stealth] (\ox,\oy-\ooverlen) -- (\ox,\oy+\yaxislength);
        \node[inner sep=0pt,align=center] (x_label) at (\ox+\xaxislength/2,\oy-\haveaxislabels*0.9-0.3) {\small Total Key Count~$N$};
        \node[inner sep=0pt,rotate=90] (y_label) at (\ox-\haveaxislabels*0.7-0.3,\oy+\yaxislength/2) {\small Error};

        % CMS
        \node[inner sep=1pt,circle,fill=black] (start) at (\ox,\oy) {};
        \node[inner sep=1pt,circle,fill=black] (cms_end) at (\ox+\xkk*\xaxislength,\oy+\pkk*\yaxislength/\w) {};
        \draw (start) -- (cms_end);
        \node[inner sep=0pt,anchor=west,align=left] (cms_label) at (\labelx,\labely) {\footnotesize CMS \\[-1pt] \footnotesize Error~$={\scriptstyle \frac{N}{w}}$};

        \node[inner sep=1pt,circle,fill=black] (cms_legend_start) at (\labellinex,\labelliney) {};
        \node[inner sep=1pt,circle,fill=black] (cms_legend_end) at (\labellinex+\labellinelen,\labelliney) {};
        \draw (cms_legend_start) -- (cms_legend_end);

        % Counter Sharing
        \def\countersharingintersectionx{\ox+0.25*\xaxislength}
        \def\countersharingintersectiony{\oy+0.123*\yaxislength}
        \ifnum \haveaxislabels = 1 {
            \node[inner sep=1pt,circle,fill=black] (counter_sharing_intersection) at (\countersharingintersectionx,\countersharingintersectiony) {};
        }
        \else {
            \node[inner sep=0pt] (counter_sharing_intersection) at (\countersharingintersectionx,\countersharingintersectiony) {};
        }
        \fi

        \node[inner sep=1pt,circle,fill=black] (counter_sharing_end) at (\ox+\xkk*\xaxislength,\oy+\pkk*\yaxislength) {};
        \draw[dashed] plot[hobby] coordinates { (start) (\ox+0.13*\xaxislength,\oy+0.025*\yaxislength) (counter_sharing_intersection) (counter_sharing_end) };
        \node[inner sep=0pt,below=0.60 of cms_label.south west,anchor=west,align=left] (counter_sharing_label) {\footnotesize Counter Sharing \\[-1pt] \footnotesize Error~$={\scriptstyle \Theta\left( \min\left( \frac{N}{w}, w \right) \cdot \frac{N}{w} \right)}$};

        \node[inner sep=1pt,circle,fill=black] (counter_sharing_legend_start) at (\labellinex,\labelliney-\labellineydiffcountersharing) {};
        \node[inner sep=1pt,circle,fill=black] (counter_sharing_legend_end) at (\labellinex+\labellinelen,\labelliney-\labellineydiffcountersharing) {};
        \draw[dashed] (counter_sharing_legend_start) -- (counter_sharing_legend_end);

        % Counter Merging
        \def\countermergingintersectionx{\ox+0.612*\xaxislength}
        \def\countermergingintersectiony{\oy+0.305*\yaxislength}
        \ifnum \haveaxislabels = 1 {
            \node[inner sep=1pt,circle,fill=black] (counter_merging_intersection) at (\countermergingintersectionx,\countermergingintersectiony) {};
        }
        \else {
            \node[inner sep=0pt] (counter_merging_intersection) at (\countermergingintersectionx,\countermergingintersectiony) {};
        }
        \fi
        \def\longdashlen{3pt}
        \def\shortdashlen{0.5pt}
        \def\dashsep{1pt}
        \draw[dash pattern={on \longdashlen off \dashsep on \shortdashlen off \dashsep}] plot[hobby] coordinates { (start) (\ox+0.15*\xaxislength,\oy+0.050*\yaxislength) (counter_merging_intersection) };
        \node[inner sep=1pt,circle,fill=black] (counter_merging_end) at (\ox+\xkk*\xaxislength,\oy+\pkk*\yaxislength/\w*1.1) {};
        \draw[dash pattern={on \longdashlen off \dashsep on \shortdashlen off \dashsep}] (counter_merging_intersection) -- (counter_merging_end);
        \node[inner sep=0pt,below=0.60 of counter_sharing_label.south west,anchor=west,align=left] (counter_merging_label) {\footnotesize Counter Merging \\[-1pt] \footnotesize Error~$={\scriptstyle \Theta\left( \min\left( \log \frac{N}{w}, 32 \right) \cdot \frac{N}{w} \right)}$};

        \node[inner sep=1pt,circle,fill=black] (counter_merging_start) at (\labellinex,\labelliney-\labellineydiffcountermerging) {};
        \node[inner sep=1pt,circle,fill=black] (counter_merging_end) at (\labellinex+\labellinelen,\labelliney-\labellineydiffcountermerging) {};
        \draw[dash pattern={on \longdashlen off \dashsep on \shortdashlen off \dashsep}] (counter_merging_start) -- (counter_merging_end);

        % Focal Points and lines
        \ifnum \haveaxislabels = 1 {
            \draw[densely dotted] (\ox,\oy+\pkk*\yaxislength/\w) -- (cms_end);
            \node[inner sep=2pt,anchor=east] (y_cms_end_label) at (\ox,\oy+\pkk*\yaxislength/\w) {\footnotesize $\scriptstyle 2^{32}$};

            \draw[densely dotted] (\ox+\pkk*\xaxislength,\oy) -- (counter_sharing_end);
            \node[inner sep=3pt,anchor=north] (x_end_label) at (\ox+\pkk*\xaxislength,\oy) {\footnotesize $\scriptstyle w \cdot 2^{32}$};

            \draw[densely dotted] (\ox,\oy+\pkk*\yaxislength) -- (counter_sharing_end);
            \node[inner sep=2pt,anchor=east] (y_counter_sharing_label) at (\ox,\oy+\pkk*\yaxislength) {\footnotesize $\scriptstyle \approx w \cdot 2^{32}$};

            \draw[densely dotted] (\countersharingintersectionx,\oy) -- (counter_sharing_intersection);
            \node[inner sep=4.25pt,anchor=north] (x_counter_sharing_intersection_label) at (\countersharingintersectionx,\oy) {\footnotesize $\scriptstyle \Theta(w)$};

            \draw[densely dotted] (\ox,\countersharingintersectiony) -- (counter_sharing_intersection);
            \node[inner sep=2pt,anchor=east] (y_counter_sharing_intersection_label) at (\ox,\countersharingintersectiony) {\footnotesize $\scriptstyle \Theta(1)$};

            \draw[densely dotted] (\countermergingintersectionx,\oy) -- (counter_merging_intersection);
            \node[inner sep=3pt,anchor=north] (x_counter_merging_intersection_label) at (\countermergingintersectionx,\oy) {\footnotesize $\scriptstyle \approx w \cdot 2^{24}$};

            \draw[densely dotted] (\ox,\countermergingintersectiony) -- (counter_merging_intersection);
            \node[inner sep=2pt,anchor=east] (y_counter_merging_intersection_label) at (\ox,\countermergingintersectiony) {\footnotesize $\scriptstyle \approx 2^{24}$};
        }
        \fi
    \end{tikzpicture}
    \vspace{3mm}
    \caption{Under the same memory budget as a CMS instance with~32-bit
    counters, counter sharing and counter merging can lead to lower accuracy
    when processing a growing stream.}
    \label{fig:counter_sharing_deterioration}
\end{figure}
\fi

Unbounded stream growth is also a challenge for applications that only require
key counts over recent time windows. For example, network management systems
detect anomalous or malicious activity by counting packets from each source
within fixed time
windows~\cite{SketchBasedChangeDetection,WhatsNew,TrafficManagementSurvey}. For
bursty streams such as network traffic, the number of keys~$N$ in each window
is unpredictable. Bounding~$N$ by defining windows based on insertion count
fails to resolve the issue, as the notion of time is lost.

In existing FE sketches, the number of counters~$w$ is fixed at allocation
time. Consequently, the expected error degrades linearly with~$N$ (e.g., $N/w$
in the case of CMS). In hope of balancing accuracy and memory, practitioners
can tune~$w$ upfront. This tuning is non-trivial, however, as having too few
counters leads to significant errors, while having too many wastes memory.
Deletions make the problem more challenging by causing~$N$ to fluctuate over
time. Worse yet, even if the ultimate stream length is known upfront,
allocating a large CMS from the start wastes memory during the early stages
when the stream is short. 

\textbf{Summary.}
An ideal FE sketch should: 1)~compactly encode counter values without incurring
memory waste from skew, 2)~scale its size with the stream's length to control
both the memory footprint and the error rate,
3)~support deletions without sacrificing accuracy, and 4)~provide
  high-performance queries, insertions, and deletions. Is it possible to
  achieve these goals simultaneously?

\section{\sketch}\label{sec:sublime} 
We present \sketch, a framework for generalizing an FE sketch to optimize its
memory footprint and error rate with respect to the stream's skew and its
length. \Cref{sec:accommodating_skew}
shows how to adapt to skew by compactly encoding each counter as a short
integer and extending it as it overflows.
\Cref{sec:accommodating_unknown_stream_lengths}
shows how to control both accuracy and space for a stream of unknown length by
expanding and contracting the FE sketch. At the same time, \sketch\ supports
deletions and maintains high performance. We first describe \sketch\ as applied
to the Count-Min Sketch~(CMS), denoted as \sketchcms. 
We apply \sketch\ to Count Sketch in \Cref{sec:cs}
and to
\ifappendix
Misra-Gries in
\Cref{sec:mg}.
\else
Misra-Gries~in~the~Appendix~\cite{SublimeArxiv}.
\fi

\subsection{Accommodating Skew}\label{sec:accommodating_skew}
Under skew, many counter values in CMS are small and do not use all the bits in
their counters. To avoid wasting space under any workload, \sketchcms employs a
general-purpose \counterencoding~(\counterencodingabbrv) encoding. Upfront,
\counterencodingabbrv\ stores each counter as a fixed-size integer called a
\emph{Stub}. When a stub overflows, \counterencodingabbrv\ elongates it using a
dedicated variable-length \emph{Extension}. We adaptively tune the stub length
based on the workload's skew to control the rate of stub overflows, thereby
optimizing the overall memory footprint. To improve cache behavior and quickly
reconstruct a counter, \counterencodingabbrv\ colocates a stub and its
corresponding extension in the same cache line. It quickly finds a stub's
extension using a bit manipulation workflow that leverages rank
and~select~operations.

\begin{figure}
    \centering
    \begin{tikzpicture}
        \ifappendix
            \def\arrayw{7.5}
            \def\arrayh{0.3}
        \else
            \def\arrayw{9.0}
            \def\arrayh{0.35}
        \fi
        \def\stubw{0.15*\arrayw}
        \def\stubcountl{2}
        \def\overflowbitmapw{0.15*\arrayw}
        \def\extensionsarrayw{0.3*\arrayw}
        \def\bitw{0.025*\arrayw}
        \def\extensionw{2*\bitw}
        \def\extensioncount{3}
        \def\labely{-2.0*\arrayh}

        \draw[black] (0,0) rectangle (\arrayw,\arrayh);

        % Overflow Bitmap
        \draw[very thick] (\overflowbitmapw,0) -- (\overflowbitmapw,\arrayh);
        \node[inner sep=2pt,anchor=west] (overflow_bitmap) at (0,0.5*\arrayh) {\small $0 1$};
        \node[inner sep=2pt] (overflow_dots) at (0.5*\overflowbitmapw+\bitw,0.5*\arrayh) {\small $\dots$};
        \draw[decorate,decoration={brace,mirror,raise=1pt,amplitude=2pt}] (0,0) -- (\overflowbitmapw,0)
            node[pos=0.5,below=4pt,inner sep=1pt,align=center] (overflow_bitmap_label) {\small Overflows \\[-4pt] \small Bitmap};

        % Mode Flag
        \draw[very thick] (\arrayw-\extensionsarrayw-\bitw,0) -- (\arrayw-\extensionsarrayw-\bitw,\arrayh);
        \node[inner sep=1.5pt] (mode_bit) at (\arrayw-\extensionsarrayw-0.5*\bitw,0.5*\arrayh) {\small 0};
        \node[inner sep=1pt] (mode_label) at (\arrayw-\extensionsarrayw-0.5*\bitw,\labely) {\small Mode};
        \draw[-stealth] ($(mode_bit.south)+(0,-0.5pt)$) -- (mode_label.north);

        % Extensions
        \draw[very thick] (\arrayw-\extensionsarrayw,0) -- (\arrayw-\extensionsarrayw,\arrayh);
        \foreach \i in {1, ..., \extensioncount} {
            \ifnum \i = 3 {
                \draw[black] (\arrayw-\extensionsarrayw+\i*\extensionw,0) -- (\arrayw-\extensionsarrayw+\i*\extensionw,\arrayh);
            }
            \else {
                \draw[black,dotted] (\arrayw-\extensionsarrayw+\i*\extensionw,0) -- (\arrayw-\extensionsarrayw+\i*\extensionw,\arrayh);
            }
            \fi
        }
        \node[inner sep=1.75pt] (extension_0) at (\arrayw-\extensionsarrayw+0.5*\extensionw,0.5*\arrayh) {\small \textit{01}};
        \node[inner sep=1.75pt] (extension_1) at (\arrayw-\extensionsarrayw+1.5*\extensionw,0.5*\arrayh) {\small \textit{10}};
        \node[inner sep=1.75pt] (extension_2) at (\arrayw-\extensionsarrayw+2.5*\extensionw,0.5*\arrayh) {\small 11};
        \node[inner sep=0pt] (extension_dots) at (\arrayw-0.5*\extensionsarrayw+0.5*\extensioncount*\extensionw,0.5*\arrayh) {\small $\dots$};
        \node[inner sep=1pt] (delimiter_label) at (\arrayw-\extensionsarrayw+4.7*\extensionw,2*\arrayh) {\small Delimiter};
        \draw[-stealth] plot[hobby] coordinates { (extension_2.north) ($(extension_2.north)+(0.4*\extensionw,0.9*\arrayh)$) (delimiter_label.west) };
        \draw[decorate,decoration={brace,mirror,raise=1pt,amplitude=2pt}] (\arrayw-\extensionsarrayw,0) -- (\arrayw,0)
            node[pos=0.5,below=4pt,inner sep=1pt] (extensions_label) {\small Extension Pool};

        % Stubs
        \foreach \i in {1, ..., \stubcountl} {
            \draw[black] (\overflowbitmapw+\i*\stubw,0) -- (\overflowbitmapw+\i*\stubw,\arrayh);
        }
        \node[inner sep=1.8pt] (stub_0) at (\overflowbitmapw+0.5*\stubw,0.5*\arrayh) {\small 111111};
        \node[inner sep=1.8pt] (stub_1) at (\overflowbitmapw+1.5*\stubw,0.5*\arrayh) {\small 101010};
        \node[inner sep=1pt] (stub_dots) at (0.5*\overflowbitmapw+0.5*\stubcountl*\stubw+0.5*\arrayw-0.5*\bitw-0.5*\extensionsarrayw,0.5*\arrayh) {\small $\dots$};
        \draw[decorate,decoration={brace,mirror,raise=1pt,amplitude=2pt}] (\overflowbitmapw,0) -- (\arrayw-\extensionsarrayw-\bitw,0)
            node[pos=0.5,below=4pt,inner sep=1pt] (stub_label) {\small Stubs};

        % Example Counters
        \ifappendix
            \def\counterw{1.9}
        \else
            \def\counterw{2.25}
        \fi
        \def\counterh{0.3}
        \def\countersep{0.5}
        \def\counterx{0.5*\arrayw-0.5*\countersep-0.5*\counterw}
        \def\countery{1.0}
        \def\counterlow{0.38*\counterw}

        \ifappendix
            \draw (\counterx-0.5*\counterw+1.09*\counterlow,\countery) -- (\counterx-0.5*\counterw+1.09*\counterlow,\countery+\counterh);
        \else
            \draw (\counterx-0.5*\counterw+1.115*\counterlow,\countery) -- (\counterx-0.5*\counterw+1.115*\counterlow,\countery+\counterh);
        \fi
        \node[inner sep=1.5pt] (counter_0_lower_bits) at (\counterx-0.5*\counterw+0.5*\counterlow,\countery+0.5435*\counterh) {\small 111111};
        \draw[-stealth] (counter_0_lower_bits.south) -- (stub_0.north);
        \node[inner sep=1.5pt] (counter_0_higher_bits) at (\counterx+0.5*\counterlow,\countery+0.525*\counterh) {\small 0000$\dots$};
        \node[inner sep=0pt] (counter_0_label) at (\counterx,\countery+1.75*\counterh) {\small Counter 0};

        \def\counterx{0.5*\arrayw+0.5*\countersep+0.5*\counterw}

        \ifappendix
            \draw (\counterx-0.5*\counterw+1.12*\counterlow,\countery) -- (\counterx-0.5*\counterw+1.12*\counterlow,\countery+\counterh);
        \else
            \draw (\counterx-0.5*\counterw+1.13*\counterlow,\countery) -- (\counterx-0.5*\counterw+1.13*\counterlow,\countery+\counterh);
        \fi
        \node[inner sep=1.5pt] (counter_1_lower_bits) at (\counterx-0.5*\counterw+0.5*\counterlow,\countery+0.5*\counterh) {\small 101010};
        \draw[-stealth] (counter_1_lower_bits.south) -- (stub_1.north);
        \node[inner sep=1.0pt] (counter_1_higher_bits) at (\counterx+0.5*\counterlow,\countery+0.5*\counterh) {\small 1010$\dots$};
        \ifappendix
            \draw[decorate,decoration={brace,raise=1pt,amplitude=2pt}] ($(counter_1_higher_bits.south west)+(12pt,0)$) -- ($(counter_1_higher_bits.south west)+(1.5pt,0)$)
                node[pos=0.5,below=2pt,inner sep=0pt] (counter_1_higher_bits_brace) { };
        \else
            \draw[decorate,decoration={brace,raise=1pt,amplitude=2pt}] ($(counter_1_higher_bits.south west)+(12.5pt,0)$) -- ($(counter_1_higher_bits.south west)+(2.0pt,0)$)
                node[pos=0.5,below=2pt,inner sep=0pt] (counter_1_higher_bits_brace) { };
        \fi
        \draw[decorate,decoration={brace,raise=1pt,amplitude=2pt}] (\arrayw-\extensionsarrayw,\arrayh) -- (\arrayw-\extensionsarrayw+2*\extensionw,\arrayh)
            node[pos=0.5,above=2pt,inner sep=0pt] (counter_1_extensions_brace) { };
        \draw[-stealth] (counter_1_higher_bits_brace.south) -- (counter_1_extensions_brace.north);
        \node[inner sep=0pt] (counter_1_label) at (\counterx,\countery+1.75*\counterh) {\small Counter 1};

        \def\axisl{0.5*\arrayw-0.8*\counterw}
        \def\axisr{0.5*\arrayw+0.8*\counterw}
        \def\axisy{\countery+3*\counterh}
    \end{tikzpicture}
    \vspace{-2.3mm}
    \caption{\counterencodingabbrv\ encodes a chunk of~$c$ counters in a cache
    line. It encodes the~$s$ lower-order bits of each counter in a stub and
    stores its remaining higher-order bits in extensions comprised of 2-bit
    fragments representing base-3 digits. Here, we use a stub length of~$s=6$.}
    \label{fig:counter_chunk}
    \vspace{-5mm}
\end{figure}

\textbf{Chunking.}
To colocate stubs and their extensions in 512-bit cache lines,
\counterencodingabbrv\ partitions each counter array into contiguous
\emph{Chunks}, each occupying a cache line. A chunk stores its counters' stubs
in a contiguous array and tracks the counters that overflow their stubs in an
\emph{Overflows Bitmap} using one bit per counter. \counterencodingabbrv\
stores the extensions of these overflowing counters contiguously in the
remaining space towards the end of the chunk, called the \emph{Extension Pool}.
Counter~0 and Counter~1 in \Cref{fig:counter_chunk}
are examples of a non-overflowing and an overflowing counter, respectively.

\textbf{Parameters.}
\counterencodingabbrv\ has two global parameters: the stub length in bits~$s$
and the number of counters per chunk~$c$. With more skewed workloads, having
shorter stubs and more counters per chunk reduces space consumption. The reason
is that skew shrinks most counter values, enabling the use of shorter stubs for
storing them. Skew also reduces the average length of the extensions, leaving
more space in a chunk to store more counters.

Later in this section, we discuss how to adapt the tuning of these parameters
to the workload as the stream evolves. When the workload is uniform, this
procedure tunes the stub length so that all counter values fit in stubs,
allowing it to remove the overflows bitmaps and yield the configuration of a
standard CMS. We illustrate the impact of this procedure on the memory
footprint in \Cref{sec:evaluation}.

\textbf{Extensions.}
An extension consists of one or more 2{\nobreakdash-}bit \emph{Fragments} that
encode the higher-order bits that exceed the length of the corresponding stub.
We use 2{\nobreakdash-}bit fragments to encode these bits at the finest possible granularity.
We use the fragment~$11$ as a delimiter at the end
of each extension. The other three combinations of bits are used to encode each
extension's value in base 3. Specifically, the base{\nobreakdash-}3
digits~$(0)_3$, $(1)_3$, and $(2)_3$ are represented as the
fragments~\textit{00}, \textit{10}, and \textit{01}, respectively. We store
fragments in increasing order of their digits' significance from left to right.
Formally, a value~$v$ is encoded in an extension by representing each of
its~$\left\lfloor \log_3 v \right\rfloor + \nolinebreak 1$ digits in a
fragment, with the~$j${\nobreakdash-}th digit being~\mbox{$v_j=\left\lfloor
v/3^j \right\rfloor \mod 3$}. We decode~$v$ back into a 32-bit binary integer
by computing the sum~$3^0 \cdot v_0 + 3^1 \cdot v_1 + \dots$. At the end of
this section, we show how to efficiently compute this sum without
multiplication~operations.

For example, Counter~1 in
\Cref{fig:counter_chunk},
has the value~21 in its lower-order bits and the value~5 in its higher-order
bits. Since~5 has two digits in base 3, it is encoded as the extension~$\langle
(2)_3, (1)_3, \text{delimiter} \rangle=\langle \text{\textit{01}},
\text{\textit{10}}, 11 \rangle$. Decoding this extension yields~$3^0 \cdot 2 +
3^1 \cdot 1 = 5$, \mbox{as expected}.

Forming a chunk's extensions from 2-bit fragments allows them to fit in 1-2
machine words in most cases. We prove in \Cref{sec:memory_analysis}
that using these extensions in conjunction with
short stubs enables \counterencodingabbrv\ to encode each counter's value in
space close to the length of its binary encoding. Even still, since extensions
are slightly less space-efficient than stubs due to their base{\nobreakdash-}3
encoding and their delimiters, we adapt the stub length at runtime to optimize
the memory footprint.

\ifappendix
\begin{figure}
    \centering
    \pgfdeclarelayer{background layer}
    \pgfsetlayers{background layer,main}
    \begin{tikzpicture}
        \def\arrayw{8.5}
        \def\arrayh{0.3}
        \def\stubw{0.09*\arrayw}
        \def\stubcountl{2}
        \def\stubcountr{2}
        \def\overflowbitmapw{0.18*\arrayw}
        \def\extensionsarrayw{0.4*\arrayw}
        \def\bitw{0.025*\arrayw}
        \def\extensionw{2*\bitw}
        \def\extensioncount{6}
        \def\partsep{0.15*\arrayw}

        \def\labely{-8.0*\arrayh}

        % Overflow Bitmap
        \draw[black] (0,0) rectangle (\overflowbitmapw,\arrayh);
        \draw (\overflowbitmapw,0) -- (\overflowbitmapw+\bitw,0);
        \draw (\overflowbitmapw,\arrayh) -- (\overflowbitmapw+\bitw,\arrayh);
        \draw[very thick] (\overflowbitmapw,0) -- (\overflowbitmapw,\arrayh);
        \node[inner sep=2pt,anchor=west] (overflow_bitmap) at (0,0.5*\arrayh) {\small $0 1 0 \text{\textbf{1}}$};
        \node[inner sep=0pt] (overflow_bitmap_dots) at (0.5*\overflowbitmapw+1.4*\bitw,0.5*\arrayh) {\small $\dots$};
        \node[inner sep=0pt] (overflow_bitmap_label) at (0.5*\overflowbitmapw,\labely) {\small Overflows};
        \node[inner sep=1pt] (rank_value) at (0.5*\overflowbitmapw,-0.75*\arrayh) {\small $\texttt{rank}(3,$~Overflows$)=1$};

        % Middle Dots
        \node[inner sep=0pt] (middle_dots) at (\overflowbitmapw+\bitw+0.5*\partsep,0.5*\arrayh) {\small $\dots$};

        % Extensions
        \draw[black] (\overflowbitmapw+\bitw+\partsep+\bitw,0) rectangle (\overflowbitmapw+\bitw+\partsep+\bitw+\extensionsarrayw,\arrayh);
        \draw[very thick] (\overflowbitmapw+\bitw+\partsep+\bitw,0) -- (\overflowbitmapw+\bitw+\partsep+\bitw,\arrayh);
        \draw (\overflowbitmapw+\bitw+\partsep,0) -- (\overflowbitmapw+\bitw+\partsep+\bitw,0);
        \draw (\overflowbitmapw+\bitw+\partsep,\arrayh) -- (\overflowbitmapw+\bitw+\partsep+\bitw,\arrayh);
        \foreach \i in {1, ..., \extensioncount} {
            \ifnum \i = 3 {
                \draw[black] (\overflowbitmapw+\bitw+\partsep+\bitw+\i*\extensionw,0) -- (\overflowbitmapw+\bitw+\partsep+\bitw+\i*\extensionw,\arrayh);
            }
            \else \ifnum \i = 5 {
                \draw[black] (\overflowbitmapw+\bitw+\partsep+\bitw+\i*\extensionw,0) -- (\overflowbitmapw+\bitw+\partsep+\bitw+\i*\extensionw,\arrayh);
            }
            \else {
                \draw[black,dotted] (\overflowbitmapw+\bitw+\partsep+\bitw+\i*\extensionw,0) -- (\overflowbitmapw+\bitw+\partsep+\bitw+\i*\extensionw,\arrayh);
            }
            \fi \fi
        }

        \begin{scope}[name prefix=extension_]
            \node[inner sep=0pt] (0) at (\overflowbitmapw+\bitw+\partsep+\bitw+0.51*\extensionw,0.5*\arrayh) {\small \textit{01}};
            \node[inner sep=0pt] (1) at (\overflowbitmapw+\bitw+\partsep+\bitw+1.5*\extensionw,0.5*\arrayh) {\small \textit{10}};
            \draw[decorate,decoration={brace,raise=1pt,amplitude=2pt}] (\overflowbitmapw+\bitw+\partsep+\bitw,\arrayh) -- (\overflowbitmapw+\bitw+\partsep+\bitw+3*\extensionw,\arrayh)
                node[pos=0.5,above=4pt,inner sep=1pt] (stub_length) {\small Skipped};
            \node[inner sep=0pt] (2) at (\overflowbitmapw+\bitw+\partsep+\bitw+2.5*\extensionw,0.5*\arrayh) {\small 11};

            \node[inner sep=0pt] (3) at (\overflowbitmapw+\bitw+\partsep+\bitw+3.5*\extensionw,0.5*\arrayh) {\small \textit{01}};
            \node[inner sep=0pt] (4) at (\overflowbitmapw+\bitw+\partsep+\bitw+4.5*\extensionw,0.5*\arrayh) {\small 11};
            % Highlight the example extension
            \begin{pgfonlayer}{background layer}
                \draw[fill=gray!20] (\overflowbitmapw+\bitw+\partsep+\bitw+3*\extensionw,0) rectangle (\overflowbitmapw+\bitw+\partsep+\bitw+5*\extensionw,\arrayh);
            \end{pgfonlayer}
            \node[inner sep=0pt] (5) at (\overflowbitmapw+\bitw+\partsep+\bitw+5.5*\extensionw,0.5*\arrayh) {\small \textit{10}};

            \node[inner sep=0pt] (dots) at (\overflowbitmapw+\bitw+\partsep+\bitw+0.5*\extensioncount*\extensionw+0.5*\extensionsarrayw,0.5*\arrayh) {\small $\dots$};
            \node[inner sep=0pt] (label) at (\overflowbitmapw+\bitw+\partsep+\bitw+0.5*\extensionsarrayw,\labely) {\small Extension Pool};
        \end{scope}

        % Select Optimization
        \begin{scope}[name prefix=shifted_extension_]
            \def\offy{1.25*\arrayh}
            \node[inner sep=0pt] (ampersand) at (\overflowbitmapw+\bitw+\partsep+\bitw-0.5*\extensionw,0.5*\arrayh-\offy) {\small \&};
            \node[inner sep=0pt] (0) at (\overflowbitmapw+\bitw+\partsep+\bitw+0.51*\extensionw,0.5*\arrayh-\offy) {\small 00};
            \node[inner sep=0pt] (1) at (\overflowbitmapw+\bitw+\partsep+\bitw+1.5*\extensionw,0.5*\arrayh-\offy) {\small 11};
            \node[inner sep=0pt] (2) at (\overflowbitmapw+\bitw+\partsep+\bitw+2.5*\extensionw,0.5*\arrayh-\offy) {\small 01};
            \node[inner sep=0pt] (3) at (\overflowbitmapw+\bitw+\partsep+\bitw+3.5*\extensionw,0.5*\arrayh-\offy) {\small 10};
            \node[inner sep=0pt] (4) at (\overflowbitmapw+\bitw+\partsep+\bitw+4.5*\extensionw,0.5*\arrayh-\offy) {\small 11};
            \node[inner sep=0pt] (5) at (\overflowbitmapw+\bitw+\partsep+\bitw+5.5*\extensionw,0.5*\arrayh-\offy) {\small 11};
            \node[inner sep=0pt] (dots) at (\overflowbitmapw+\bitw+\partsep+\bitw+0.5*\extensioncount*\extensionw+0.5*\extensionsarrayw,0.5*\arrayh-\offy) {\small $\dots$};
            \node[inner sep=0pt,right=0.2 of dots] (label) {\small Right Shift};
        \end{scope}

        \draw (\overflowbitmapw+\bitw+\partsep+\bitw,-1.5*\arrayh) -- (\overflowbitmapw+\bitw+\partsep+\bitw+\extensioncount*\extensionw,-1.5*\arrayh);

        \begin{scope}[name prefix=anded_extension_]
            \def\offy{2.75*\arrayh}
            \node[inner sep=0pt] (0) at (\overflowbitmapw+\bitw+\partsep+\bitw+0.51*\extensionw,0.5*\arrayh-\offy) {\small 00};
            \node[inner sep=0pt] (1) at (\overflowbitmapw+\bitw+\partsep+\bitw+1.5*\extensionw,0.5*\arrayh-\offy) {\small 10};
            \node[inner sep=0pt] (2) at (\overflowbitmapw+\bitw+\partsep+\bitw+2.5*\extensionw,0.5*\arrayh-\offy) {\small 01};
            \node[inner sep=0pt] (3) at (\overflowbitmapw+\bitw+\partsep+\bitw+3.5*\extensionw,0.5*\arrayh-\offy) {\small 00};
            \node[inner sep=0pt] (4) at (\overflowbitmapw+\bitw+\partsep+\bitw+4.5*\extensionw,0.5*\arrayh-\offy) {\small 11};
            \node[inner sep=0pt] (5) at (\overflowbitmapw+\bitw+\partsep+\bitw+5.5*\extensionw,0.5*\arrayh-\offy) {\small 10};
            \node[inner sep=0pt] (dots) at (\overflowbitmapw+\bitw+\partsep+\bitw+0.5*\extensioncount*\extensionw+0.5*\extensionsarrayw,0.5*\arrayh-\offy) {\small $\dots$};
        \end{scope}

        \begin{scope}[name prefix=mask_]
            \def\offy{3.75*\arrayh}
            \node[inner sep=0pt] (ampersand) at (\overflowbitmapw+\bitw+\partsep+\bitw-0.5*\extensionw,0.5*\arrayh-\offy) {\small \&};
            \node[inner sep=0pt] (0) at (\overflowbitmapw+\bitw+\partsep+\bitw+0.51*\extensionw,0.5*\arrayh-\offy) {\small 01};
            \node[inner sep=0pt] (1) at (\overflowbitmapw+\bitw+\partsep+\bitw+1.5*\extensionw,0.5*\arrayh-\offy) {\small 01};
            \node[inner sep=0pt] (2) at (\overflowbitmapw+\bitw+\partsep+\bitw+2.5*\extensionw,0.5*\arrayh-\offy) {\small 01};
            \node[inner sep=0pt] (3) at (\overflowbitmapw+\bitw+\partsep+\bitw+3.5*\extensionw,0.5*\arrayh-\offy) {\small 01};
            \node[inner sep=0pt] (4) at (\overflowbitmapw+\bitw+\partsep+\bitw+4.5*\extensionw,0.5*\arrayh-\offy) {\small 01};
            \node[inner sep=0pt] (5) at (\overflowbitmapw+\bitw+\partsep+\bitw+5.5*\extensionw,0.5*\arrayh-\offy) {\small 01};
            \node[inner sep=0pt] (dots) at (\overflowbitmapw+\bitw+\partsep+\bitw+0.5*\extensioncount*\extensionw+0.5*\extensionsarrayw,0.5*\arrayh-\offy) {\small $\dots$};
            \node[inner sep=0pt,right=0.2 of dots] (label) {\small Mask};
        \end{scope}

        \draw (\overflowbitmapw+\bitw+\partsep+\bitw,-4.0*\arrayh) -- (\overflowbitmapw+\bitw+\partsep+\bitw+\extensioncount*\extensionw,-4.0*\arrayh);

        \begin{scope}[name prefix=ends]
            \def\offy{5.25*\arrayh}
            \node[inner sep=0pt] (0) at (\overflowbitmapw+\bitw+\partsep+\bitw+0.51*\extensionw,0.5*\arrayh-\offy) {\small 00};
            \node[inner sep=0pt] (1) at (\overflowbitmapw+\bitw+\partsep+\bitw+1.5*\extensionw,0.5*\arrayh-\offy) {\small 00};
            \node[inner sep=0pt] (2) at (\overflowbitmapw+\bitw+\partsep+\bitw+2.5*\extensionw,0.5*\arrayh-\offy) {\small 01};
            \node[inner sep=0pt] (3) at (\overflowbitmapw+\bitw+\partsep+\bitw+3.5*\extensionw,0.5*\arrayh-\offy) {\small 00};
            \node[inner sep=0pt] (4) at (\overflowbitmapw+\bitw+\partsep+\bitw+4.5*\extensionw,0.5*\arrayh-\offy) {\small 0\textbf{1}};
            \node[inner sep=0pt] (5) at (\overflowbitmapw+\bitw+\partsep+\bitw+5.5*\extensionw,0.5*\arrayh-\offy) {\small 00};
            \node[inner sep=0pt] (dots) at (\overflowbitmapw+\bitw+\partsep+\bitw+0.5*\extensioncount*\extensionw+0.5*\extensionsarrayw,0.5*\arrayh-\offy) {\small $\dots$};
            \node[inner sep=0pt,right=0.2 of dots] (label) {\small Delimiters};
        \end{scope}

        \node[inner sep=1pt,anchor=west] (select_value) at (\overflowbitmapw+\bitw+\partsep+\bitw,-6.25*\arrayh) {\small $\texttt{select}(1,$ Delimiters$)=9$};
        \draw[-stealth] ($(rank_value.south east)+(-3pt,0)$) |- (select_value.west);
    \end{tikzpicture}
    \caption{\sketchcms applies rank and select operations to locate an
    overflowing counter's extension.}
    \vspace{-5mm}
    \label{fig:extension_location_example}
\end{figure}
\fi

\textbf{Locating an Extension.}
\counterencodingabbrv\ keeps the extensions in each extension pool in the same
order as that of the counters. It locates the~$i${\nobreakdash-}th counter's
extension by counting the number of preceding bits set to~1s in the overflows
bitmap. This is equivalent to the number of overflowing counters preceding
the~$i${\nobreakdash-}th counter, denoted by~$m$. Skipping over the first~$m$
extensions in the extension pool leads to the position of the target extension.
From this position, \counterencodingabbrv\ traverses the extension backwards to
decode its value.

\Cref{fig:extension_location_example}
illustrates how to identify the 3rd counter's extension in a chunk. Since the
3rd bit in the overflows bitmap is~1, the 3rd counter has an extension. As
there is one bit set to~1 before the 3rd bit in the overflows bitmap, we skip
one extension to locate the 3rd counter's extension,~shown~in~gray.

\ifnoappendix
\begin{figure}
    \vspace{-0.5mm}
    \centering
    \pgfdeclarelayer{background layer}
    \pgfsetlayers{background layer,main}
    \begin{tikzpicture}
        \def\arrayw{9.5}
        \def\arrayh{0.35}
        \def\stubw{0.09*\arrayw}
        \def\stubcountl{2}
        \def\stubcountr{2}
        \def\overflowbitmapw{0.18*\arrayw}
        \def\extensionsarrayw{0.4*\arrayw}
        \def\bitw{0.025*\arrayw}
        \def\extensionw{2*\bitw}
        \def\extensioncount{6}
        \def\partsep{0.15*\arrayw}

        \def\labely{-7.5*\arrayh}

        % Overflow Bitmap
        \draw[black] (0,0) rectangle (\overflowbitmapw,\arrayh);
        \draw (\overflowbitmapw,0) -- (\overflowbitmapw+\bitw,0);
        \draw (\overflowbitmapw,\arrayh) -- (\overflowbitmapw+\bitw,\arrayh);
        \draw[very thick] (\overflowbitmapw,0) -- (\overflowbitmapw,\arrayh);
        \node[inner sep=2pt,anchor=west] (overflow_bitmap) at (0,0.5*\arrayh) {\small $0 1 0 \text{\textbf{1}}$};
        \node[inner sep=0pt] (overflow_bitmap_dots) at (0.5*\overflowbitmapw+1.4*\bitw,0.5*\arrayh) {\small $\dots$};
        \node[inner sep=0pt] (overflow_bitmap_label) at (0.5*\overflowbitmapw,\labely) {\small Overflows};
        \node[inner sep=1pt] (rank_value) at (0.5*\overflowbitmapw,-0.75*\arrayh) {\small $\texttt{rank}(3,$~Overflows$)=1$};

        % Middle Dots
        \node[inner sep=0pt] (middle_dots) at (\overflowbitmapw+\bitw+0.5*\partsep,0.5*\arrayh) {\small $\dots$};

        % Extensions
        \draw[black] (\overflowbitmapw+\bitw+\partsep+\bitw,0) rectangle (\overflowbitmapw+\bitw+\partsep+\bitw+\extensionsarrayw,\arrayh);
        \draw[very thick] (\overflowbitmapw+\bitw+\partsep+\bitw,0) -- (\overflowbitmapw+\bitw+\partsep+\bitw,\arrayh);
        \draw (\overflowbitmapw+\bitw+\partsep,0) -- (\overflowbitmapw+\bitw+\partsep+\bitw,0);
        \draw (\overflowbitmapw+\bitw+\partsep,\arrayh) -- (\overflowbitmapw+\bitw+\partsep+\bitw,\arrayh);
        \foreach \i in {1, ..., \extensioncount} {
            \ifnum \i = 3 {
                \draw[black] (\overflowbitmapw+\bitw+\partsep+\bitw+\i*\extensionw,0) -- (\overflowbitmapw+\bitw+\partsep+\bitw+\i*\extensionw,\arrayh);
            }
            \else \ifnum \i = 5 {
                \draw[black] (\overflowbitmapw+\bitw+\partsep+\bitw+\i*\extensionw,0) -- (\overflowbitmapw+\bitw+\partsep+\bitw+\i*\extensionw,\arrayh);
            }
            \else {
                \draw[black,dotted] (\overflowbitmapw+\bitw+\partsep+\bitw+\i*\extensionw,0) -- (\overflowbitmapw+\bitw+\partsep+\bitw+\i*\extensionw,\arrayh);
            }
            \fi \fi
        }

        \begin{scope}[name prefix=extension_]
            \node[inner sep=0pt] (0) at (\overflowbitmapw+\bitw+\partsep+\bitw+0.51*\extensionw,0.5*\arrayh) {\small \textit{01}};
            \node[inner sep=0pt] (1) at (\overflowbitmapw+\bitw+\partsep+\bitw+1.5*\extensionw,0.5*\arrayh) {\small \textit{10}};
            \draw[decorate,decoration={brace,raise=1pt,amplitude=2pt}] (\overflowbitmapw+\bitw+\partsep+\bitw,\arrayh) -- (\overflowbitmapw+\bitw+\partsep+\bitw+3*\extensionw,\arrayh)
                node[pos=0.5,above=4pt,inner sep=1pt] (stub_length) {\small Skipped};
            \node[inner sep=0pt] (2) at (\overflowbitmapw+\bitw+\partsep+\bitw+2.5*\extensionw,0.5*\arrayh) {\small 11};

            \node[inner sep=0pt] (3) at (\overflowbitmapw+\bitw+\partsep+\bitw+3.5*\extensionw,0.5*\arrayh) {\small \textit{01}};
            \node[inner sep=0pt] (4) at (\overflowbitmapw+\bitw+\partsep+\bitw+4.5*\extensionw,0.5*\arrayh) {\small 11};
            % Highlight the example extension
            \begin{pgfonlayer}{background layer}
                \draw[fill=gray!20] (\overflowbitmapw+\bitw+\partsep+\bitw+3*\extensionw,0) rectangle (\overflowbitmapw+\bitw+\partsep+\bitw+5*\extensionw,\arrayh);
            \end{pgfonlayer}
            \node[inner sep=0pt] (5) at (\overflowbitmapw+\bitw+\partsep+\bitw+5.5*\extensionw,0.5*\arrayh) {\small \textit{10}};

            \node[inner sep=0pt] (dots) at (\overflowbitmapw+\bitw+\partsep+\bitw+0.5*\extensioncount*\extensionw+0.5*\extensionsarrayw,0.5*\arrayh) {\small $\dots$};
            \node[inner sep=0pt] (label) at (\overflowbitmapw+\bitw+\partsep+\bitw+0.5*\extensionsarrayw,\labely) {\small Extension Pool};
        \end{scope}

        % Select Optimization
        \begin{scope}[name prefix=shifted_extension_]
            \def\offy{1.25*\arrayh}
            \node[inner sep=0pt] (ampersand) at (\overflowbitmapw+\bitw+\partsep+\bitw-0.5*\extensionw,0.5*\arrayh-\offy) {\small \&};
            \node[inner sep=0pt] (0) at (\overflowbitmapw+\bitw+\partsep+\bitw+0.51*\extensionw,0.5*\arrayh-\offy) {\small 00};
            \node[inner sep=0pt] (1) at (\overflowbitmapw+\bitw+\partsep+\bitw+1.5*\extensionw,0.5*\arrayh-\offy) {\small 11};
            \node[inner sep=0pt] (2) at (\overflowbitmapw+\bitw+\partsep+\bitw+2.5*\extensionw,0.5*\arrayh-\offy) {\small 01};
            \node[inner sep=0pt] (3) at (\overflowbitmapw+\bitw+\partsep+\bitw+3.5*\extensionw,0.5*\arrayh-\offy) {\small 10};
            \node[inner sep=0pt] (4) at (\overflowbitmapw+\bitw+\partsep+\bitw+4.5*\extensionw,0.5*\arrayh-\offy) {\small 11};
            \node[inner sep=0pt] (5) at (\overflowbitmapw+\bitw+\partsep+\bitw+5.5*\extensionw,0.5*\arrayh-\offy) {\small 11};
            \node[inner sep=0pt] (dots) at (\overflowbitmapw+\bitw+\partsep+\bitw+0.5*\extensioncount*\extensionw+0.5*\extensionsarrayw,0.5*\arrayh-\offy) {\small $\dots$};
            \node[inner sep=0pt,right=0.2 of dots] (label) {\small Right Shift};
        \end{scope}

        \draw (\overflowbitmapw+\bitw+\partsep+\bitw,-1.5*\arrayh) -- (\overflowbitmapw+\bitw+\partsep+\bitw+\extensioncount*\extensionw,-1.5*\arrayh);

        \begin{scope}[name prefix=anded_extension_]
            \def\offy{2.75*\arrayh}
            \node[inner sep=0pt] (0) at (\overflowbitmapw+\bitw+\partsep+\bitw+0.51*\extensionw,0.5*\arrayh-\offy) {\small 00};
            \node[inner sep=0pt] (1) at (\overflowbitmapw+\bitw+\partsep+\bitw+1.5*\extensionw,0.5*\arrayh-\offy) {\small 10};
            \node[inner sep=0pt] (2) at (\overflowbitmapw+\bitw+\partsep+\bitw+2.5*\extensionw,0.5*\arrayh-\offy) {\small 01};
            \node[inner sep=0pt] (3) at (\overflowbitmapw+\bitw+\partsep+\bitw+3.5*\extensionw,0.5*\arrayh-\offy) {\small 00};
            \node[inner sep=0pt] (4) at (\overflowbitmapw+\bitw+\partsep+\bitw+4.5*\extensionw,0.5*\arrayh-\offy) {\small 11};
            \node[inner sep=0pt] (5) at (\overflowbitmapw+\bitw+\partsep+\bitw+5.5*\extensionw,0.5*\arrayh-\offy) {\small 10};
            \node[inner sep=0pt] (dots) at (\overflowbitmapw+\bitw+\partsep+\bitw+0.5*\extensioncount*\extensionw+0.5*\extensionsarrayw,0.5*\arrayh-\offy) {\small $\dots$};
        \end{scope}

        \begin{scope}[name prefix=mask_]
            \def\offy{3.75*\arrayh}
            \node[inner sep=0pt] (ampersand) at (\overflowbitmapw+\bitw+\partsep+\bitw-0.5*\extensionw,0.5*\arrayh-\offy) {\small \&};
            \node[inner sep=0pt] (0) at (\overflowbitmapw+\bitw+\partsep+\bitw+0.51*\extensionw,0.5*\arrayh-\offy) {\small 01};
            \node[inner sep=0pt] (1) at (\overflowbitmapw+\bitw+\partsep+\bitw+1.5*\extensionw,0.5*\arrayh-\offy) {\small 01};
            \node[inner sep=0pt] (2) at (\overflowbitmapw+\bitw+\partsep+\bitw+2.5*\extensionw,0.5*\arrayh-\offy) {\small 01};
            \node[inner sep=0pt] (3) at (\overflowbitmapw+\bitw+\partsep+\bitw+3.5*\extensionw,0.5*\arrayh-\offy) {\small 01};
            \node[inner sep=0pt] (4) at (\overflowbitmapw+\bitw+\partsep+\bitw+4.5*\extensionw,0.5*\arrayh-\offy) {\small 01};
            \node[inner sep=0pt] (5) at (\overflowbitmapw+\bitw+\partsep+\bitw+5.5*\extensionw,0.5*\arrayh-\offy) {\small 01};
            \node[inner sep=0pt] (dots) at (\overflowbitmapw+\bitw+\partsep+\bitw+0.5*\extensioncount*\extensionw+0.5*\extensionsarrayw,0.5*\arrayh-\offy) {\small $\dots$};
            \node[inner sep=0pt,right=0.2 of dots] (label) {\small Mask};
        \end{scope}

        \draw (\overflowbitmapw+\bitw+\partsep+\bitw,-4.0*\arrayh) -- (\overflowbitmapw+\bitw+\partsep+\bitw+\extensioncount*\extensionw,-4.0*\arrayh);

        \begin{scope}[name prefix=ends]
            \def\offy{5.25*\arrayh}
            \node[inner sep=0pt] (0) at (\overflowbitmapw+\bitw+\partsep+\bitw+0.51*\extensionw,0.5*\arrayh-\offy) {\small 00};
            \node[inner sep=0pt] (1) at (\overflowbitmapw+\bitw+\partsep+\bitw+1.5*\extensionw,0.5*\arrayh-\offy) {\small 00};
            \node[inner sep=0pt] (2) at (\overflowbitmapw+\bitw+\partsep+\bitw+2.5*\extensionw,0.5*\arrayh-\offy) {\small 01};
            \node[inner sep=0pt] (3) at (\overflowbitmapw+\bitw+\partsep+\bitw+3.5*\extensionw,0.5*\arrayh-\offy) {\small 00};
            \node[inner sep=0pt] (4) at (\overflowbitmapw+\bitw+\partsep+\bitw+4.5*\extensionw,0.5*\arrayh-\offy) {\small 0\textbf{1}};
            \node[inner sep=0pt] (5) at (\overflowbitmapw+\bitw+\partsep+\bitw+5.5*\extensionw,0.5*\arrayh-\offy) {\small 00};
            \node[inner sep=0pt] (dots) at (\overflowbitmapw+\bitw+\partsep+\bitw+0.5*\extensioncount*\extensionw+0.5*\extensionsarrayw,0.5*\arrayh-\offy) {\small $\dots$};
            \node[inner sep=0pt,right=0.2 of dots] (label) {\small Delimiters Bitmap};
        \end{scope}

        \node[inner sep=1pt,anchor=west] (select_value) at (\overflowbitmapw+\bitw+\partsep+\bitw,-6.25*\arrayh) {\small $\texttt{select}(1,$ Delimiters$)=9$};
        \draw[-stealth] ($(rank_value.south east)+(-3pt,0)$) |- (select_value.west);
    \end{tikzpicture}
    \vspace{-3mm}
    \caption{\sketchcms applies rank and select operations to locate an
    overflowing counter's extension.}
    \vspace{-5mm}
    \label{fig:extension_location_example}
\end{figure}
\fi

Although one can locate an overflowing counter's extension by looping over the
fragments in the extension pool, doing so incurs high CPU overheads. Instead,
\counterencodingabbrv\ quickly locates an extension using rank and select
operations~\cite{Poppy,FlorianRankandSelect,CompactPATTrees,RankandSelectDict,GQF}.
Given a bitmap~$B$, the $\texttt{rank}(i,B)$ operation counts the number of~1s
strictly before the $i$-th bit, and the $\text{select}(i,B)$ operation returns
the position of the~$i${\nobreakdash-}th 1. Later in this section, we show how
to use specialized CPU instructions to implement these operations efficiently. 

The first step is to determine~$m$, the number of overflowing counters to the
left of the~$i${\nobreakdash-}th counter. As this is the number of 1s before
the $i${\nobreakdash-}th bit in the overflows bitmap, \counterencodingabbrv\
applies the rank operation to evaluate~$m=\texttt{rank}(i,$~overflows$)$.

The second step is to identify the target extension's delimiter by skipping
over the~$m$ delimiters preceding it. We do this using several bit manipulation
techniques. First, we convert the extension pool into a so-called
\emph{Delimiters Bitmap}. This bitmap contains one bit set to~1 at the position
of every delimiter and~0s everywhere else, as illustrated at the bottom right
of
\Cref{fig:extension_location_example}.
To derive this bitmap, we take the bitwise-and of the extension pool with a
version of itself shifted to the right by one bit. The result is a bitmap with
a~1 at the end of any two consecutive~1s in the original extension pool, as
illustrated in
\Cref{fig:extension_location_example}.
We mask out all bits corresponding to the first bit of a fragment to ensure
only the end of each delimiter is a~1. In sum, the delimiters bitmap is derived
as
\begin{center}
    delimiters $=$ extensions \& $($extensions >‌> 1$)$ \& $0101\dots0101$.
\end{center}
Finally, we use the select operation to
evaluate~$\texttt{select}(m,$~delimiters$)$, which skips over the first~$m$
delimiters and returns the position of the~$m${\nobreakdash-}th delimiter. We
sequentially apply the bitwise operations above to each machine word comprising
the chunk's extensions. Locating extensions in this way is efficient, as a
chunk's extensions typically fit in 1-2 words. 
%We will shortly show that updates rarely modify extensions, enabling high insertion throughput.

\Cref{fig:extension_location_example} shows an example of locating the 3rd
counter's extension using this workflow. Here,
$\texttt{rank}(3,$~overflows$)=1$. After deriving the delimiters bitmap by
manipulating the extension pool, \counterencodingabbrv\
evaluates~$\texttt{select}(1,$~delimiters$)=9$ to conclude that the target
extension ends at the 9th bit.

\textbf{Tails.}
In the rare event that many heavy hitters hash to the same chunk, many of the
chunk's counters will overflow, potentially growing the cumulative length of
the extensions beyond the size of the chunk. \counterencodingabbrv\ handles
this case by allocating an external array of 32-bit integers called
\emph{Tails} for the chunk, with one integer per counter. It stores
the~$i${\nobreakdash-}th counter's higher-order bits in
the~$i${\nobreakdash-}th tail in binary. Each chunk contains a mode flag
indicating if it has a tails array, as depicted in
\Cref{fig:counter_chunk_with_tails}.
If it does, we repurpose the space taken up by the extension pool to store a
pointer to the tails array. We ensure that the extension pool is large enough
to accommodate this pointer (i.e., at least 48 bits\footnote{The upper 16 bits
in a 64-bit pointer are unused on many operating systems due to their use of
4-level paging~\cite{Intel4LevelPaging}. As such, we only store the lower 48
bits of the pointer. That said, modern machines and operating systems have
moved towards using 5-level paging~\cite{Intel5LevelPaging}, which uses 57-bit
addresses. Despite this shift, operating systems like Linux still default to
using 48-bit addresses for compatibility reasons~\cite{Linux5LevelPaging}.})
when tuning the number of counters per chunk and the stub size. Any space after
the pointer is unused. 

\begin{figure}
    \centering
    \begin{tikzpicture}
        \ifappendix
            \def\arrayw{7.5}
            \def\arrayh{0.3}
        \else
            \def\arrayw{9.0}
            \def\arrayh{0.35}
        \fi
        \def\tailsarrayw{0.5*\arrayw}
        \def\tailsarrayh{1.0*\arrayh}
        \def\tailsarrayoffsetx{0.52*\arrayw}
        \def\tailsarrayoffsety{3.0*\arrayh}
        \def\tailsptrarrowoffset{0.05*\tailsarrayw}
        \def\stubw{0.15*\arrayw}
        \def\tailw{1.1*\stubw}
        \def\stubcountl{2}
        \def\overflowbitmapw{0.15*\arrayw}
        \def\extensionsarrayw{0.3*\arrayw}
        \def\bitw{0.025*\arrayw}
        \def\extensionw{2*\bitw}
        \def\extensioncount{4}
        \def\labely{-2.25*\arrayh}
        \def\ptrw{10*\bitw}
        \def\extensioncount{3}

        \draw[black] (0,0) rectangle (\arrayw,\arrayh);

        % Overflow Bitmap
        \draw[very thick] (\overflowbitmapw,0) -- (\overflowbitmapw,\arrayh);
        \node[inner sep=2pt,anchor=west] (overflow_bitmap) at (0,0.5*\arrayh) {\small $1 1$};
        \node[inner sep=0pt] (overflow_bitmap_dots) at (0.5*\overflowbitmapw+\bitw,0.5*\arrayh) {\small $\dots$};

        % Spill Bit and Unused Part
        \draw[very thick] (\arrayw-\extensionsarrayw-\bitw,0) -- (\arrayw-\extensionsarrayw-\bitw,\arrayh);
        %\draw[pattern=north east lines,pattern color=gray!50] (\arrayw-\extensionsarrayw-\bitw,0) rectangle (\arrayw-\extensionsarrayw,\arrayh);
        \node[inner sep=1.5pt] (mode_bit) at (\arrayw-\extensionsarrayw-0.5*\bitw,0.5*\arrayh) {\small $1$};
        \node[inner sep=1pt,below=0.4*\arrayh of mode_bit] (mode_bit_label) {\small Mode};

        % Spill Pointer
        \draw[pattern=north east lines,pattern color=gray!50] (\arrayw-\extensionsarrayw+\ptrw,0) rectangle (\arrayw,\arrayh);
        \draw[very thick] (\arrayw-\extensionsarrayw,0) -- (\arrayw-\extensionsarrayw,\arrayh);
        \node[inner sep=1.6pt] (tails_ptr) at (\arrayw-\extensionsarrayw+0.5*\ptrw,0.5*\arrayh) {\small Pointer};

        % Stub Counters
        \foreach \i in {1, ..., \stubcountl} {
            \draw[black] (\overflowbitmapw+\i*\stubw,0) -- (\overflowbitmapw+\i*\stubw,\arrayh);
        }
        \node[inner sep=1.6pt] (stub_0) at (\overflowbitmapw+0.5*\stubw,0.5*\arrayh) {\small $111111$};
        \node[inner sep=1.6pt] (stub_1) at (\overflowbitmapw+1.5*\stubw,0.5*\arrayh) {\small $000101$};
        \node[inner sep=1pt] (stub_dots) at (0.5*\overflowbitmapw+0.5*\stubcountl*\stubw+0.5*\arrayw-0.5*\bitw-0.5*\extensionsarrayw,0.5*\arrayh) {\small $\dots$};

        % Tails
        \draw[black] (\tailsarrayoffsetx,\tailsarrayoffsety) rectangle (\tailsarrayoffsetx+\tailsarrayw,\tailsarrayoffsety-\tailsarrayh);
        \foreach \i in {1, ..., \stubcountl} {
            \draw[black] (\tailsarrayoffsetx+\i*\tailw,\tailsarrayoffsety) -- (\tailsarrayoffsetx+\i*\tailw,\tailsarrayoffsety-\tailsarrayh);
        }
        \node[inner sep=1.6pt] (tail_0) at (\tailsarrayoffsetx+0.3*\tailw,\tailsarrayoffsety-0.5*\tailsarrayh) {\small $1000$};
        \node[inner sep=0pt,right=1pt of tail_0] (tail_0_dots) {\small $\dots$};
        \node[inner sep=1.6pt] (tail_1) at (\tailsarrayoffsetx+1.3*\tailw,\tailsarrayoffsety-0.5*\tailsarrayh) {\small $1110$};
        \node[inner sep=0pt,right=1pt of tail_1] (tail_1_dots) {\small $\dots$};
        \node[inner sep=1pt] (tail_dots) at (\tailsarrayoffsetx+0.5*\tailsarrayw+0.5*\stubcountl*\tailw,\tailsarrayoffsety-0.5*\tailsarrayh) {\small $\dots$};
        \draw[densely dotted] (tails_ptr.north) |- (\tailsarrayoffsetx-\tailsptrarrowoffset,\tailsarrayoffsety/2);
        \draw[-stealth,densely dotted] (\tailsarrayoffsetx-\tailsptrarrowoffset,\tailsarrayoffsety/2) |- (\tailsarrayoffsetx,\tailsarrayoffsety-\tailsarrayh/2);
        \node[inner sep=2pt,anchor=west] (tails_label) at (\tailsarrayoffsetx+\tailsarrayw,\tailsarrayoffsety-\tailsarrayh/2) {\small Tails};

        % Example Counters
        \ifappendix
            \def\counterw{1.9}
        \else
            \def\counterw{2.25}
        \fi
        \def\counterh{0.3}
        \def\countersep{0.5}
        \def\counterx{0.5*\arrayw-0.5*\countersep-0.5*\counterw}
        \ifappendix
            \def\countery{1.4}
        \else
            \def\countery{1.6}
        \fi
        \def\counterlow{0.38*\counterw}

        \ifappendix
            \draw (\counterx-0.5*\counterw+1.09*\counterlow,\countery) -- (\counterx-0.5*\counterw+1.09*\counterlow,\countery+\counterh);
        \else
            \draw (\counterx-0.5*\counterw+1.11*\counterlow,\countery) -- (\counterx-0.5*\counterw+1.11*\counterlow,\countery+\counterh);
        \fi
        \node[inner sep=1.5pt] (counter_0_lower_bits) at (\counterx-0.5*\counterw+0.5*\counterlow,\countery+0.5435*\counterh) {\small $111111$};
        \draw[-stealth] (counter_0_lower_bits.south) -- (stub_0.north);
        \node[inner sep=1.5pt] (counter_0_higher_bits) at (\counterx+0.5*\counterlow,\countery+0.535*\counterh) {\small $1000\dots$};
        \draw[-stealth] (counter_0_higher_bits.south) -- (\tailsarrayoffsetx+0.5*\tailw,\tailsarrayoffsety);
        \node[inner sep=0pt] (counter_0_label) at (\counterx,\countery+1.75*\counterh) {\small Counter 0};

        \def\counterx{0.5*\arrayw+0.5*\countersep+0.5*\counterw}

        \draw (\counterx-0.5*\counterw+1.09*\counterlow,\countery) -- (\counterx-0.5*\counterw+1.09*\counterlow,\countery+\counterh);
        \node[inner sep=1.5pt] (counter_1_lower_bits) at (\counterx-0.5*\counterw+0.5*\counterlow,\countery+0.5*\counterh) {\small $000101$};
        \draw[-stealth] (counter_1_lower_bits.south) -- (stub_1.north);
        \node[inner sep=1.0pt] (counter_1_higher_bits) at (\counterx+0.5*\counterlow,\countery+0.5*\counterh) {\small $1110\dots$};
        \draw[-stealth] (counter_1_higher_bits.south) -- (\tailsarrayoffsetx+1.5*\tailw,\tailsarrayoffsety);
        \node[inner sep=0pt] (counter_1_label) at (\counterx,\countery+1.75*\counterh) {\small Counter 1};
    \end{tikzpicture}
    \ifnoappendix
        \vspace{-2mm}
    \fi
    \caption{When a chunk's extensions outgrow its space, it stores a pointer
    to a dedicated external array of tails to store the counters' higher-order
    bits.}
    \vspace{-5mm}
    \label{fig:counter_chunk_with_tails}
\end{figure}

\textbf{Incrementing Counters.}
When an insertion maps to the~$i${\nobreakdash-}th counter within a chunk,
\sketchcms increments it by adding one to the lower-order bits stored in
the~$i${\nobreakdash-}th stub, i.e., stub[$i$]. Most of the time, stub[$i$]'s
increment does not carry over to the counter's higher-order bits, allowing the
insertion to terminate early and achieve high throughput.

Whenever stub[$i$]'s increment carries over, \sketchcms adds one to the
counter's extension or tail. For an extension, this entails incrementing its
first digit and propagating a carry to the other digits as necessary. This
carry may propagate beyond the end of the extension, leading to the addition of
a new digit. \sketchcms makes room for the fragment representing this digit by
shifting subsequent extensions to the right. If the~$i${\nobreakdash-}th
counter did not have an extension, a new extension encoding the value~1 is
created for it. We use rank and select operations to determine the appropriate
position of this extension in the extension pool. Updating a counter's
higher-order bits is much simpler if they are stored within a tail's binary
integer. For the heavy hitters (which comprise most of the stream), this yields
faster insertions, as their counters often use tails.

\textbf{Decrementing Counters.}
Analogously, when a deletion maps to the~$i${\nobreakdash-}th counter in a
chunk, \sketchcms decrements the lower-order bits in~stub[$i$]. If stub[$i$]
borrows from the counter's higher-order bits as a result, \sketchcms decrements
the corresponding extension or tail. This may trigger the removal of a digit
from the counter's extension, in which case subsequent extensions are shifted
to the left to keep the extension pool tightly packed.

\textbf{Adaptive Tuning.}
To minimize space usage and ensure that a small fraction of the extension pools
are unused at any time, \sketchcms adaptively retunes \counterencodingabbrv's
number of counters per chunk~$c$ and its stub length~$s$. Retuning is triggered
whenever more than~$1\%$ of the chunks use tails arrays (due to insertions
growing the extensions), or there are more than two unused bits per counter
(due to deletions shrinking the extensions).\footnote{These conditions lead to
at most 20\% overprovisioned space upfront when counters are short (1 byte).
This fraction decreases quickly as the counters grow.} During tuning,
\sketchcms scans the arrays and computes a coarse histogram of the counter
values based on their bit length. It uses this histogram to calculate the
expected total length of a chunk's extensions for a given combination of
parameters~$c$ and $s$. For each combination, it applies Chebyshev's inequality
to the total extension length within a chunk to derive a conservative bound on
the proportion of chunks that will use tails arrays. It uses this bound to
estimate \sketchcms's memory footprint when using the corresponding tuning, and
picks the tuning that minimizes the memory footprint. The result sets the stub
length~$s$ to be slightly longer than the length of the average counter value.
This bounds the number of extensions and makes the number of stored tails
arrays and their memory overhead negligible. The amortized cost of retuning
\counterencodingabbrv\ is small, as many more updates than the number of
counters are required to trigger retuning. We demonstrate this under
Experiment~\hyperlink{experiment:vale_tuning}{3} of \Cref{sec:evaluation}.

\textbf{Impact of Cache Line Size.}
The aforementioned tuning procedure also allows \sketchcms to support different
cache line sizes. Having smaller cache lines leads to the tuning procedure
storing fewer counters per chunk. This reduction increases the variance of the
total size of a chunk's extensions, causing the tuning procedure to slightly
increase the overprovisioned space for the chunk's extension
pool.\footnote{We have observed that reducing the cache line size from 512
bits to 256 and 128 bits increases the proportion of overprovisioned memory
by~5\% and 6\%, respectively.} Conversely, having larger cache lines does not
increase the required overprovisioned space, though it may degrade query
performance if we increase the chunk size to match the cache line. This is
because locating a counter's extension becomes more expensive as extension
pools grow, as we process them word by word. \sketchcms easily addresses this
problem by storing multiple chunks per cache line (e.g., two 512-bit chunks
within a 1024{\nobreakdash-}bit cache line).

\textbf{General Applicability.}
\counterencodingabbrv\ is able to compactly encode arbitrary variable-length
data by treating the binary representation of each data item as the value of a
counter. Thus, it can be used in other applications operating on such data to
improve memory~efficiency.

\hypertarget{optimization:rank_and_select}{\textbf{Optimization~1: Hardware Accelerated Rank and Select.}}
\sketchcms utilizes specialized CPU instructions to implement rank and select
operations. It implements~$\texttt{rank}(i,B)$ by applying the \texttt{POPCNT}
instruction to the first~$\left\lceil i/64 \right\rceil$ words of the
bitmap~$B$ to count the~1s before the~$i${\nobreakdash-}th
bit.\footnote{Formally, denoting the~$j${\nobreakdash-}th word in~$B$ as~$B_j$,
this process corresponds to computing \texttt{POPCNT}$ \left(B_{\lfloor i/64
\rfloor}\text{~\&~}((1 \text{ <‌< } (i \mod 64))-1) \right) +
\sum_{j=0}^{\lfloor i/64 \rfloor -1}$ \texttt{POPCNT}$(B_j)$.} It also
implements~$\texttt{select}(i,B)$ by counting the 1s in~$B$'s words to locate
the word containing the $i${\nobreakdash-}th~1. Once found, it applies the x86
BMI instructions \texttt{PDEP} and \texttt{TZCNT} to evaluate the target bit's
position within that word, as described in \cite{GQF}.\footnote{If
the~$j${\nobreakdash-}th word $B_j$ contains the~$i${\nobreakdash-}th 1 and
there are~$r$ 1s in the preceding words, the result is computed as~$64 \cdot j
+ $\texttt{TZCNT}$($\texttt{PDEP}$((1 \text{ <‌< } (i-r)) - 1, B_j))$, where the
first operand of \texttt{PDEP} is the payload and the second operand is
the~mask~\cite{PDEP}.} Each of these commands is applied to an average of 1-2
machine words, as the bitmaps \sketchcms operates over are 1-2 words long.

\sketchcms includes efficient fallback implementations of rank and select for
processors that lack the above instructions (e.g., ARM). We employ two lookup
tables structured as direct-access arrays: one for counting the number of 1s in
a byte to compute rank, and one for applying select within a byte. The entry
corresponding to the value of a byte and the parameter of the operation in each
table indicates the result of the relevant operation with these inputs.
\sketchcms iteratively applies these tables to the bytes within a bitmap to
evaluate rank and select. Doing so slightly degrades query performance, as
queries must always retrieve the extension of a counter. Nevertheless,
insertions and deletions remain unaffected since they rarely access a
counter's~extension.

\hypertarget{optimization:lookup_tables}{\textbf{Optimization~2: Updating Stubs with Lookup Tables.}}
Since stubs are not byte-aligned, modifying them can be costly. This is because
extracting and writing back a stub to the machine words containing it entails
five extra bitwise operations on top of the addition operation. \sketchcms
bypasses this cost by directly updating the stub's words using a single
addition. It achieves this using two precomputed lookup tables. The first table
stores the position of the word each stub lies in, and the second table stores
the stub's position within its word. \sketchcms updates a stub by adding or
subtracting one from the bit indicated by the second table in the word
indicated by the first. For example, in the context of
\Cref{fig:counter_chunk}, the first lookup table indicates that the~1st stub
lies in the 1st word, which has a binary representation of~$111111101010\dots$.
The second lookup table indicates that this stub starts at the~6th bit in the
1st word, allowing us to increment it by adding~$000000100000\dots$, thus
yielding~$111111011010\dots$ as the updated word.

\hypertarget{optimization:prefetching}{\textbf{Optimization~3: Prefetching Chunks.}}
We allow multiple chunks to be read in parallel within a single operation as
well as across multiple operations using prefetching techniques similar to that
of~\cite{StingySketch}. Specifically, we maintain a prefetching queue for each
array in \sketchcms that records the eight most recent requests for chunks.
Each time a chunk is requested from an array, a software prefetch instruction
is issued for it, and the request is pushed into the array's queue. When a
queue fills up, the chunk of its oldest request will be present in the L1
cache, allowing it to be read efficiently while other chunks are being loaded
into the cache. 

\hypertarget{optimization:hash_splicing}{\textbf{Optimization~4: Hash Splicing.}}
Recall that a key is hashed into each of the~$d$ arrays during an update or
query. Instead of applying~$d$ separate hash functions and incurring their CPU
overheads, \sketchcms applies a single hash function to the key to compute a
long hash. It splices the result into~$d$ sub-hashes of equal length and uses
each as the hash value of an array. 

\hypertarget{optimization:fast_division}{\textbf{Optimization~5: Fast Division.}}
Determining the chunk of a counter at offset~$i$ in an array can be expensive
if done by divinding $i$ by the number of counters per chunk~$c$, as $c$ may
not be a power of two. We avoid this cost by simulating the division through
fixed-point multiplication~\cite{HackersDelight}. Specifically, we compute the
expression~$\left\lfloor \left(i/2^{\lceil \log_2 c \rceil}\right) \cdot
\left(2^{\lceil \log_2 c \rceil}/c\right) \right\rfloor = \left\lfloor i/c
\right\rfloor$ by approximating the multiplicands on the left-hand side as
fractional binary numbers with a fixed number of precision bits. To this end,
we precompute the latter multiplicand's fixed-point representation
with~$\left\lceil \log_2 w \right\rceil$ precision bits, where~$w$ is the size
of an array. This many bits of precision ensures that, when multiplying by a
counter offset~$i$ smaller than $w$, the result has an error of at most one,
allowing \sketchcms to safely truncate it to compute the chunk's (integral)
offset. For the former multiplicand, we treat~$i$'s binary representation as a
fixed-point fractional value with~$\left\lceil \log_2 c \right\rceil$ precision
bits. Finally, we multiply these values using standard integer multiplication
and shift the result~$\left\lceil \log_2 w \right\rceil + \left\lceil \log_2 c
\right\rceil$ bits to the right to remove the precision
bits~and~derive~the~answer. 

\hypertarget{optimization:horner_hack}{\textbf{Optimization~6: Horner's Method with Shifts and Adds.}}
\sketchcms decodes a counter's extension with
digits~$v_0, v_1, \cdots, v_k$ by computing the sum~$3^0 \cdot v_0 + 3^1 \cdot
v_1 + \dots + 3^k \cdot v_k$. We speed up this process using Horner's
method~\cite{Horner,HornerKnuth} and implementing multiplication by~3 using
shift and add operations~\cite{HackersDelight}. Specifically, we compute a
running sum starting from the last term (i.e., initializing the sum to~$3^0
\cdot v_k$) and move towards the first term, handling them one by one. That is,
we handle~$v_{k-1}$ by multiplying the running sum by~3 and adding~$v_{k-1}$ to
derive the sum~$3^0 \cdot v_{k-1} + 3^1 \cdot v_k$, and so on. We multiply the
running sum by~3 by adding it to a copy of itself shifted to the left by one
bit (which is equivalent to multiplication~by~2).

\begin{figure}
    \centering
    \pgfdeclarelayer{background layer}
    \pgfsetlayers{background layer,main}
    \begin{tikzpicture}
        \def\arraycount{3}
        \ifappendix
            \def\arrayw{1}
            \def\arrayh{0.25}
            \def\arraysep{0.25}
        \else
            \def\arrayw{1.25}
            \def\arrayh{0.3}
            \def\arraysep{0.3}
        \fi
        \def\cmssep{0.28}
        \def\cmstitlesep{0.1}
        \def\arrowsep{0.2*\cmssep}
        \def\midy{0.5*\arraycount*\arraysep-0.5*\arraysep+0.5*\arraycount*\arrayh}
        \def\countercount{2} % Have to change this thing manually in the for loops, sadly
        \def\counterw{\arrayw/\countercount}

        %\node[inner sep=0pt,rotate=90] (chain_label) at (-0.75*\cmssep,\midy) {\small Chain of Sketches};

        % First CMS
        \def\counterarray{{3,1,3,1,4,0}}
        \foreach \i in {1, ..., \arraycount} {
            \draw[gray] (0,\i*\arraysep+\i*\arrayh-\arrayh-\arraysep) rectangle (\arrayw,\i*\arraysep+\i*\arrayh-\arraysep);
            \foreach \j in {1, ..., 2} {
                \pgfmathtruncatemacro{\counterindex}{2*\i+\j-3}
                \pgfmathparse{\counterarray[\counterindex]}
                \let\countervalue\pgfmathresult
                \node[inner sep=0pt] at (\j*\counterw-\counterw/2,\i*\arraysep+\i*\arrayh-\arrayh/2-\arraysep) {\small \textcolor{gray}{$\countervalue$}};
                \draw[gray] (\j*\counterw,\i*\arraysep+\i*\arrayh-\arrayh-\arraysep) rectangle (\j*\counterw,\i*\arraysep+\i*\arrayh-\arraysep);
            }
        }

        % Second CMS
        \def\counterarray{{3,5,4,8,6,3,6,5,9,1,8,2}}
        \foreach \i in {1, ..., \arraycount} {
            \draw (\arrayw+\cmssep,\i*\arraysep+\i*\arrayh-\arrayh-\arraysep) rectangle (3*\arrayw+\cmssep,\i*\arraysep+\i*\arrayh-\arraysep);
            \foreach \j in {1, ..., 4} {
                \pgfmathtruncatemacro{\counterindex}{4*\i+\j-5}
                \pgfmathparse{\counterarray[\counterindex]}
                \let\countervalue\pgfmathresult
                \node[inner sep=0pt] at (\arrayw+\cmssep+\j*\counterw-\counterw/2,\i*\arraysep+\i*\arrayh-\arrayh/2-\arraysep) {\small $\countervalue$};
                \draw (\arrayw+\cmssep+\j*\counterw,\i*\arraysep+\i*\arrayh-\arrayh-\arraysep) rectangle (\arrayw+\cmssep+\j*\counterw,\i*\arraysep+\i*\arrayh-\arraysep);
            }
        }

        % Last CMS
        \def\lastsketcharrayw{4*\arrayw}
        \def\leftcountercount{4}
        \def\rightcountercount{4}
        \def\counterarray{{3,5,4,8,3,5,4,8,6,3,6,5,6,3,6,5,9,1,8,2,9,1,8,2}}
        \foreach \i in {1, ..., \arraycount} {
            \draw (3*\arrayw+2*\cmssep,\i*\arraysep+\i*\arrayh-\arrayh-\arraysep) rectangle (3*\arrayw+\lastsketcharrayw+2*\cmssep,\i*\arraysep+\i*\arrayh-\arraysep);
            \foreach \j in {1, ..., \leftcountercount} {
                \draw (3*\arrayw+2*\cmssep+\j*\counterw,\i*\arraysep+\i*\arrayh-\arrayh-\arraysep) -- (3*\arrayw+2*\cmssep+\j*\counterw,\i*\arraysep+\i*\arrayh-\arraysep);
                \pgfmathtruncatemacro{\counterindex}{8*\i+\j-9}
                \pgfmathparse{\counterarray[\counterindex]}
                \let\countervalue\pgfmathresult
                \node[inner sep=0pt] at (3*\arrayw+2*\cmssep+\j*\counterw-\counterw/2,\i*\arraysep+\i*\arrayh-\arrayh/2-\arraysep) {\small \countervalue};
            }
            \foreach \j in {1, ..., \rightcountercount} {
                \draw (3*\arrayw+\lastsketcharrayw+2*\cmssep-\j*\counterw,\i*\arraysep+\i*\arrayh-\arrayh-\arraysep) -- (3*\arrayw+\lastsketcharrayw+2*\cmssep-\j*\counterw,\i*\arraysep+\i*\arrayh-\arraysep);
                \pgfmathtruncatemacro{\counterindex}{8*\i+8-\j-8}
                \pgfmathparse{\counterarray[\counterindex]}
                \let\countervalue\pgfmathresult
                \node[inner sep=0pt] at (3*\arrayw+\lastsketcharrayw+2*\cmssep-\j*\counterw+\counterw/2,\i*\arraysep+\i*\arrayh-\arrayh/2-\arraysep) {\small \countervalue};
            }
        }

        % Expansion Lines
        \draw[-stealth] plot[hobby] coordinates { (2*\arrayw+\cmssep,3*\arrayh+2.25*\arraysep) (3*\arrayw+1.5*\cmssep,3*\arrayh+3*\arraysep) (4*\arrayw+2*\cmssep,3*\arrayh+2.25*\arraysep) };
        \draw[-stealth] plot[hobby] coordinates { (2*\arrayw+\cmssep,3*\arrayh+2.25*\arraysep) (3*\arrayw+1.5*\cmssep,3*\arrayh+3.5*\arraysep) (6*\arrayw+2*\cmssep,3*\arrayh+2.25*\arraysep) };
        \node[inner sep=0pt] (expansion_title) at (3*\arrayw+1.5*\cmssep,3*\arrayh+4.25*\arraysep) {\small Expansion};

        \node[inner sep=1pt,anchor=north] (first_expansion_threshold) at (\arrayw+0.5*\cmssep,-\arraysep/2) {\small $N > 4$};
        \draw[densely dotted] (\arrayw+0.5*\cmssep,3*\arrayh+3*\arraysep) -- (first_expansion_threshold.north);

        \node[inner sep=1pt,anchor=north] (second_expansion_threshold) at (3*\arrayw+1.5*\cmssep,-\arraysep/2) {\small $N > 16$};
        \draw[densely dotted] (3*\arrayw+1.5*\cmssep,3*\arrayh+3*\arraysep) -- (second_expansion_threshold.north);

        % Sketch Titles
        \def\sketchtitley{4.75*\arrayh+4.25*\arraysep}
        %\node[inner sep=0pt] (sketch_0_label) at (\arrayw/2,\sketchtitley) {\small Sketch 0};
        %\node[inner sep=0pt] (sketch_1_label) at (2*\arrayw+\cmssep,\sketchtitley) {\small Sketch 1};
        %\node[inner sep=0pt] (sketch_2_label) at (5*\arrayw+2*\cmssep,\sketchtitley) {\small Sketch 2};

        \ifappendix
            \node[inner sep=1pt,anchor=north] (third_expansion_threshold) at (7*\arrayw+3.7*\cmssep,-\arraysep/2) {\small $N > 64$};
            \draw[densely dotted] (7*\arrayw+3.7*\cmssep,3*\arrayh+3*\arraysep) -- (third_expansion_threshold.north);
        \else
            \node[inner sep=1pt,anchor=north] (third_expansion_threshold) at (7*\arrayw+7.0*\cmssep,-\arraysep/2) {\small $N > 64$};
            \draw[densely dotted] (7*\arrayw+7.0*\cmssep,3*\arrayh+3*\arraysep) -- (third_expansion_threshold.north);
        \fi

        \def\texty{\arraycount*\arraysep+\arraycount*\arrayh+\cmstitlesep+0.2}
        \def\arrowoffsetx{3*\arrayw+1.85*\cmssep}
        \def\arrowoffsety{\arraysep/2}
        \ifappendix
            \def\querytextx{3*\arrayw+\lastsketcharrayw+2.4*\cmssep}
        \else
            \def\querytextx{3*\arrayw+\lastsketcharrayw+4.0*\cmssep}
        \fi
        \node[inner sep=1pt,align=center] (query) at (\querytextx,\texty) {\small Query~$q$ \\[-3pt] \small Result $=$ 5};
        \draw[-stealth] (\querytextx,\arrayh+\arrowoffsety) -- node[pos=0.5,below=1pt,rotate=90,inner sep=2pt] {\small Min} (query.south);

        \def\posarray{{0,6,3,1}}
        \def\arrayx{3*\arrayw+2*\cmssep}
        \foreach \i in {1, ..., \arraycount} {
            \def\arrayy{\i*\arrayh+\i*\arraysep-\arrayh-\arraysep}
            \pgfmathsetmacro{\pos}{\posarray[\i]}
            \draw (\arrayx+\lastsketcharrayw-\pos*\counterw-\counterw/2,\arrayy+\arrayh) |- (\querytextx,\arrayy+\arrayh+\arrowoffsety);
            \begin{pgfonlayer}{background layer}
                \draw[fill=gray!20] (\arrayx+\lastsketcharrayw-\pos*\counterw-\counterw,\arrayy) rectangle (\arrayx+\lastsketcharrayw-\pos*\counterw,\arrayy+\arrayh);
            \end{pgfonlayer}
        }
    \end{tikzpicture}
    \vspace{-1.4mm}
    \ifnoappendix
        \vspace{-0.6mm}
    \fi
    \caption{\sketchcms expands by concatenating each array with a copy of
    itself. It maps a key to a counter in an array by taking the appropriate
    number of bits from its hash.}
    \vspace{-5mm}
    \label{fig:sketch_chain}
\end{figure}

\subsection{Accommodating Unbounded Growth}\label{sec:accommodating_unknown_stream_lengths}
We now address the problem of maintaining bounded estimation error for
frequency estimation sketches over growing streams. Recall from
\Cref{sec:problem_analysis} that processing
growing streams with existing FE sketches (e.g., CMS) yields an estimation
error that increases at least linearly with the total key count~$N$. Since~$N$
is typically unknown in advance, bounding this error requires overprovisioning
memory upfront, leading to significant waste. \sketch\ addresses this
limitation by generalizing an FE sketch to increase its number of counters as
the stream grows. We also show how to support deletions and contractions
without deteriorating accuracy.

In the case of CMS, \sketchcms expands whenever~$N$ exceeds a tunable
threshold. An expansion concatenates the current CMS with a copy of itself,
doubling the size of each array. This operation is inexpensive since it copies
counters without needing to rehash the keys. These larger arrays reduce the
collision probability of future insertions. \Cref{fig:sketch_chain}
illustrates
\sketchcms just after the second expansion took place (on the right), as well
as its state just before the first expansion (left) and the second (middle).
Each half of an array in the largest sketch is a copy of the corresponding
array just before the second expansion. The counters in each array before this
expansion sum to the expansion threshold of~16, plus the total value of the
counters copied from the previous expansion, i.e., 4. Queries follow the same
algorithm as CMS, targeting the newly expanded, largest sketch. \sketchcms
maintains~$d$ arrays across expansions to ensure a stable confidence bound on
the estimation errors while retaining the same~$O(d)$ time complexity~as~CMS.

\textbf{Consistent Hash Functions.}
\sketchcms uses the same hash function for each of its arrays across
expansions. It hashes a key to a counter within an array of size~$w$ by taking
the~$\log_2 w$ lower-order bits of its hash. This ensures that after an
expansion, a key hashes to either the same counter it would have hashed to
before or its copy.\footnote{Consistent hashing also grants \sketchcms
mergeability: the ability to merge sketches constructed on distinct streams to
represent their union without needing to rescan them.} For example, the queried
key~$q$ in \Cref{fig:sketch_chain}
has~$h_1(q)=011\dots$
as its hash. Since the size of the array before expansion was~4, we only needed
the first~$\log_2 4=2$ bits of $q$'s hash, i.e., 01, to map it to the 2nd
counter in the first array. After expansion, we use the first~$\log_2 8=3$ bits
of the hash, i.e., 011, to map it to the~6th counter, which is a copy of the
2nd counter before expansion.

\textbf{Initial Size.}
By default, \sketchcms initializes each of its~$d$ arrays to a single chunk
(i.e., a cache line) containing~$c$ counters, as described in
\Cref{sec:accommodating_skew}. Thus, it
avoids wasting memory from the onset in case the stream grows slowly.

\ifappendix
\begin{figure}
    \centering
    \pgfdeclarelayer{background layer}
    \pgfsetlayers{background layer,main}
    \begin{tikzpicture}
        \def\arraycount{1}
        \def\arrayw{2}
        \def\arrayh{0.25}
        \def\arraysep{0.1}
        \def\cmssep{1.0}
        \def\cmstitley{0.7}
        \def\insertdeletesep{1.45}
        \def\hashsep{\insertdeletesep-0.4}
        \def\arrowsep{0.2*\cmssep}
        \def\smallhashbracesep{0}
        \def\largehashbracesep{0.3}
        \def\midy{0.5*\arraycount*\arraysep-0.5*\arraysep+0.5*\arraycount*\arrayh}
        \def\countercount{4} % Have to change this thing manually in the for loops, sadly
        \def\counterw{\arrayw/\countercount}

        % Earlier Sketch
        \def\counterarray{{0,9,1,8,2}}
        \foreach \i in {1, ..., \arraycount} {
            \draw (0,\i*\arraysep+\i*\arrayh-\arrayh-\arraysep) rectangle (\arrayw,\i*\arraysep+\i*\arrayh-\arraysep);
            \foreach \j in {1, ..., 4} {
                \pgfmathsetmacro{\countervalue}{\counterarray[\j]}
                \node[inner sep=0pt] at (\j*\counterw-\counterw/2,\i*\arraysep+\i*\arrayh-\arrayh/2-\arraysep) {\small \countervalue};
                \draw (\j*\counterw,\i*\arraysep+\i*\arrayh-\arrayh-\arraysep) rectangle (\j*\counterw,\i*\arraysep+\i*\arrayh-\arraysep);
            }
        }
        \node[inner sep=0pt,align=center] (cms_penultimate) at (0.5*\arrayw-0.265,\cmstitley) {\small 1) Earlier FE Sketch \\[-4pt] \small \hspace{18pt} Before Expansion};

        % Current Sketch
        \def\counterarray{{0,8,1,4,1,8,3,7,2}}
        \foreach \i in {1, ..., \arraycount} {
            \draw (\arrayw+\cmssep,\i*\arraysep+\i*\arrayh-\arrayh-\arraysep) rectangle (2*\arrayw+\cmssep,\i*\arraysep+\i*\arrayh-\arraysep);
            \draw[fill=gray!20] (2*\arrayw+\cmssep,\i*\arraysep+\i*\arrayh-\arrayh-\arraysep) rectangle (3*\arrayw+\cmssep,\i*\arraysep+\i*\arrayh-\arraysep);
            \foreach \j in {1, ..., 8} {
                \pgfmathtruncatemacro{\countervalue}{\counterarray[\j]}
                \node[inner sep=0pt] at (\arrayw+\cmssep+\j*\counterw-\counterw/2,\i*\arraysep+\i*\arrayh-\arrayh/2-\arraysep) {\small $\countervalue$};
                \draw (\arrayw+\cmssep+\j*\counterw,\i*\arraysep+\i*\arrayh-\arrayh-\arraysep) rectangle (\arrayw+\cmssep+\j*\counterw,\i*\arraysep+\i*\arrayh-\arraysep);
            }
        }
        \draw[densely dotted] (2*\arrayw+\cmssep,\arrayh+\arraysep) -- (2*\arrayw+\cmssep,-9*\arraysep);
        \node[inner sep=0pt,align=center] (cms_last) at (2*\arrayw+\cmssep-0.2,\cmstitley) {\small 2) Current FE Sketch \\[-4pt] \small \hspace{20pt} Before Contraction};

        % Extraction
        \def\extractiony{-2.25*\arraysep}
        \draw[-stealth] (\arrayw,\arrayh/2) -| (\arrayw+0.6*\cmssep,\extractiony) -- (\arrayw+0.80*\cmssep,\extractiony);
        \def\counterarray{{0,9,1,8,2,9,1,8,2}}
        \node[inner sep=0pt] at (\arrayw+0.95*\cmssep,\extractiony-0.04) {\small $-$};
        \foreach \i in {1, ..., 8} {
            \pgfmathtruncatemacro{\countervalue}{\counterarray[\i]}
            \node[inner sep=0pt] at (\arrayw+\cmssep+\i*\counterw-\counterw/2,\extractiony) {\small $\countervalue$};
        }
        \def\resulty{-6.5*\arraysep}
        \draw (\arrayw+0.97*\cmssep,\extractiony/2+\resulty/2) -- (3*\arrayw+1.03*\cmssep,\extractiony/2+\resulty/2);
        \def\counterarray{{0,-1,0,-4,-1,-1,2,-1,0}}
        \foreach \i in {1, ..., 8} {
            \pgfmathtruncatemacro{\countervalue}{\counterarray[\i]}
            \node[inner sep=0pt] at (\arrayw+\cmssep+\i*\counterw-\counterw/2,\resulty) {\small $\countervalue$};
        }
        \node[inner sep=0pt] (updates_label) at (3*\arrayw+1.415*\cmssep,\resulty) {\small Deltas};

        % Transfer
        \def\resulty{-8.0*\arraysep}
        \draw[decorate,decoration={brace,raise=1pt,amplitude=2pt,mirror}] (\arrayw+\cmssep,\resulty) -- (2*\arrayw+\cmssep,\resulty)
            node[pos=0.5,below=1.5pt,inner sep=0pt] (last_sketch_left_half) { };
        \def\addy{-2.25*\arraysep}
        \draw[-stealth] (last_sketch_left_half) |- (\arrayw+0.4*\cmssep,\resulty-3.0*\arraysep) |- (\arrayw+0.05*\cmssep,\addy);
        \node[inner sep=0pt] at (-\counterw/10,\addy-0.0175) {\small $+$};
        \foreach \j in {1, ..., 4} {
            \pgfmathtruncatemacro{\countervalue}{\counterarray[\j]}
            \node[inner sep=0pt] at (\j*\counterw-\counterw/2,\addy) {\small $\countervalue$};
        }

        \draw[decorate,decoration={brace,raise=1pt,amplitude=2pt,mirror}] (2*\arrayw+\cmssep,\resulty) -- (3*\arrayw+\cmssep,\resulty)
            node[pos=0.5,below=1.5pt,inner sep=0pt] (last_sketch_right_half) { };
        \def\addy{-6.0*\arraysep}
        \draw[-stealth] (last_sketch_right_half) |- (\arrayw+0.2*\cmssep,\resulty-5.0*\arraysep) |- (\arrayw+0.05*\cmssep,\addy);
        \node[inner sep=0pt] at (-\counterw/10,\addy-0.0175) {\small $+$};
        \foreach \i in {5, ..., 8} {
            \pgfmathtruncatemacro{\countervalue}{\counterarray[\i]}
            \node[inner sep=0pt] at (\i*\counterw-4.5*\counterw,\addy) {\small $\countervalue$};
        }

        \def\resulty{-9.5*\arraysep}
        \draw (-0.03*\cmssep,\addy/2+\resulty/2) -- (\arrayw+0.03*\cmssep,\addy/2+\resulty/2);

        % New Current Sketch
        \def\newlastsketchy{-\arrayh-9.0*\arraysep}
        \draw (0,\newlastsketchy) rectangle (\arrayw,\newlastsketchy+\arrayh);
        \def\counterarray{{0,7,3,3,1}}
        \foreach \i in {1, ..., 4} {
            \ifnum \i = 4 {
            }
            \else {
                \draw (\i*\counterw,\newlastsketchy) -- (\i*\counterw,\newlastsketchy+\arrayh);
            }
            \fi
            \pgfmathtruncatemacro{\countervalue}{\counterarray[\i]}
            \node[inner sep=0pt] at (\i*\counterw-\counterw/2,\newlastsketchy+\arrayh/2) {\small $\countervalue$};
        }
        \node[inner sep=0pt,align=center] (cms_last_new) at (\arrayw/2-0.2,\newlastsketchy-1.5*\arrayh) {\small 3) Current FE Sketch \\[-4pt] \small \hspace{14pt} After Contraction};
    \end{tikzpicture}
    \caption{During contractions, \sketchcms extracts the recent updates by
    subtracting the previous FE sketch. It adds the results to the previous
    sketch to transfer the updates.}
    \vspace{-5mm}
    \label{fig:contraction}
\end{figure}
\fi

\textbf{Expansion Thresholds.} 
By default, \sketchcms expands whenever~$N$ quadruples. As shown at the bottom
of
\Cref{fig:sketch_chain},
expansions occur when the number of keys~$N$ reaches 4, 16, and 64. These
thresholds cause the number of counters in each array to grow as~$\sqrt{N}$. As
a consequence, we achieve a sublinear memory footprint and sublinear errors. To
see why, recall from \Cref{sec:background} that a CMS
allocated with arrays of~$w$ counters has an expected error of~$N/w$. By
making~$w$ grow as~$\sqrt{N}$, \sketchcms instead obtains an expected error
of~$O(N/\sqrt{N})=O(\sqrt{N})$ even when the stream's length is unknown.
Although keys inserted when the sketch was smaller contribute disproportionally
to the error due to their counts duplicating with expansions, they are
exponentially fewer in number. Hence, such keys only increase the error by a
small constant factor.
\Cref{thm:cms-w}, proven in
\ifappendix
\Cref{sec:upper_bounds_cms},
\else
the Appendix~\cite{SublimeArxiv},
\fi
formalizes these intuitions. In sum, compared to a standard CMS with linear
errors and a fixed memory footprint, \sketchcms achieves significantly lower
errors at a slightly higher space cost. As such, it is a better fit for
applications that require maintaining low errors as the stream grows
indefinitely. We shortly show how to set the expansion thresholds to open a new
accuracy-memory tradeoff space.

\textbf{Deletions and Contractions.} 
\sketchcms handles deletions using the same procedure as CMS. It contracts
after many deletions to save memory. The simplest way to contract is to revert
an expansion by folding the halves of each array onto each other and adding up
the counters. Yet, this approach would inflate the error since it adds the
counters copied during the last expansion to the originals, thereby doubling
them. Instead, \sketchcms keeps a record of the sketch's state just before each
expansion. It extracts the recent updates made to the current sketch by
subtracting the state of the sketch before the last expansion from its left and
right halves. We illustrate this for a single array in
\Cref{fig:contraction}.
The result is an array of deltas indicating how each counter from the earlier
sketch and its copy have changed. \sketchcms applies these deltas to the
original counters in the earlier sketch. Finally, it discards the current
sketch and uses the earlier sketch to handle future~operations.

\ifnoappendix
\begin{figure}
    \centering
    \pgfdeclarelayer{background layer}
    \pgfsetlayers{background layer,main}
    \begin{tikzpicture}
        \def\arraycount{1}
        \ifappendix
            \def\arrayw{2}
            \def\arrayh{0.25}
            \def\arraysep{0.1}
            \def\cmssep{1.0}
            \def\cmstitley{0.7}
        \else
            \def\arrayw{2.5}
            \def\arrayh{0.3}
            \def\arraysep{0.1}
            \def\cmssep{2.0}
            \def\cmstitley{0.8}
        \fi
        \def\insertdeletesep{1.45}
        \def\hashsep{\insertdeletesep-0.4}
        \def\arrowsep{0.2*\cmssep}
        \def\smallhashbracesep{0}
        \def\largehashbracesep{0.3}
        \def\midy{0.5*\arraycount*\arraysep-0.5*\arraysep+0.5*\arraycount*\arrayh}
        \def\countercount{4} % Have to change this thing manually in the for loops, sadly
        \def\counterw{\arrayw/\countercount}

        % Earlier Sketch
        \def\counterarray{{0,9,1,8,2}}
        \foreach \i in {1, ..., \arraycount} {
            \draw (0,\i*\arraysep+\i*\arrayh-\arrayh-\arraysep) rectangle (\arrayw,\i*\arraysep+\i*\arrayh-\arraysep);
            \foreach \j in {1, ..., 4} {
                \pgfmathsetmacro{\countervalue}{\counterarray[\j]}
                \node[inner sep=0pt] at (\j*\counterw-\counterw/2,\i*\arraysep+\i*\arrayh-\arrayh/2-\arraysep) {\small \countervalue};
                \draw (\j*\counterw,\i*\arraysep+\i*\arrayh-\arrayh-\arraysep) rectangle (\j*\counterw,\i*\arraysep+\i*\arrayh-\arraysep);
            }
        }
        \node[inner sep=0pt,align=center] (cms_penultimate) at (0.5*\arrayw-0.265,\cmstitley) {\small 1) Earlier FE Sketch \\[-4pt] \small \hspace{9pt} Before Expansion};

        % Current Sketch
        \def\counterarray{{0,8,1,4,1,8,3,7,2}}
        \foreach \i in {1, ..., \arraycount} {
            \draw (\arrayw+\cmssep,\i*\arraysep+\i*\arrayh-\arrayh-\arraysep) rectangle (2*\arrayw+\cmssep,\i*\arraysep+\i*\arrayh-\arraysep);
            \draw[fill=gray!20] (2*\arrayw+\cmssep,\i*\arraysep+\i*\arrayh-\arrayh-\arraysep) rectangle (3*\arrayw+\cmssep,\i*\arraysep+\i*\arrayh-\arraysep);
            \foreach \j in {1, ..., 8} {
                \pgfmathtruncatemacro{\countervalue}{\counterarray[\j]}
                \node[inner sep=0pt] at (\arrayw+\cmssep+\j*\counterw-\counterw/2,\i*\arraysep+\i*\arrayh-\arrayh/2-\arraysep) {\small $\countervalue$};
                \draw (\arrayw+\cmssep+\j*\counterw,\i*\arraysep+\i*\arrayh-\arrayh-\arraysep) rectangle (\arrayw+\cmssep+\j*\counterw,\i*\arraysep+\i*\arrayh-\arraysep);
            }
        }
        \draw[densely dotted] (2*\arrayw+\cmssep,\arrayh+\arraysep) -- (2*\arrayw+\cmssep,-9*\arraysep);
        \node[inner sep=0pt,align=center] (cms_last) at (2*\arrayw+\cmssep-0.2,\cmstitley) {\small 2) Current FE Sketch \\[-4pt] \small \hspace{11pt} Before Contraction};

        % Extraction
        \def\extractiony{-2.25*\arraysep}
        \draw[-stealth] (\arrayw,\arrayh/2) -| (\arrayw+0.6*\cmssep,\extractiony) -- (\arrayw+0.80*\cmssep,\extractiony);
        \def\counterarray{{0,9,1,8,2,9,1,8,2}}
        \node[inner sep=0pt] at (\arrayw+0.95*\cmssep,\extractiony-0.04) {\small $-$};
        \foreach \i in {1, ..., 8} {
            \pgfmathtruncatemacro{\countervalue}{\counterarray[\i]}
            \node[inner sep=0pt] at (\arrayw+\cmssep+\i*\counterw-\counterw/2,\extractiony) {\small $\countervalue$};
        }
        \def\resulty{-6.5*\arraysep}
        \draw (\arrayw+0.97*\cmssep,\extractiony/2+\resulty/2) -- (3*\arrayw+1.03*\cmssep,\extractiony/2+\resulty/2);
        \def\counterarray{{0,-1,0,-4,-1,-1,2,-1,0}}
        \foreach \i in {1, ..., 8} {
            \pgfmathtruncatemacro{\countervalue}{\counterarray[\i]}
            \node[inner sep=0pt] at (\arrayw+\cmssep+\i*\counterw-\counterw/2,\resulty) {\small $\countervalue$};
        }
        \ifappendix
            \node[inner sep=0pt] (updates_label) at (3*\arrayw+1.415*\cmssep,\resulty) {\small Deltas};
        \else
            \node[inner sep=0pt] (updates_label) at (3*\arrayw+1.3*\cmssep,\resulty) {\small Deltas};
        \fi

        % Transfer
        \def\resulty{-8.0*\arraysep}
        \draw[decorate,decoration={brace,raise=1pt,amplitude=2pt,mirror}] (\arrayw+\cmssep,\resulty) -- (2*\arrayw+\cmssep,\resulty)
            node[pos=0.5,below=1.5pt,inner sep=0pt] (last_sketch_left_half) { };
        \def\addy{-2.25*\arraysep}
        \draw[-stealth] (last_sketch_left_half) |- (\arrayw+0.4*\cmssep,\resulty-3.0*\arraysep) |- (\arrayw+0.05*\cmssep,\addy);
        \node[inner sep=0pt] at (-\counterw/10,\addy-0.0175) {\small $+$};
        \foreach \j in {1, ..., 4} {
            \pgfmathtruncatemacro{\countervalue}{\counterarray[\j]}
            \node[inner sep=0pt] at (\j*\counterw-\counterw/2,\addy) {\small $\countervalue$};
        }

        \draw[decorate,decoration={brace,raise=1pt,amplitude=2pt,mirror}] (2*\arrayw+\cmssep,\resulty) -- (3*\arrayw+\cmssep,\resulty)
            node[pos=0.5,below=1.5pt,inner sep=0pt] (last_sketch_right_half) { };
        \def\addy{-6.0*\arraysep}
        \draw[-stealth] (last_sketch_right_half) |- (\arrayw+0.2*\cmssep,\resulty-5.0*\arraysep) |- (\arrayw+0.05*\cmssep,\addy);
        \node[inner sep=0pt] at (-\counterw/10,\addy-0.0175) {\small $+$};
        \foreach \i in {5, ..., 8} {
            \pgfmathtruncatemacro{\countervalue}{\counterarray[\i]}
            \node[inner sep=0pt] at (\i*\counterw-4.5*\counterw,\addy) {\small $\countervalue$};
        }

        \def\resulty{-9.5*\arraysep}
        \draw (-0.03*\cmssep,\addy/2+\resulty/2) -- (\arrayw+0.03*\cmssep,\addy/2+\resulty/2);

        % New Current Sketch
        \def\newlastsketchy{-\arrayh-9.0*\arraysep}
        \draw (0,\newlastsketchy) rectangle (\arrayw,\newlastsketchy+\arrayh);
        \def\counterarray{{0,7,3,3,1}}
        \foreach \i in {1, ..., 4} {
            \ifnum \i = 4 {
            }
            \else {
                \draw (\i*\counterw,\newlastsketchy) -- (\i*\counterw,\newlastsketchy+\arrayh);
            }
            \fi
            \pgfmathtruncatemacro{\countervalue}{\counterarray[\i]}
            \node[inner sep=0pt] at (\i*\counterw-\counterw/2,\newlastsketchy+\arrayh/2) {\small $\countervalue$};
        }
        \node[inner sep=0pt,align=center] (cms_last_new) at (\arrayw/2-0.2,\newlastsketchy-1.5*\arrayh) {\small 3) Current FE Sketch \\[-4pt] \small \hspace{5pt} After Contraction};
    \end{tikzpicture}
    \vspace{-1mm}
    \caption{During contractions, \sketchcms extracts the recent updates by
    subtracting the previous FE sketch. It adds the results to the previous
    sketch to transfer the updates.}
    \label{fig:contraction}
    \vspace{-4mm}
\end{figure}
\fi

The contraction threshold is set to the average of the two previous expansion
thresholds to avoid the problem of updates frequently resizing the structure.
As such, a contraction is triggered when the majority of the stream's keys are
deleted. Therefore, \sketchcms always handles most of the stream using arrays
that are at most twice as large as~$\sqrt{N}$. Thus, the expected error and the
size of each array remain the same as before, despite the state of the sketch
depending on the sequence of insertions and deletions.

Maintaining the sketch record increases the memory footprint by at most a
factor of~$2$ compared to allocating and expanding a single sketch. When the
total key count fluctuates, the memory saved by contracting outweighs this
cost.

\textbf{Rich Space of Tradeoffs.}
By varying the expansion thresholds, we expose a rich spectrum of
accuracy-memory tradeoffs. We encode these policies using a \emph{Size
Function}, denoted as~$W(N)$, which captures the growth rate of the number of
counters in each array with respect to~$N$. The faster the size function grows
with respect to~$N$, the sooner \sketchcms expands, thereby tracking more of
the stream's keys accurately in exchange for a larger memory footprint. By
default, the size function is set to~$W(N)=\sqrt{N}$. We generalize it to a
\emph{Sublinear Power} of the form~$W(N)=N^\alpha/\epsilon$, with
constants~$\alpha \in [0,1)$ and $\epsilon > 0$. This corresponds to
having~$1/\epsilon$ chunks in each array at the beginning and expanding
when~$N$ grows by a factor of~$2^{1/\alpha}$. Note that this factor is greater
than~two.

\ifappendix
\begin{figure}
    \centering
    \begin{tikzpicture}
        \def\arraycount{3}
        \def\arrayw{5}
        \def\arrayh{0.25}
        \def\arraysep{0.6}
        \def\arraylen{10}
        \def\counterw{0.5}

        \foreach \i in {1, ..., \arraycount} {
            \draw[black] (0,-\i*\arrayh+1.5*\arrayh-\i*\arraysep+\arraysep) rectangle (\arrayw,-\i*\arrayh+0.5*\arrayh-\i*\arraysep+\arraysep);
            \foreach \j in {2, ..., \arraylen} {
                \draw (\j*\counterw-\counterw,-\i*\arrayh+1.5*\arrayh-\i*\arraysep+\arraysep) -- (\j*\counterw-\counterw,-\i*\arrayh+0.5*\arrayh-\i*\arraysep+\arraysep);
            }
        }

        \def\texty{5*\arrayh}
        \def\inserttextx{-\counterw}
        \node[inner sep=1pt,align=center] (key) at (\inserttextx,\texty) {\small Insert $x$};
        \draw (key.south) -- (\inserttextx,-1.5*\arrayh-1.7*\arraysep);

        \def\pa{2.0}
        \def\pb{3.0}
        \def\pc{1.0}
        \draw[-stealth] (\inserttextx,0.5*\arrayh+0.3*\arraysep) -| (\pa*\counterw+0.5*\counterw,0.5*\arrayh) node[pos=0.22,above=2pt,align=center,inner sep=0pt] {\footnotesize $h_1(x)$, $+\sigma_1(x)$};
        \draw[-stealth] (\inserttextx,-0.5*\arrayh-0.7*\arraysep) -| (\pb*\counterw+0.5*\counterw,-0.5*\arrayh-1.0*\arraysep) node[pos=0.171,above=2pt,align=center,inner sep=0pt] {\footnotesize $h_2(x)$, $+\sigma_2(x)$};
        \draw[-stealth] (\inserttextx,-1.5*\arrayh-1.7*\arraysep) -| (\pc*\counterw+0.5*\counterw,-1.5*\arrayh-2.0*\arraysep) node[pos=0.206,above=2pt,align=center,inner sep=0pt] {\footnotesize \hspace{4.5mm} $h_3(x)$, $+\sigma_3(x)$};

        \def\querytextx{\arrayw+1*\counterw}
        \node[inner sep=1pt,align=center] (query) at (\querytextx,\texty) {\small Query $q$ \\[-2pt] \small Result $=2$};
        \draw[stealth-] (query.south) -- (\querytextx,-1.5*\arrayh-1.7*\arraysep) node[pos=0.5,below=2pt,rotate=90,inner sep=2pt] {\small Median};

        \def\pa{7.0}
        \def\pb{8.0}
        \def\pc{6.0}
        \draw (\querytextx,0.5*\arrayh+0.3*\arraysep) -| (\pa*\counterw+0.5*\counterw,0.5*\arrayh) node[pos=0.300,above=2pt,align=center,inner sep=0pt] {\footnotesize $h_1(q)$, $\times \sigma_1(q)=-1$};
        \draw (\querytextx,-0.5*\arrayh-0.7*\arraysep) -| (\pb*\counterw+0.5*\counterw,-0.5*\arrayh-1.0*\arraysep) node[pos=0.291,above=2pt,align=center,inner sep=0pt] {\footnotesize $h_2(q)$, $\times \sigma_2(q)=1$ \hspace{4mm}};
        \draw (\querytextx,-1.5*\arrayh-1.7*\arraysep) -| (\pc*\counterw+0.5*\counterw,-1.5*\arrayh-2.0*\arraysep) node[pos=0.213,above=2pt,align=center,inner sep=0pt] {\footnotesize $h_3(q)$, $\times \sigma_3(q)=1$};

        \draw[fill=gray!20] (\pa*\counterw,0.5*\arrayh) rectangle (\pa*\counterw+\counterw,-0.5*\arrayh);
        \draw[fill=gray!20] (\pb*\counterw,-0.5*\arrayh-1.0*\arraysep) rectangle (\pb*\counterw+\counterw,-1.5*\arrayh-1.0*\arraysep);
        \draw[fill=gray!20] (\pc*\counterw,-1.5*\arrayh-2.0*\arraysep) rectangle (\pc*\counterw+\counterw,-2.5*\arrayh-2.0*\arraysep);
        \node[inner sep=0pt] at (\pa*\counterw+0.5*\counterw,-0.0*\arrayh-0.0*\arraysep) {\footnotesize $-2$};
        \node[inner sep=0pt] at (\pb*\counterw+0.5*\counterw,-1.0*\arrayh-1.0*\arraysep) {\footnotesize $9$};
        \node[inner sep=0pt] at (\pc*\counterw+0.5*\counterw,-2.0*\arrayh-2.0*\arraysep) {\footnotesize $-4$};

        %\node[inner sep=0pt,anchor=west] at (-3.1*\counterw,3.0*\arrayh) {\textbf{B) CS}};
    \end{tikzpicture}
    \caption{CS assigns a random update direction to each key and uses it to
    insert them into~$d$ counter arrays of size~$w$. The random update
    directions enable unbiased estimation.}
    \label{fig:cs_overview}
\end{figure}
\fi

When using a sublinear power as the size function, \sketchcms achieves an
expected error of~$O(N/W(N))=O(\epsilon \cdot N^{1-\alpha})$.\footnote{The
constant factor hidden by the asymptotics is~$(2^{1/\alpha}-1) \cdot
\frac{1-2^{(\alpha-1) \cdot \log_2 N}}{2^{1/\alpha-1}-1}$, which is at
most~$\frac{2^{1/\alpha}-1}{2^{1/\alpha-1}-1}$ and converges to~$\log_2 N$ as
$\alpha$ tends to 1.} One may hope to achieve a constant expected error by
growing the number of counters linearly with the data size, i.e., by using a
linear size function of the form~$W(N)=N/\epsilon$. Yet, \sketchcms obtains an
error bound of~$\log_2 N \cdot \epsilon$ in this case.
\Cref{sec:memory_analysis} shows that without
knowing the ultimate length of the stream, these expected error bounds are the
best achievable by any FE sketch of (approximately) the same size as
\sketchcms. The following theorem, proven in \ifappendix
\Cref{sec:upper_bounds_cms}, \else
the Appendix~\cite{SublimeArxiv},
\fi
formalizes the above error bounds:\footnote{
\ifappendix
In \Cref{sec:mg}, we show that
similar bounds hold for \sketch\ when it is applied to Misra-Gries.
\else
In the Appendix~\cite{SublimeArxiv}, we show that similar bounds hold for
\sketch\ when it is applied to Misra-Gries.
\fi
}

\begin{restatable}{theorem}{thmcms}
    \label{thm:cms-w}
    When using a size function~$W(\cdot)$ and a single array on a stream with a
    total key count of~$N$, \sketchcms provides an estimation error~$E(N)$
    satisfying
	$$
    \expectation{E(N)} \leq  
    \begin{cases}
        O(1) \cdot \epsilon N^{1-\alpha} & \text{if $W(N)=N^{\alpha}/\epsilon$, for $\alpha \in [0,1)$,} \\
		\log_2 N \cdot \epsilon & \text{if $W(N)=N/\epsilon$.}
	\end{cases}
    $$
    Compared to a fixed-size CMS allocated upfront with an array of~$W(N)$
    counters with knowledge of~$N$, the expected error terms above are higher
    by at most a constant and a logarithmic factor.
\end{restatable}
Applying Markov's inequality to the expected error bounds above yields the same
confidence bounds as CMS. Moreover, employing~$d$ independent arrays in
\sketchcms raises this confidence to the power of~$d$. These confidence bounds
remain stable as the stream grows.

\ifnoappendix
\begin{figure}
    \centering
    \begin{tikzpicture}
        \def\arraycount{3}
        \def\arrayw{5}
        \def\arrayh{0.3}
        \def\arraysep{0.25}
        \def\arraylen{10}
        \def\counterw{0.5}

        \foreach \i in {1, ..., \arraycount} {
            \draw[black] (0,-\i*\arrayh+1.5*\arrayh-\i*\arraysep+\arraysep) rectangle (\arrayw,-\i*\arrayh+0.5*\arrayh-\i*\arraysep+\arraysep);
            \foreach \j in {2, ..., \arraylen} {
                \draw (\j*\counterw-\counterw,-\i*\arrayh+1.5*\arrayh-\i*\arraysep+\arraysep) -- (\j*\counterw-\counterw,-\i*\arrayh+0.5*\arrayh-\i*\arraysep+\arraysep);
            }
        }

        \def\texty{5*\arrayh}
        \def\inserttextx{-4.5*\counterw}
        \node[inner sep=1pt,align=center] (key) at (\inserttextx,\texty) {\small Insert $x$};
        \draw (key.south) -- (\inserttextx,-1.5*\arrayh-1.5*\arraysep);

        \def\pa{2.0}
        \def\pb{3.0}
        \def\pc{1.0}
        \draw[-stealth] (\inserttextx,0.5*\arrayh+0.5*\arraysep) -| (\pa*\counterw+0.5*\counterw,0.5*\arrayh) node[pos=0.1175,above=2pt,align=center,inner sep=0pt] {\footnotesize $h_1(x)$, $+\sigma_1(x)$};
        \draw[-stealth] (\inserttextx,-0.5*\arrayh-0.5*\arraysep) -| (\pb*\counterw+0.5*\counterw,-0.5*\arrayh-1.0*\arraysep) node[pos=0.1029,above=2pt,align=center,inner sep=0pt] {\footnotesize $h_2(x)$, $+\sigma_2(x)$};
        \draw[-stealth] (\inserttextx,-1.5*\arrayh-1.5*\arraysep) -| (\pc*\counterw+0.5*\counterw,-1.5*\arrayh-2.0*\arraysep) node[pos=0.0938,above=2pt,align=center,inner sep=0pt] {\footnotesize \hspace{4.5mm} $h_3(x)$, $+\sigma_3(x)$};

        \def\querytextx{\arrayw+4.5*\counterw}
        \node[inner sep=1pt,align=center] (query) at (\querytextx,\texty) {\small Query $q$ \\[-2pt] \small Result $=2$};
        \draw[stealth-] (query.south) -- (\querytextx,-1.5*\arrayh-1.5*\arraysep) node[pos=0.5,below=2pt,rotate=90,inner sep=2pt] {\small Median};

        \def\pa{7.0}
        \def\pb{8.0}
        \def\pc{6.0}
        \draw (\querytextx,0.5*\arrayh+0.5*\arraysep) -| (\pa*\counterw+0.5*\counterw,0.5*\arrayh) node[pos=0.165,above=2pt,align=center,inner sep=0pt] {\footnotesize $h_1(q)$, $\times \sigma_1(q)=-1$};
        \draw (\querytextx,-0.5*\arrayh-0.5*\arraysep) -| (\pb*\counterw+0.5*\counterw,-0.5*\arrayh-1.0*\arraysep) node[pos=0.1367,above=2pt,align=center,inner sep=0pt] {\footnotesize $h_2(q)$, $\times \sigma_2(q)=1$ \hspace{4mm}};
        \draw (\querytextx,-1.5*\arrayh-1.5*\arraysep) -| (\pc*\counterw+0.5*\counterw,-1.5*\arrayh-2.0*\arraysep) node[pos=0.1324,above=2pt,align=center,inner sep=0pt] {\footnotesize $h_3(q)$, $\times \sigma_3(q)=1$};

        \draw[fill=gray!20] (\pa*\counterw,0.5*\arrayh) rectangle (\pa*\counterw+\counterw,-0.5*\arrayh);
        \draw[fill=gray!20] (\pb*\counterw,-0.5*\arrayh-1.0*\arraysep) rectangle (\pb*\counterw+\counterw,-1.5*\arrayh-1.0*\arraysep);
        \draw[fill=gray!20] (\pc*\counterw,-1.5*\arrayh-2.0*\arraysep) rectangle (\pc*\counterw+\counterw,-2.5*\arrayh-2.0*\arraysep);
        \node[inner sep=0pt] at (\pa*\counterw+0.5*\counterw,-0.0*\arrayh-0.0*\arraysep) {\footnotesize $-2$};
        \node[inner sep=0pt] at (\pb*\counterw+0.5*\counterw,-1.0*\arrayh-1.0*\arraysep) {\footnotesize $9$};
        \node[inner sep=0pt] at (\pc*\counterw+0.5*\counterw,-2.0*\arrayh-2.0*\arraysep) {\footnotesize $-4$};

        %\node[inner sep=0pt,anchor=west] at (-3.1*\counterw,3.0*\arrayh) {\textbf{B) CS}};
    \end{tikzpicture}
    \vspace{-2mm}
    \caption{CS assigns a random update direction to each key and uses it to
    insert them into~$d$ counter arrays of size~$w$. The random update
    directions enable unbiased estimation.}
    \vspace{-3mm}
    \label{fig:cs_overview}
\end{figure}
\fi

\begin{table}
    \centering
    \ifappendix
    \setlength{\tabcolsep}{3pt}
    \fi
    \begin{tabular}{cccccccc}
        \toprule

        \ifappendix
        \textbf{\footnotesize FE Sketch}
        & \footnotesize \textbf{Bias}
        & \footnotesize \textbf{Error Bound}
        & \footnotesize \textbf{Confidence}
        & \footnotesize \textbf{Insert}
        & \footnotesize \textbf{Delete}
        & \footnotesize \textbf{Query} \\
        \else
        \textbf{\small FE Sketch} 
        & \small \textbf{Bias}
        & \small \textbf{Error Bound}
        & \small \textbf{Confidence}
        & \small \textbf{Insert}
        & \small \textbf{Delete}
        & \small \textbf{Query} \\
        \fi

        \midrule

        \ifappendix
        \footnotesize CMS 
        & \footnotesize Over 
        & \footnotesize $\frac{e \cdot N}{w}$ 
        & \footnotesize $1-e^{-d}$ 
        & \footnotesize $O(d)$ 
        & \footnotesize $O(d)$ 
        & \footnotesize $O(d)$ \\
        \else
        \small CMS 
        & \small Over 
        & \small $\frac{e \cdot N}{w}$ 
        & \small $1-e^{-d}$ 
        & \small $O(d)$ 
        & \small $O(d)$ 
        & \small $O(d)$ \\
        \fi

        \ifappendix
        \footnotesize CS 
        & \footnotesize Unbiased 
        & \footnotesize $e \cdot \sqrt{\frac{\sum_y (f(y))^2}{w}}$ 
        & \footnotesize $\geq 1-e^{-d}$ 
        & \footnotesize $O(d)$ 
        & \footnotesize $O(d)$ 
        & \footnotesize $O(d)$ \\
        \else
        \small CS 
        & \small Unbiased 
        & \small $e \cdot \sqrt{\frac{\sum_y (f(y))^2}{w}}$ 
        & \small $\geq 1-e^{-d}$ 
        & \small $O(d)$ 
        & \small $O(d)$ 
        & \small $O(d)$ \\
        \fi

        \bottomrule
    \end{tabular}
    \caption{CMS and CS strike different tradeoffs.}
    \label{tab:classic_sketch_stats}
    \ifnoappendix
      \vspace{-8mm}
    \fi
\end{table}

\subsection{Applying \sketch\ to Count Sketch}\label{sec:cs}
As discussed in \Cref{sec:introduction}, many FE
sketches differ from CMS in their error guarantees and use cases. We now
demonstrate the generality of \sketch\ by applying it to Count Sketch to derive
\sketchcs. We also discuss applying \sketch\ to any kind of FE sketch, whether
existing or yet to be proposed.

\textbf{Count Sketch (CS).}
Recall from \Cref{sec:introduction} that unbiased
estimates equal the ground-truth frequency on average. This allows for
multiplying estimates, as required in join
size~\cite{Compass,QueryOptimizationSketches,ConvolutionSketch} and
covariance~\cite{ASCS} estimation, without the risk of compounding errors. CS
provides such unbiased estimates and reaps its benefits.\footnote{One can also
achieve unbiased estimation using CMS by, in each array, subtracting the
average of the counter values a query key does not hash to from its
counter~\cite{CountMeanMin}. Yet, the resulting estimate has a slightly higher
variance~than~that~of~CS.} \Cref{tab:classic_sketch_stats}
summarizes the properties of CS and compares it to CMS.

Similarly to CMS, the core component of CS is a counter array of size~$w$
coupled with a hash function~$h$ that maps each key to a counter. Unlike CMS,
CS further associates each key with a random ``update direction'' from the
set~$\{-1,+1\}$ by hashing it with another hash function~$\sigma$. It inserts a
new key~$x$ by incrementing or decrementing its counter according
to~$\sigma(x)$, as shown in
\Cref{fig:cs_overview}.
CS deletes~$x$ by subtracting~$\sigma(x)$ from its counter to undo its
insertion. Since each key's update direction is random, the expected
interference on~$x$'s counter from any other key is zero due to the positive
and negative possibilities cancelling out. Thus, the value of the counter along
the update direction is an estimate of~$x$'s frequency. The zero expected
interference removes bias, though it introduces two-sided errors (both under
and overestimations). 

Although the estimates are unbiased and accurate in expectation for any
workload, they may suffer from large deviations. This is because many keys
hashing to the same counter may have the same update direction, despite the
zero expected interference. Skew exacerbates this phenomenon since a heavy
hitter drastically changes the estimates of the keys hashing to the same
counter, no matter its update direction. Therefore, to bound the probability of
severe errors, we compute the estimation variance.\footnote{CS uses 4-wise
independent hash functions~$\sigma$ to bound the variance and the probability
of large errors~\cite{CountSketch}.} Denoting the frequency of key~$y$ by
$f(y)$, this variance can be expressed as~$V=\sum_y (f(y))^2/w$. This lines up
with the intuition that larger arrays make hash collisions less likely,
reducing the variance. A straight-forward application of Chebyshev's inequality
shows that the error does not exceed~$e \cdot \sqrt{V}$ with a confidence of at
least~$1-e^{-2}$.\footnote{CS's error bound is~$\sqrt{V} = O(\sqrt{N/w})$ when
the workload is more uniform due to the squared summands being small. This
bound is smaller than CMS's error, i.e., $N/w$.}

Similarly to CMS, CS boosts confidence by employing~$d$ independent arrays,
each with its own hash functions~$h_i$ and $\sigma_i$. It answers a query by
returning the median estimate, as shown in
\Cref{fig:cs_overview}.
The median is returned, as it only has large errors when most estimates err
significantly, which becomes exponentially less likely given more arrays.
Formally, by a Chernoff bound, CS guarantees an error bound of at most~$e \cdot
\sqrt{V}$ with a confidence of at least~$1-e^{-d}$.\footnote{The precise
confidence given by the Chernoff bound is~$1-e^{-(e^4/8) \cdot d}$, which is
higher than the text's simplified expression.}

\textbf{Accommodating Skew.}
To avoid the space wastage of uniformly-sized counters, \sketchcs applies
\counterencodingabbrv\ to CS. It stores the absolute value of each counter
using \counterencodingabbrv\ as before, while representing the sign of each
counter by appending a sign bit to its stub. 

\textbf{Accommodating Unbounded Growth.}
Since the estimation variance~$V$ determines CS's error bound, \sketchcs
controls the error by expanding based on the growth of the variance. That is,
it expands its number of counters, i.e., the denominator~$w$ of the variance,
to balance out the numerator~$F=\sum_{y} (f(y))^2$. This translates to
expanding when the numerator~$F$ more than doubles, which is in contrast with
how \sketchcms expands when~$N$ (the numerator of CMS's error bound) more than
doubles. We encode different expansion policies by defining the size function
based on the numerator~$F$, i.e., as~$W(F)$. The quantity~$F$ is also known as
the squared~$\ell_2$-norm of the stream.

In
\ifappendix
\Cref{sec:upper_bounds_cs},
\else 
the Appendix~\cite{SublimeArxiv},
\fi
we prove through a similar analysis to that of \sketchcms, the following
variance bounds for \sketchcs:
\begin{restatable}{theorem}{thmcs}
    \label{thm:cs-w}
    When using a size function~$W(\cdot)$ and a single array on a stream with a
    squared~$\ell_2$-norm of $F = \sum_y f(y)^2$, \sketchcs provides unbiased
    estimates with a variance~$V(F)$ satisfying
	$$
    V(F) \leq  
    \begin{cases}
        O(1) \cdot \epsilon F^{1-\alpha} & \text{if $W(F)=F^{\alpha}/\epsilon$, for $\alpha \in [0,1)$, } \\
        \log_2 F \cdot \epsilon & \text{if $W(F)=F/\epsilon$.}
	\end{cases}
    $$
    Compared to a fixed-size CS allocated upfront with an array of~$W(N)$
    counters with  knowledge of~$F$, the variance terms above are higher by at
    most a constant and a logarithmic\footnote{This logarithmic term can be
    bounded in terms of the total key count~$N$, as $F \leq N^2$ and thus
    $\log_2 F \leq 2 \cdot \log_2 N$.} factor.
\end{restatable}
Applying Chebyshev's inequality and employing~$d$ independent arrays yields
stable confidence bounds that match those of CS.

\sketchcs maps keys to consistent counters and update directions by employing
the same hash functions~$h_1,\dots,h_d$ and $\sigma_1,\dots,\sigma_d$ across
expansions. It handles insertions, deletions, and queries as in the case of a
standard CS.

\textbf{Tracking the Squared~$\ell_2$-norm~$F$.}
Similarly to \sketchcms, where we have to track the total number of keys~$N$ to
determine when to expand, \sketchcs must track the squared~$\ell_2$-norm~$F$.
Yet, one cannot simply increment a value with each insertion to track~$F$. This
is because, due to the squared summands, inserting a frequent key grows~$F$
more than inserting an infrequent key. While it is possible to use the counter
arrays within \sketchcs itself to estimate~$F$, the resulting estimate has a
confidence that is determined by the number of arrays~$d$. For very long
streams, this estimate eventually errs significantly, leading to mistimed
expansions and contractions that adversely affect the accuracy and the memory
footprint. We would like to have an estimate of the squared~$\ell_2$-norm~$F$
with very high confidence to support the entire stream. 

To this end, we employ AMS sketches, which are sketches tailored to
approximating the~$\ell_2$-norm of the stream~\cite{AMS}. Each AMS sketch is a
single counter that is updated with random update directions derived by hashing
the stream's keys, similarly to CS. Squaring the value of the counter yields an
estimate of~$F$. It is known that by employing~$4 \cdot \log_2 N$ AMS sketches,
one can estimate~$F$ within a factor of 2 with a high confidence
of~$1-1/N$~\cite{AMS}, which allows for supporting the entire stream. Since~$N$
is at most the maximum value representable in a 64-bit machine word in
practice, \sketchcs employs~64 AMS sketches to track~$F$. It uses SIMD to
quickly update\footnote{More generally, it is possible to maintain~$k$ AMS
sketches, where~$k$ is the machine word size, in constant time with high
probability to approximate~$F$. Since~$k=\Omega(\log n)$ is standard in both
theory and practice, it is possible to maintain a high confidence estimate
of~$F$ in constant time without the use of SIMD. We omit the corresponding
algorithm and its proof of correctness due to space constraints.} and query
these sketches after each insertion or deletion.

\textbf{Applying \sketch\ to Other FE Sketches.}
Using the same strategy as above, \sketch\ can be adapted to any FE sketch with
an error bound of the form~$X/w$, where~$X$ is a measurable parameter of the
workload that the error depends on, and~$w$ is the FE sketch's size. For
example, in the case of CMS and Misra-Gries, $X$ is the stream's total key
count~$N$, whereas in the case of CS, it is the stream's squared~$\ell_2$-norm
$F$. In all these cases, $w$ is the number of counters within the FE sketch. By
expanding when the quantity~$X$ more than doubles and defining the size
function based on~$X$ (i.e., as $W(X)$), one obtains \sketch's sublinear
scaling of error and memory.

\section{Memory Analysis}\label{sec:memory_analysis}
So far in Section~\ref{sec:sublime}, we characterized
\sketch's error with respect to the number of counters it uses. We now
incorporate \counterencodingabbrv\ into the analysis to prove an upper bound on
\sketch's memory footprint in bits
(\Cref{thm:all_counters_space}). We also prove a memory footprint
lower bound for any FE sketch that processes a growing stream without knowledge
of its ultimate length (Theorems~\ref{thm:sketch_expandability_special_case}
and \ref{thm:sketch_expandability}). This lower
bound is parameterized by the sketch's accuracy. We observe that \sketch's
upper bound is within a small logarithmic factor of this lower bound,
demonstrating \sketch's ability to give rise to FE sketches that attain a
desired level of accuracy with a close-to-optimal memory~footprint.

\textbf{Upper Bound.}
As described in
\Cref{sec:accommodating_skew}, \sketch\
adaptively tunes \counterencodingabbrv\ to reclaim memory and ensure that a
small fraction of its bits are unused at any time. This implies that the memory
footprint is dominated by the space occupied by the stubs and the extensions.
We quantify this cost for a single counter based on its count as:
\begin{lemma}
    \label{lemma:single_counter_space}
    When working with unsigned counters and stubs of length~$s$ bits,
    \counterencodingabbrv\ encodes a counter value~$v$ in $s$ bits if $v < 2^s$
    and in at most $4 + 1.26 \cdot \log_2 v$ bits otherwise. Signed counters
    use one more bit for encoding the sign in both cases.
\end{lemma}
\begin{proof}
    If~$v < 2^s$, it fits in its stub and occupies~$s$ bits. Otherwise, its~$s$
    lower-order bits are stored in its stub and its higher-order bits
    encoding~$\lfloor v/2^s \rfloor$ are stored in an extension. This extension
    has~$\lfloor \log_3 \lfloor v/2^s \rfloor \rfloor + 2$ fragments to
    represent the digits and the delimiter. Thus, it occupies
    $$ 2 \cdot (\lfloor \log_3 \lfloor v/2^s \rfloor \rfloor + 2) \leq 2 \cdot (\log_3 2 \cdot \log_2 \lfloor v/2^s \rfloor + 2) \leq 4 - s + 1.26 \cdot \log_2 v $$
    bits. This, added to the~$s$ bits of space consumed by the stub, yields the
    result.
\end{proof}
Applying \Cref{lemma:single_counter_space} to a counter array while
accounting for \counterencodingabbrv's tuning yields the following lemma.
Intuitively, it shows that the worst-case scenario is when the counter values
are similar, i.e., their distribution is uniform. This is because skew gives
\counterencodingabbrv\ more opportunities to save space.
\begin{lemma}
    \label{thm:counter_array_per_counter_space}
    Consider a counter array of size~$w$ and let~$k$ be the sum of its
    counters' absolute values. When the counters are unsigned,
    \counterencodingabbrv\ encodes the array in~$w \cdot \left[ 4.41 + 1.26
    \cdot \log_2 (k/w) \right]$ bits. When the counters are signed, it uses one
    extra bit per counter to encode their signs.
\end{lemma}
\begin{proof}
    Applying \Cref{lemma:single_counter_space} to every counter
    within the array and summing results in a total memory footprint of~$m \leq
    \sum_i \max(s, 4+1.26 \cdot \log_2 |v_i|)$ bits. To simplify this
    expression, we bound the stub length~$s$ based on the total absolute
    counter values. Recalling \counterencodingabbrv's tuning mechanics from
    \Cref{sec:accommodating_skew},
    \sketch\ uses stubs that are slightly longer than the length of the average
    counter. In fact, it never sets~$s$ to be more than~2 bits longer than the
    average counter length. Formally, this corresponds to~$s \leq 2 +
    \frac{1}{w} \cdot \sum_i \log_2 |v_i|$, where~$v_i$ denote the signed value
    of the~$i${\nobreakdash-}th counter in the array. Since the logarithm is a
    concave function, we use Jensen's inequality~\cite{Jensen} to bound the sum
    on the right-hand side as~$\log_2(\sum_i |v_i|/w)$. Noting that the
    absolute counter values~$|v_i|$ sum to~$k$ yields $s \leq 2+\log_2(k/w) = 4
    + \log_2(k/(4w))$. Plugging into the memory bound from before yields~$m
    \leq \sum_i 4 + \max(\log_2 (k/(4w)), 1.26 \cdot \log_2 |v_i|)$. We further
    bound each individual max term by 
    \ifappendix
    \begin{align*}
      \max & \left( 1.26 \cdot \log_2 \frac{k}{4w}, 1.26 \cdot \log_2|v_i| \right) \\
      & \leq 1.26 \cdot \max\left( \log_2 \frac{k}{4w}, \log_2|v_i| \right) \\ 
      & \leq 1.26 \cdot \log_2\left( \frac{k}{4w} + |v_i| \right),
    \end{align*}
    \else
    $$ \max\left( 1.26 \cdot \log_2 \frac{k}{4w}, 1.26 \cdot \log_2|v_i| \right) \leq 1.26 \cdot \max\left( \log_2 \frac{k}{4w}, \log_2|v_i| \right) \leq 1.26 \cdot \log_2\left( \frac{k}{4w} + |v_i| \right), $$
    \fi
    thus implying~$m \leq \sum_i 4 + 1.26 \cdot \log(k/(4w) + |v_i|)$.
    Rearranging the sum and once again using the concavity of the logarithm to
    apply Jensen's inequality yields the desired result:
    \ifappendix
    \begin{align*}
      m &\leq w \cdot \left[4 + 1.26 \cdot \frac{1}{w} \cdot \sum_i \log_2(k/(4w)+k/w) \right] \\ 
      & \leq w \cdot \left[ 4.41 + 1.26 \cdot \log_2 (k/w) \right]. 
    \end{align*}
    \else
    $$ m \leq w \cdot \left[4 + 1.26 \cdot \frac{1}{w} \cdot \sum_i \log_2(k/(4w)+k/w) \right] \leq w \cdot \left[ 4.41 + 1.26 \cdot \log_2 (k/w) \right]. $$
    \fi
    \vspace{-8pt}
\end{proof}

We cannot simply set~$k$ to the stream length~$N$ and multiply the space bound
of~\Cref{thm:counter_array_per_counter_space} by the number of
arrays~$d$ to derive \sketch's memory footprint. This is because \sketch\
increases the sum of the absolute values of the counters during expansions by
copying them. By accounting for this, we show the following space bound for
\sketch:
\begin{theorem}
    \label{thm:all_counters_space}
    When using a sublinear power size as the function~$W(N)$, \sketchcms, has a
    memory footprint of~$\approx d \cdot W(N) \cdot \left[ 4.41 + 1.26 \cdot
    \log_2 N/W(N) \right]$. With a linear size function, \sketchcms
    uses~$\approx d \cdot W(N) \cdot \left[4.41 + 1.26 \cdot \log_2(N \log_2
    N/W(N)) \right]$ bits of memory. The same bounds hold for \sketchcs
    with~$W(F)$ instead of $W(N)$ and with an additive overhead of one bit per
    counter.
\end{theorem}
\begin{proof}
    By a similar analysis to that of
    \Cref{thm:cms-w}, the sum
    of the absolute counter values within each array in \sketch\ when using a
    sublinear power as the size function is~$\approx N$, modulo a small
    constant factor. This constant factor become especially negligible in the
    space bound, since it will appear within a logarithm. Similarly, when the
    size function is linear, we have~$k \approx N \log N$. Plugging these
    values into the bound of~\Cref{thm:counter_array_per_counter_space}
    and multiplying by~$d$ proves the result.
\end{proof}

\ifappendix
\begin{figure*}
    \centering
    \includegraphics[width=0.65\textwidth]{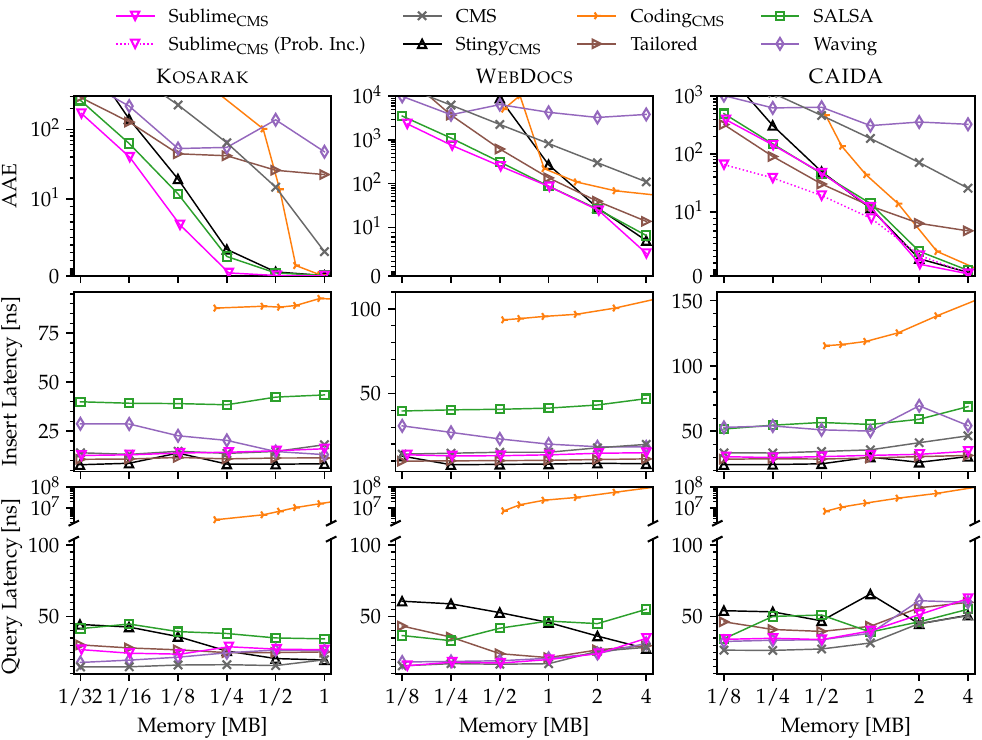} 
    \hspace*{4pt}
    \small
    \begin{minipage}{0.31\textwidth}
    \setlength{\tabcolsep}{2.5pt}
    \begin{tabular}[b]{ccccccc} 
        \toprule 
        \multirow{2}{*}[0.0em]{Memory [MB]} & \multicolumn{2}{c}{$\textsc{Kosarak}$} & \multicolumn{2}{c}{$\textsc{WebDocs}$} & \multicolumn{2}{c}{$\textsc{CAIDA}$} \\ 
        & \hspace*{1.5mm}$c$ & $s$ & \hspace*{1.5mm}$c$ & $s$\hspace*{0.75mm} & \hspace*{0.75mm}$c$ & $s$ \\ 
        \midrule 
        1/32 & \hspace*{1.5mm}39 & 10 & \hspace*{1.5mm}- & -\hspace*{0.75mm} & \hspace*{0.75mm}- & - \\ 
        1/16 & \hspace*{1.5mm}42 & 9 & \hspace*{1.5mm}- & -\hspace*{0.75mm} & \hspace*{0.75mm}- & - \\ 
        1/8 & \hspace*{1.5mm}43 & 7 & \hspace*{1.5mm}33 & 12\hspace*{0.75mm} & \hspace*{0.75mm}39 & 10 \\ 
        1/4 & \hspace*{1.5mm}68 & 5 & \hspace*{1.5mm}34 & 12\hspace*{0.75mm} & \hspace*{0.75mm}42 & 9 \\ 
        1/2 & \hspace*{1.5mm}75 & 5 & \hspace*{1.5mm}39 & 10\hspace*{0.75mm} & \hspace*{0.75mm}45 & 8 \\ 
        1 & \hspace*{1.5mm}81 & 4 & \hspace*{1.5mm}38 & 10\hspace*{0.75mm} & \hspace*{0.75mm}51 & 7 \\ 
        2 & \hspace*{1.5mm}- & - & \hspace*{1.5mm}50 & 7\hspace*{0.75mm} & \hspace*{0.75mm}64 & 5 \\ 
        4 & \hspace*{1.5mm}- & - & \hspace*{1.5mm}70 & 5\hspace*{0.75mm} & \hspace*{0.75mm}75 & 4 \\ 
        \bottomrule 
    \end{tabular} 
        \captionof{table}{Tuning \counterencodingabbrv's parameters for
        each plot point to the left according to the table above enables high
        performance and accuracy for \sketchcms.}
    \vspace{7.6cm}
    \end{minipage} \\
    \vspace*{-7.0cm}
    \caption{\sketchcms achieves the highest accuracy under real-world
    workloads while maintaining equal or superior insertion and query
    performance. We use the optimal tuning of the baselines at each point of
    the curves. The table on the right details the number of counters per
    chunk~$c$ and the stub length~$s$ used by \sketchcms in each point.}
    \vspace*{-2.0mm}
    \label{fig:accuracy_bench}
\end{figure*}
\fi

\textbf{Lower Bound.}
We prove a memory footprint lower bound for FE sketches by reducing them to
filters. That is, we use an FE sketch to implement filter functionality, which
allows us to inherit known memory lower bounds for filtering. A filter is a
compact data structure that answers whether a query key exists in a given set
of keys. A filter never returns a false negative, but it may return a false
positive with a probability called the \emph{False Positive Rate (FPR)}. The
lower the desired FPR is, i.e., the higher the target accuracy level, the
larger the filter's memory footprint must be. It is known that a filter that
supports a set of a growing size of~$n$ with an FPR of~$\delta$ must use at
least~$n \cdot [ \Omega( \log 1/\delta + \log\log n) - \Theta(1)]$ bits of
memory at some point in time~\cite{PaghExpandability}. Thus, an FE sketch that
can implement such a filter must also use at least this much memory.

Our reduction works as follows. Consider a FE sketch (e.g., CMS). Suppose this
sketch never underestimates a key's count and guarantees an error of at
most~$E(N)$ with a confidence of~$1-\delta$ for streams of length~$N$. We use
it to implement a filter with an FPR of~$\delta$ over a set of
size~$n=N/(E(N)+1)$. Specifically, we insert each key in the set~$E(N)+1$ times
into the FE sketch. This causes the FE sketch to return an estimate of at
least~$E(N)+1$ during a query to an existing key within the set. In contrast,
the FE sketch's accuracy guarantees imply that a query to a non-existent key
returns an estimate of at most~$E(N)$ with a probability of~$1-\delta$. Thus,
by thresholding the estimate by~$E(N)+1$, we return a true positive for
existing keys and a false positive for non-existent keys with an FPR of at
most~$\delta$. The following theorem formalizes these intuitions:

\begin{theorem}
    \label{thm:sketch_expandability_special_case}
    Consider an FE sketch with only overestimation errors that processes
    growing streams with an error bound of~$E(N)$ and a confidence
    of~$1-\delta$. Such a sketch must use at least~$N/E(N) \cdot [\Omega\left(
    \log 1/\delta + \log\log (N/E(N)) \right) - \Theta(1)]$ bits of memory at
    some point in time while processing a stream.
\end{theorem}

\Cref{thm:sketch_expandability_special_case} only holds for FE
sketches that do not underestimate (e.g., CMS). It is non-trivial to adapt the
same argument to FE sketches that can underestimate (e.g., CS), since
underestimations can introduce false negatives and violate filtering semantics.
Nevertheless, by accounting for the estimation variance, we reason about the
false negatives and
\ifappendix
prove in \Cref{sec:ommitted_proofs_lower}
\else 
prove in our Appendix~\cite{SublimeArxiv}
\fi
that the same lower bound holds for these sketches as well:
\begin{restatable}{theorem}{thmsketchexpandabilitylabel}
    \label{thm:sketch_expandability}
    Consider an FE sketch providing estimates with an expected absolute error
    of~$E(N)$ and a variance of~$\delta \cdot (E(N))^2$ for any key when
    processing a stream of size~$N$. Such an FE sketch must use at
    least~$N/E(N) \cdot [\Omega(\log 1/\delta + \log\log (N/E(N))) -
    \Theta(1)]$ bits of memory at some point in time while processing a stream.
\end{restatable}

One can verify from \Cref{thm:all_counters_space} that, when
fixing the confidence using a constant number of arrays~$d$ while using a
sublinear power (resp. linear) size function~$W(N)$ to achieve an error bound
of~$E(N)$, \sketch\ uses a memory of~$O(N/E(N) \cdot \log (N/E(N)))$ (resp.
$O(N/E(N) \cdot \log N \cdot \log\log N)$) bits. This almost matches the memory
lower bound of \Cref{thm:sketch_expandability}
with a slight difference in the logarithmic
terms. One can extend \sketch\ to match these lower bounds by compressing the
counter arrays via coarsely quantizing their counters and applying run-length
encoding. Doing so can potentially reduce space consumption by as much
as~$5\times$ for long streams in practice. These compression techniques may
compromise update and query performance, though we expect to be able to resolve
this degradation using techniques based on rank and
select~\cite{Poppy,FlorianRankandSelect,CompactPATTrees,RankandSelectDict,Succincter}.

\section{Evaluation}\label{sec:evaluation}
We empirically evaluate \sketch.
Experiments~\hyperlink{experiment:accuracy}{1}, \hyperlink{experiment:skew}{2},
\hyperlink{experiment:vale_tuning}{3}, and
\hyperlink{experiment:vale_tuning}{6} focus on accommodating skew, while
Experiments~\hyperlink{experiment:expansion}{4},
\hyperlink{experiment:contraction}{5}, and
\hyperlink{experiment:l2_size_function}{7} deal with unbounded data growth.
Experiments~\hyperlink{experiment:accuracy}{1}-\hyperlink{experiment:contraction}{5}
all employ \sketchcms. Experiments~\hyperlink{experiment:arruacy_unbiased}{6}
and \hyperlink{experiment:l2_size_function}{7} showcase \sketchcs.
Experiment~\hyperlink{experiment:join_size}{8} applies both variants to join
size estimation.

\textbf{Platform.}
We run the experiments on a Fedora 41 machine with an Intel Xeon w7-2495X
processor (4.8 GHz) with 24 cores. Our machine features an 80 KB L1 cache and a
2 MB L2 cache for each core, a 45 MB shared L3 cache, and 64 GB of RAM.

\ifnoappendix
\begin{figure*}
    \centering
    \includegraphics[width=0.8\textwidth]{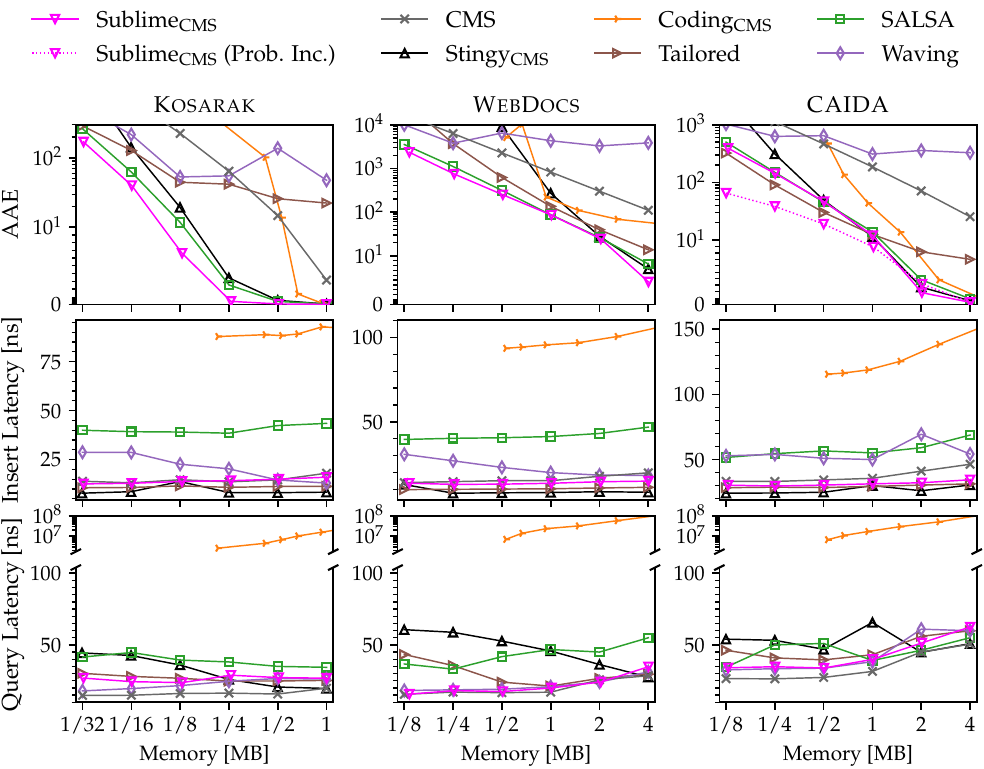} 
    \hspace*{8.5mm}
    \small
    \begin{minipage}{0.7\textwidth}
    \vspace*{2mm}
    \setlength{\tabcolsep}{4.03mm}
    \begin{tabular}[b]{ccccccc} 
        \toprule 
        \multirow{2}{*}[0.0em]{Memory [MB]} & \multicolumn{2}{c}{\textsc{\hspace*{1.2mm}Kosarak}} & \multicolumn{2}{c}{\hspace*{0.3mm}\textsc{WebDocs}} & \multicolumn{2}{c}{\hspace*{0.5mm}\textsc{CAIDA}} \\ 
        & \hspace*{1.5mm}$c$ & $s$ & \hspace*{1.5mm}$c$ & $s$\hspace*{0.75mm} & \hspace*{0.75mm}$c$ & $s$ \\ 
        \midrule 
        1/32 & \hspace*{1.5mm}39 & 10 & \hspace*{1.5mm}- & -\hspace*{0.75mm} & \hspace*{0.75mm}- & - \\ 
        1/16 & \hspace*{1.5mm}42 & 9 & \hspace*{1.5mm}- & -\hspace*{0.75mm} & \hspace*{0.75mm}- & - \\ 
        1/8 & \hspace*{1.5mm}43 & 7 & \hspace*{1.5mm}33 & 12\hspace*{0.75mm} & \hspace*{0.75mm}39 & 10 \\ 
        1/4 & \hspace*{1.5mm}68 & 5 & \hspace*{1.5mm}34 & 12\hspace*{0.75mm} & \hspace*{0.75mm}42 & 9 \\ 
        1/2 & \hspace*{1.5mm}75 & 5 & \hspace*{1.5mm}39 & 10\hspace*{0.75mm} & \hspace*{0.75mm}45 & 8 \\ 
        1 & \hspace*{1.5mm}81 & 4 & \hspace*{1.5mm}38 & 10\hspace*{0.75mm} & \hspace*{0.75mm}51 & 7 \\ 
        2 & \hspace*{1.5mm}- & - & \hspace*{1.5mm}50 & 7\hspace*{0.75mm} & \hspace*{0.75mm}64 & 5 \\ 
        4 & \hspace*{1.5mm}- & - & \hspace*{1.5mm}70 & 5\hspace*{0.75mm} & \hspace*{0.75mm}75 & 4 \\ 
        \bottomrule 
    \end{tabular} 
        \captionof{table}{Tuning \counterencodingabbrv's parameters for
        each plot point to the left according to the table above enables high
        performance and accuracy for \sketchcms.}
    \end{minipage} \\
    \vspace{-7mm}
    \caption{\sketchcms achieves the highest accuracy under real-world
    workloads while maintaining equal or superior insertion and query
    performance. We use the optimal tuning of the baselines at each point of
    the curves. The table above details the number of counters per chunk~$c$
    and the stub length~$s$ used by \sketchcms.}
    \vspace*{-8.0mm}
    \label{fig:accuracy_bench}
\end{figure*}
\fi

\textbf{Baselines.}
We compare \sketchcms to a traditional CMS, as well as the variants described
in \Cref{sec:problem_analysis} that cater to
skew. We use SALSA~\cite{SALSA} to represent counter merging methods and Waving
sketch~\cite{WavingSketch} to represent hybrid methods, as they are the most
accurate and the fastest in their category. For counter sharing methods, we
employ the following three baselines: Stingy sketch~\cite{StingySketch} is the
state of the art in terms of throughput since it implements the classic counter
sharing design with a hierarchy of shared byte-aligned counters that prefetches
the counters, as described in
\Cref{sec:problem_analysis}. Coding
sketch~\cite{CodingSketch} is the most accurate counter sharing method that
compresses its hierarchy of counters, trading performance for memory. Tailored
sketch~\cite{TailoredSketch} is the most space-efficient (and thus accurate)
counter sharing method, as it probabilistically increments each counter in the
classic hierarchy, thereby shortening the counters' length in exchange for
introducing two-sided errors. Among these six baselines, only Coding sketch is
tunable, though its tuning is manual and complex. This is because its
compression algorithm obscures the level of skew in the workload to which the
tuning should be adapted. We implement \sketch\ in \texttt{C++} and use the
open-source \texttt{C++} implementations of the baselines. The aforementioned
implementations are all single-threaded. We construct all FE sketches
with~$d=3$ arrays. This parameter setting is common in practice since it yields
a confidence bound of~$\approx 95\%$, which is significantly higher than the
confidence bounds resulting from~$d=1$ ($\approx 63\%$) and $d=2$ ($\approx
86\%$) while being comparable~to~$d=4$~($\approx 98\%$).

\ifappendix
\begin{figure}
    \centering
    \includegraphics[width=\columnwidth]{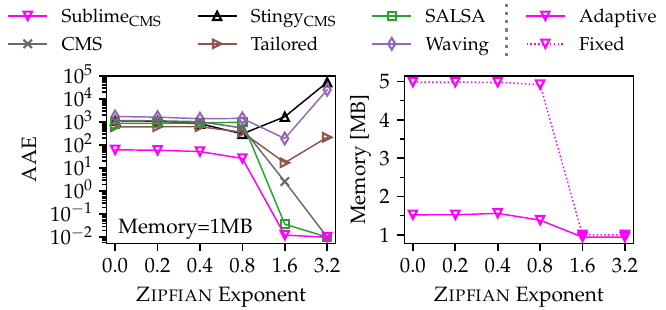}
    \begin{tikzpicture}
        \node[inner sep=1pt,anchor=north,align=left] (A) at (0,0) {\textbf{A)}};
        \node[inner sep=1pt,right=3.5 of A,align=center] (B) {\textbf{B)}};
        \node[inner sep=1pt,left=1.0 of A,align=center] (dummy) {};
    \end{tikzpicture}
    \vspace{-3mm}
    \caption{By adapting \counterencodingabbrv\ to skew, \sketchcms
    achieves the highest accuracy among all baselines (Part A) and a smaller
    memory footprint compared to using a fixed tuning (Part B).}
    \label{fig:skew_vale_tuning_bench}
\end{figure}
\fi

\textbf{Datasets and Workloads.}
Following prior
work~\cite{CountMeanMin,SALSA,BitSense,PyramidSketch,StingySketch,AdaptiveCounterSplicing,CounterTree,TreeSensing,TailoredSketch,AugmentedSketch,HeavyGuardian,HeavyKeeper,WavingSketch,ElasticSketch,MVSketch,SpaceSaving,BatchUpdateMisraGries,BitMatcher,JigsawSketch,MicroscopeSketch,PSketch,TightSketch},
we compare the baselines in a single-threaded setting by employing three
real-world~streams:
\begin{itemize}
    \item \textsc{Kosarak}: 8M clicks from a news portal~\cite{Kosarak},
      anonymized as 8-byte integers. This dataset is highly skewed, resembling
      a \textsc{Zipfian} distribution with an exponent parameter of~$\approx
      3.5$.
    \item \textsc{WebDocs}: 300M words extracted from 1.7M online English
        documents~\cite{WebDocs}. Each word in this dataset is encoded as
        an 8-byte integer. This dataset follows a skewed distribution with a
        \textsc{Zipfian} exponent of~$\approx 1.0$.
    \item \textsc{CAIDA}: 27M anonymized 5-tuples describing the network
      connection of packets from 10 traces taken from backbone Internet links
      in 2018~\cite{CAIDA2018}. Each 5-tuple is represented as a 13-byte
      string. This dataset has a \textsc{Zipfian} exponent of~$\approx 2.0$.
\end{itemize}
We evaluate the accuracy of each baseline by querying it with each distinct key
in the stream and compute the \emph{Average Absolute Error (AAE)} of the
estimates from the ground-truth. We focus on absolute errors rather than
relative errors to keep our presentation consistent with the theoretical error
bounds studied in the literature. We have found the relative errors to behave
similarly to the absolute errors.

\ifnoappendix
\begin{figure}
    \centering
    \includegraphics[width=0.99\columnwidth]{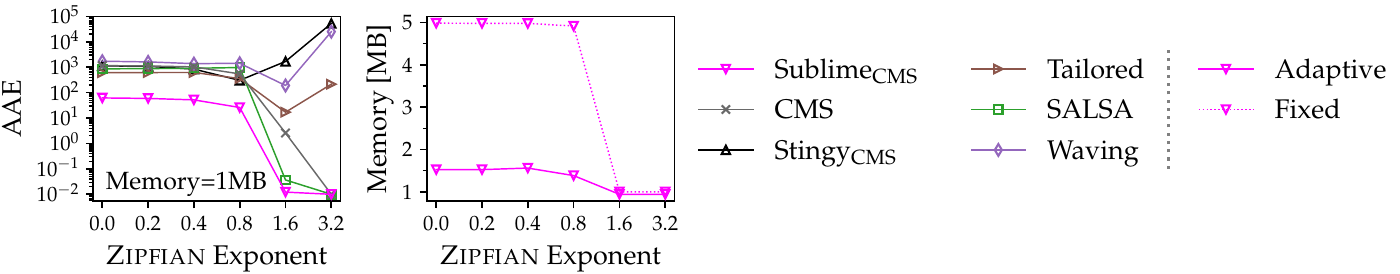}
    \\
    \begin{tikzpicture}
        \node[inner sep=1pt,anchor=north,align=left] (A) at (0,0) {\textbf{A)}};
        \node[inner sep=1pt,right=2.87 of A,align=center] (B) {\textbf{B)}};
        \node[inner sep=1pt,left=1.05 of A,align=center] (dummy) {};
    \end{tikzpicture}
    \hspace*{7.15cm}
    \vspace{-3mm}
    \caption{By adapting \counterencodingabbrv\ to skew, \sketchcms
    achieves the highest accuracy among all baselines (Part A) and a smaller
    memory footprint compared to using a fixed tuning (Part B).}
    \vspace{-5mm}
    \label{fig:skew_vale_tuning_bench}
\end{figure}
\fi

\hypertarget{experiment:accuracy}{\textbf{Experiment 1: Accuracy and Speed vs. Memory.}}
In \Cref{fig:accuracy_bench}, we feed each workload to each baseline while
varying the allotted memory budget. This experiment features bounded data
growth, allowing us to focus on how well each baseline accommodates skew. Each
point on the curves represents a complete run of the experiment starting from
scratch. We tune Coding sketch and \sketchcms at each point to minimize their
error and showcase their highest accuracy. The table to the right of
\Cref{fig:accuracy_bench} details the tuned parameters of
\sketchcms{\textemdash}the number of counters per chunk~$c$ and the stub
length~$s$. We increase the memory budget by factors of 2 on the $x$-axis up to
4 MBs, though for the \textsc{Kosarak} dataset, we use smaller budgets due to
the workload trace~being~shorter.

\ifappendix
\begin{figure}
    \centering
    \hspace*{3.25pt}
    \includegraphics[width=0.95\columnwidth]{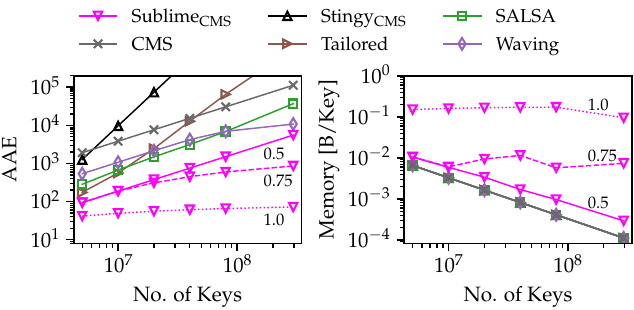}
    \vspace{1.1mm}
    \caption{\sketchcms achieves both a low error rate and a small memory
    footprint when processing growing streams.}
    \label{fig:expansion_bench}
\end{figure}

\begin{figure}
    \centering
    \includegraphics[width=0.98\columnwidth]{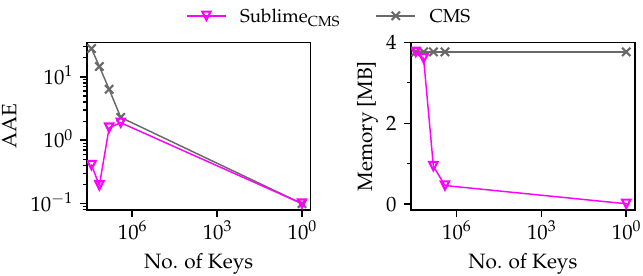}
    \caption{\sketchcms saves a significant amount of memory by contracting when
    many keys are deleted.}
    \label{fig:contraction_bench}
\end{figure}
\fi

\ifappendix
\begin{figure*}
    \centering
    \includegraphics[width=0.7\textwidth]{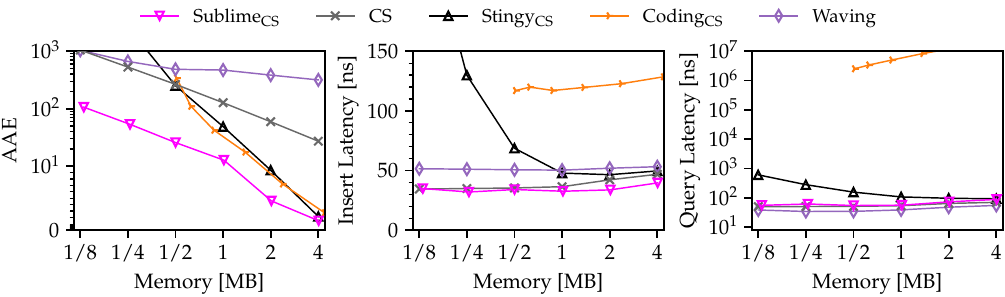}
    \hspace{4pt}
    \def\arraystretch{0.77}%  1 is the default, change whatever you need
    \small
    \begin{minipage}{0.28\textwidth}
    \setlength{\tabcolsep}{13.1pt}
    \begin{tabular}[b]{ccc} 
        \toprule 
        Memory [MB] & $c$ & $s$ \\ 
        \midrule 
        1/8 & 39 & 10 \\ 
        1/4 & 42 & 9 \\ 
        1/2 & 45 & 8 \\ 
        1 & 51 & 7 \\ 
        2 & 64 & 5 \\ 
        4 & 75 & 4 \\ 
        \bottomrule 
    \end{tabular} 
    \captionof{table}{Tuning \counterencodingabbrv\ as above leads to high
      performance and accuracy.}
    \vspace{2.45cm}
    \end{minipage} \\
    \vspace*{-3.05cm}
    \caption{\sketchcs has the lowest errors and fastest insertions compared to
    all other unbiased baselines. The table to the right presents the tuning of
    \counterencodingabbrv's parameters, i.e., the number of counters per
    chunk~$c$ and the stub length~$s$.}
    \label{fig:accuracy_unbiased_bench}
\end{figure*}
\fi

\Cref{fig:accuracy_bench}
shows that \sketchcms has a lower estimation error than most other baselines by
as much as an order of magnitude while providing similar or superior insertion
and query performance. This is because it extends overflowing counters without
affecting other counter values. Compared to the baseline with the closest
errors, i.e., SALSA, \sketchcms has lower errors by 40\% across the board while
providing faster insertions and queries by as much as~$3\times$ and
$2.4\times$. This performance improvement is due to \sketchcms's use of rank
and select primitives, as well as
Optimizations~\hyperlink{optimization:rank_and_select}{1} to
\hyperlink{optimization:horner_hack}{6} described in
\Cref{sec:accommodating_skew}. Coding
sketch exhibits slower insertion and query performance by~2 and 7 orders of
magnitude compared to the other baselines due to its decompression overheads.
As such, we henceforth exclude it from the experiments that evaluate
\sketchcms. When the memory budget is low under the \textsc{CAIDA} dataset,
Tailored sketch is more accurate than \sketchcms due to its probabilistic
incrementation strategy enabling it to store shorter counters in exchange for
introducing underestimations. Applying the same technique to \sketchcms yields
lower errors than Tailored sketch by~$2.38\times$, as shown by the dotted
curve.

\hypertarget{experiment:skew}{\textbf{Experiment 2: Accuracy under Varying Levels of Skew.}}
\Cref{fig:skew_vale_tuning_bench}-A) compares the estimation errors of the
baselines under the same memory budget as the amount of skew in the data
varies. Here, we generate streams of 100M keys from a \textsc{Zipfian}
distribution over 100K distinct keys. We vary the exponent parameter, i.e., the
skewness, on the~$x${\nobreakdash-}axis from~0.0 to 3.2, which correspond to a
uniform and a highly skewed dataset, respectively. Recall that, aside from
Coding sketch (which we exclude from the experiment due to its poor
performance), \sketchcms is the only tunable FE sketch. Thus, we allow it to
adapt \counterencodingabbrv's tuning to the workload's skew.

\Cref{fig:skew_vale_tuning_bench}-A) shows that \sketchcms provides the lowest
error under all degrees of skew. It outperforms all other baselines by an order
of magnitude for more uniform workloads and by~$3\times$ for more skewed
workloads. This is because it exposes a tuning space that allows
\counterencodingabbrv\ to save memory by employing longer and shorter stubs for
more uniform and skewed workloads, respectively. 

\ifnoappendix
\setlength{\tabcolsep}{4pt}
\begin{tabular}{cc}
    \hspace*{-6mm}
    \begin{minipage}{0.51\textwidth}
        \centering
        \includegraphics[width=\columnwidth]{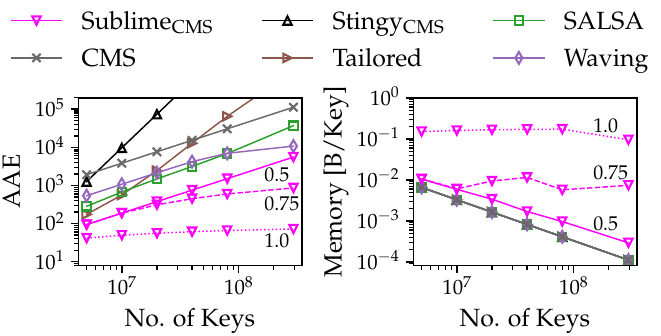}
        \vspace{-7mm}
        \captionof{figure}{\sketchcms achieves both a low error rate and a small memory
        footprint when processing~growing~streams.}
        \label{fig:expansion_bench}
    \end{minipage}
    &
    \begin{minipage}{0.47\textwidth}
        \centering
        \includegraphics[width=\columnwidth]{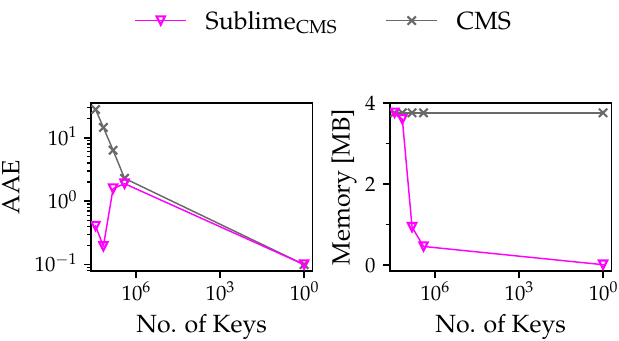}
        \vspace{-7mm}
        \captionof{figure}{\sketchcms saves a significant amount of memory by contracting when
        many keys~are~deleted.}
        \label{fig:contraction_bench}
    \end{minipage}
\end{tabular}
\vspace{2mm}
\fi

\hypertarget{experiment:vale_tuning}{\textbf{Experiment 3: Memory Savings of Adaptively Tuning \counterencodingabbrv.}}
We demonstrate the need for adapting \counterencodingabbrv's tuning as the
stream grows by comparing \sketchcms's memory footprint when using a fixed
tuning to when it adapts the tuning. We employ a total of 1M counters for both
versions and use~$c=64$ counters per chunk and stubs of length~$s=6$ bits for
the fixed tuning version.
\Cref{fig:skew_vale_tuning_bench}-B)
considers the \textsc{Zipfian} workloads from
Experiment~\hyperlink{experiment:skew}{2} and varies the skew on the $x$-axis.
As shown, adaptively tuning \counterencodingabbrv\ bounds the number of
allocated tails arrays and improves the memory footprint over using a fixed
tuning by as much as~$5\times$. With enough skew, the average extension length
decreases, and the fixed tuning version of \sketchcms avoids allocating many
tails arrays. Nevertheless, it does not optimally leverage the shorter
extension lengths and consumes~$\approx 10\%$ more memory than the adaptive
version. 

In this experiment, adaptations comprise less than~3\% of the processing time,
as their overhead is amortized. One can expect this overhead in mixed workloads
with both insertions and deletions to be similar to our insertion-only
workload. This is because our workload triggers just as many adaptations as
mixed workloads by growing the counters and causing many chunks to overflow.

\hypertarget{experiment:expansion}{\textbf{Experiment 4: Accuracy and Memory under Data Growth.}}
We now compare the accuracy and memory footprint of the baselines when the
stream grows indefinitely. Here, we allocate each baseline an initial memory
footprint of 32KB and insert the \textsc{WebDocs} dataset.
\Cref{fig:expansion_bench}
measures, for each baseline, its error and the ratio of its memory footprint to
the stream's length. We use the default configuration for \sketchcms, which has
a sublinear power size function of the form~$W(N)=N^\alpha/\epsilon$ with
$\alpha=0.5$ and $\epsilon=1$, as introduced in
\Cref{sec:accommodating_unknown_stream_lengths}.
We also compare versions of \sketchcms with the exponent~$\alpha$ of the size
function set to~$0.75$ and $1.0$. \Cref{fig:expansion_bench}
differentiates these versions by annotating each curve with the value
of~$\alpha$ used. 

\Cref{fig:expansion_bench}
shows that \sketchcms is the only FE sketch to exhibit sublinearly scaling
estimation errors. It achieves this by expanding its number of counters along
with the stream. This culminates in the three versions of \sketchcms (with
exponents~$\alpha$ of $0.5$, $0.75$, and $1.0$) improving the error over the
best baseline by one, two, and three orders of magnitude. The right subfigure
in \Cref{fig:expansion_bench} shows that, despite
expanding the number of counters, \sketchcms maintains a modest memory
footprint relative to the stream's length. Furthermore, we have found that
expansions take up only~0.5\% of the processing time, as they simply copy the
counters sequentially.

\hypertarget{experiment:contraction}{\textbf{Experiment 5: Contractions.}}
\Cref{fig:contraction_bench}
showcases
deletions and contractions in \sketchcms. Here, we insert the \textsc{CAIDA}
dataset and subsequently delete keys in a random order while measuring the
error and the memory footprint. We only compare \sketchcms to CMS since all
other baselines do not provide a deletion API.

\ifappendix
\begin{figure}
    \centering
    \includegraphics[width=0.91\columnwidth]{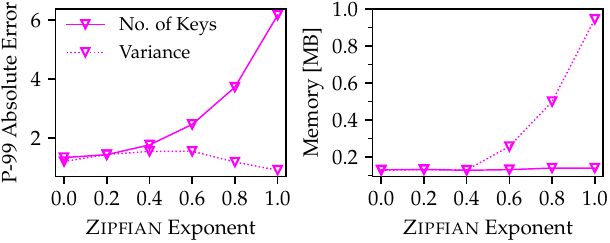}
    \caption{By expanding based on the estimation variance, \sketchcs maintains
    a stable confidence bound.}
    \label{fig:l2_size_function_bench}
\end{figure}
\fi

\ifappendix
\begin{figure*}
    \centering
    \includegraphics[width=0.86\textwidth]{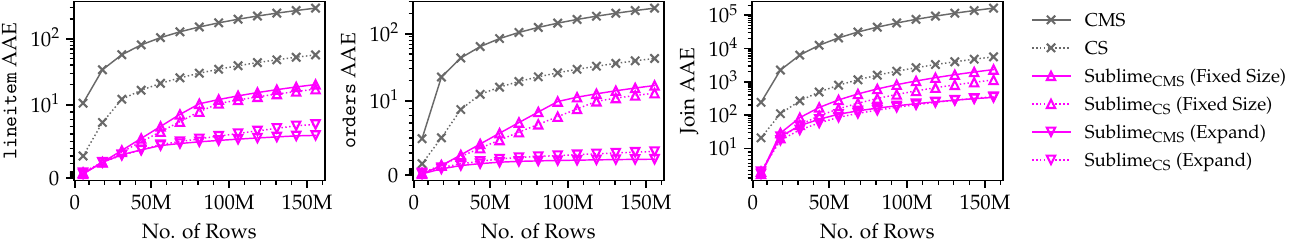}
    \vspace{-2mm}
    \caption{\sketch\ provides the most accurate frequency estimates of table
    entries and the size of their join.}
    \vspace{-2mm}
    \label{fig:join_size_bench}
\end{figure*}
\fi

As shown, \sketchcms maintains lower errors than a CMS of the same size by as
much as two orders of magnitude. This is because \counterencodingabbrv\ encodes
the counter values in less space, allowing the storage of more counters. As
keys are deleted, \sketchcms contracts to save memory, while CMS wastes memory
by remaining at the same size. This enables achieving a lower memory footprint
than CMS by 2-3 orders of magnitude while still being more accurate, despite
the reduced number of counters slightly increasing the error. Moreover, since
contractions sequentially scan the FE sketch, similarly to expansions, they
take up at most~0.5\% of the total execution time.

\ifnoappendix
\begin{figure*}
    \centering
    \includegraphics[width=0.79\textwidth]{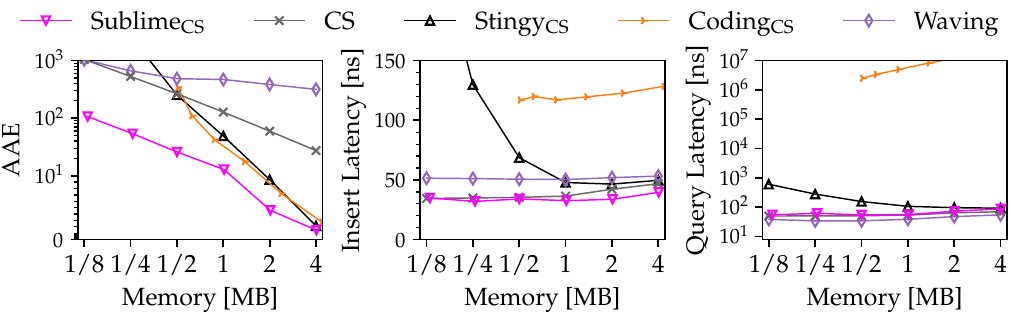}
    \hspace{1pt}
    \def\arraystretch{0.77}%  1 is the default, change whatever you need
    \small
    \begin{minipage}{0.19\textwidth}
    \setlength{\tabcolsep}{1.25pt}
    \begin{tabular}[b]{ccc} 
        \toprule 
        Memory [MB] & $c$ & $s$ \\ 
        \midrule 
        1/8 & 39 & 10 \\ 
        1/4 & 42 & 9 \\ 
        1/2 & 45 & 8 \\ 
        1 & 51 & 7 \\ 
        2 & 64 & 5 \\ 
        4 & 75 & 4 \\ 
        \bottomrule 
    \end{tabular} 
    \captionof{table}{Tuning \counterencodingabbrv\ as above leads to high
      performance and accuracy.}
    \vspace{1.89cm}
    \end{minipage} \\
    \vspace*{-2.69cm}
    \caption{\sketchcs has the lowest errors and fastest insertions compared to
    all other unbiased baselines. The table to the right presents the tuning of
    \counterencodingabbrv's parameters, i.e., the number of counters per
    chunk~$c$ and the stub length~$s$.}
    \vspace*{-5mm}
    \label{fig:accuracy_unbiased_bench}
\end{figure*}
\fi

\hypertarget{experiment:accuracy_unbiased}{\textbf{Experiment 6: Unbiased Accuracy and Speed vs. Memory.}}
We now turn to \sketchcs to demonstrate the applicability of \sketch\ beyond
CMS. We compare it to a traditional CS and Waving sketch~\cite{WavingSketch},
which is representative of unbiased hybrid methods. We also compare to Stingy
sketch~\cite{StingySketch} and Coding sketch~\cite{CodingSketch} applied to CS,
representing counter sharing methods with unbiased estimates. We exclude
SALSA~\cite{SALSA} as its library does not provide a variant with unbiased
estimates. We employ the \textsc{CAIDA} dataset and optimally tune Coding
sketch and \sketchcs.

\ifnoappendix
\begin{figure}
    \centering
    \includegraphics[width=0.50\columnwidth]{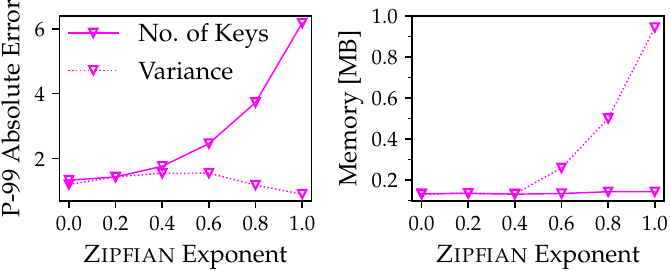}
    \vspace*{-2mm}
    \caption{By expanding based on the estimation variance, \sketchcs maintains
    a stable confidence bound.}
    \vspace*{-5mm}
    \label{fig:l2_size_function_bench}
\end{figure}
\fi

\Cref{fig:accuracy_unbiased_bench}
shows that, by leveraging the memory savings of \counterencodingabbrv\ to store
more counters, \sketchcs achieves lower errors than the next best baseline by
as much as an order of magnitude. It also has the fastest insertions while
closely matching the query speed of the most performant sketch due to its use
of Optimizations~\hyperlink{optimization:rank_and_select}{1} to
\hyperlink{optimization:horner_hack}{6} from \Cref{sec:accommodating_skew}.

\hypertarget{experiment:l2_size_function}{\textbf{Experiment 7: Expansion Based on Variance.}}
\Cref{fig:l2_size_function_bench}
illustrates how expanding \sketchcs based on the growth of the estimation
variance (as described
in~\Cref{sec:cs}) instead
of the number of keys enables a stable confidence bound. We demonstrate this
improvement by inserting~100M keys with increasing levels of skew. We measure
the P-99 absolute error of the estimates instead of the average, since it
showcases the errors that exceed their confidence bounds, whereas averages
obscure this information. We also report the final memory footprint after
ingesting~the~stream. 

As shown in
\Cref{fig:l2_size_function_bench},
expanding based on the estimation variance ensures stable tail errors in
exchange for a higher memory footprint. This is in contrast to expanding
\sketchcs based on the number of keys, which yields tail errors that increase
with skew. The reason is that a more skewed dataset grows the variance more
quickly, necessitating more frequent expansions to maintain a stable confidence
bound. As skew increases further, expanding based on the variance yields
diminishing returns. This is due to the stream's distinct keys dropping in
number with skew, making keys less likely to map to the same counter and create
severe errors.

\hypertarget{experiment:join_size}{\textbf{Experiment 8: Join Size Estimation.}}
\Cref{fig:join_size_bench}
applies both \sketchcms and \sketchcs to join size estimation. Here, we employ
a 100 GB TPC-H workload~\cite{TPCH}. We consider the \texttt{lineitem} and
\texttt{orders} tables, incrementally insert rows into each, and periodically
join them on their shared \texttt{ORDERKEY} column.
Following~\cite{JoinSizeEstimationFilterConditions}, each row in the orders
table is duplicated 1-3 times to simulate a many-to-many join. We construct FE
sketches tracking the count of each \texttt{ORDERKEY} value in these tables and
predict their counts in the join by multiplying the per-table estimates. We
compare CMS and CS, as well as fixed-size and expandable versions of \sketchcms
and \sketchcs. We set the initial memory budget to 8MB ($\approx0.1\%$ of the
data size) and distribute it to each sketch in proportion to the initial size
of its table. The expandable \sketchcms and \sketchcs grow as a sublinear power
(with a power of 0.75) of the tables. \Cref{fig:join_size_bench}
reports accuracy measurements for each table and the join result.

As shown in 
\Cref{fig:join_size_bench},
\sketch\ attains higher accuracy than CMS and CS by multiple orders of
magnitude due to their use of variable-length counters. Moreover, the
expandable \sketchcms and \sketchcs outperform their fixed-size counterparts by
as much as an order of magnitude. The fixed-size version of \sketchcs is more
accurate than the fixed-size version of \sketchcms on each separate table. This
is because the workload in this experiment is almost uniform, leading to the
positive and negative counts of colliding keys in \sketchcs to cancel out. When
allowed to expand, \sketchcms is more accurate than \sketchcs on individual
tables, as collisions reduce with expansions, and \sketchcms does not store a
sign bit for each counter. Nevertheless, both versions of \sketchcs have more
accurate join estimates than the corresponding version of \sketchcms. This is
due to \sketchcs returning unbiased estimates, which allows for multiplying
them without compounding errors.

\ifnoappendix
\begin{figure*}
    \centering
    \includegraphics[width=\textwidth]{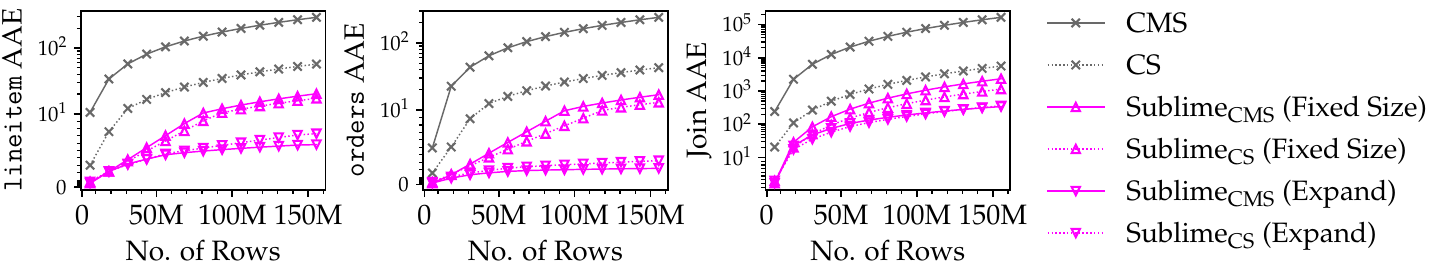}
    \vspace{-7mm}
    \caption{\sketch\ provides the most accurate frequency estimates of table
    entries and the size of their join.}
    \label{fig:join_size_bench}
    \vspace*{-5mm}
\end{figure*}
\fi

\section{Conclusion}
We introduced \sketch, the first framework for generalizing a frequency
estimation sketch to adapt to the skew and length of the stream. We showed that
\sketch\ improves the accuracy and memory of FE sketches beyond the state of
the art while maintaining high performance. Although \sketch\ is currently
single-threaded, we expect significant performance gains from parallelization
(e.g., through employing thread-local FE sketches to ingest subsets of the
stream and periodically merging them into a unified FE sketch for answering
queries~\cite{FastConcurrentDataSketches,FPGASketching}). \ifappendix Studying
the interplay of concurrent updates with expansions, contractions, and adaptive
tuning would be an interesting direction of future work. \fi More broadly,
\sketch\ opens a new research avenue on expandable data sketches applicable to
other core statistics, such as quantiles and cardinalities, which exhibit
similar error growth under workloads with growing datasets. This marks a
fundamental shift from fixed-size sketches with degrading accuracy to adaptive
structures that maintain stable accuracy for always-on analytics.

\begin{acks}
    We thank Yakub Tetek for early and insightful discussions and the reviewers
    for their constructive comments. This research was supported by the NSERC
    grant \#RGPIN-2023-03580.
\end{acks}

%\clearpage

\bibliographystyle{ACM-Reference-Format}
\bibliography{sublime}

@string{BIT = "{BIT}"}

@string{Computing = "Computing"}

@string{Computer = "{IEEE} Computer"}

@string{Academic = "Academic Press"}

@string{AMS = "American Mathematical Society"}

@string{Springer = "Springer-Verlag"}

@inproceedings{InfiniFilter,
  author = {Dayan, Niv and Bercea, Ioana O. and Reviriego, Pedro and Pagh,
            Rasmus },
  title = {InfiniFilter: Expanding Filters to Infinity and Beyond},
  year = {2023},
  issue_date = {June 2023},
  publisher = {ACM},
  address = {New York, NY, USA},
  volume = {1},
  number = {2},
  url = {https://doi.org/10.1145/3589285},
  doi = {10.1145/3589285},
  booktitle = {Proceedings of the 2023 International Conference on Management of
               Data},
  month = {jun},
  articleno = {140},
  numpages = {27},
  location = {Seattle, Washington, USA},
  series = {SIGMOD '23},
}

@article{AlephFilter,
  author = {Dayan, Niv and Bercea, Ioana O. and Pagh, Rasmus},
  title = {Aleph Filter: To Infinity in Constant Time},
  year = {2024},
  issue_date = {July 2024},
  publisher = {VLDB Endowment},
  volume = {17},
  number = {11},
  issn = {2150-8097},
  url = {https://doi.org/10.14778/3681954.3682027},
  doi = {10.14778/3681954.3682027},
  journal = {Proc. VLDB Endow.},
  month = jul,
  pages = {3644–3656},
  numpages = {13},
}

@inproceedings{Zeno,
  author = {Kim, Hyuhng Min and Eslami, Navid and Dayan, Niv},
  title = {Zeno Filter: To Infinity in Tiny Steps},
  year = {2026},
  publisher = {ACM},
  address = {New York, NY, USA},
  booktitle = {Proceedings of the 2026 International Conference on Management of
               Data},
  numpages = {24},
  location = {Bangalore, India},
  series = {SIGMOD '26},
}

@inproceedings{BeyondBloomTutorial,
  author = {Pandey, Prashant and Farach-Colton, Mart\'{\i}n and Dayan, Niv and
            Zhang, Huanchen},
  title = {Beyond Bloom: A Tutorial on Future Feature-Rich Filters},
  year = {2024},
  isbn = {9798400704222},
  publisher = {ACM},
  address = {New York, NY, USA},
  url = {https://doi.org/10.1145/3626246.3654681},
  doi = {10.1145/3626246.3654681},
  booktitle = {Companion of the 2024 International Conference on Management of
               Data},
  pages = {636–644},
  numpages = {9},
  location = {Santiago AA, Chile},
  series = {SIGMOD/PODS '24},
}

@article{TaffyFilters,
  author = {Apple, Jim},
  title = {Stretching your data with taffy filters},
  journal = {Software: Practice and Experience},
  year = {2022},
}

@article{PaghExpandability,
  title = {How to Approximate a Set without Knowing Its Size in Advance},
  author = {Rasmus Pagh and Gil Segev and Udi Wieder},
  journal = {2013 IEEE 54th Annual Symposium on Foundations of Computer Science},
  year = {2013},
  pages = {80-89},
  url = {https://api.semanticscholar.org/CorpusID:10365891},
}

@inproceedings{Poppy,
  author = "Zhou, Dong and Andersen, David G. and Kaminsky, Michael",
  editor = "Bonifaci, Vincenzo and Demetrescu, Camil and Marchetti-Spaccamela,
            Alberto",
  title = "Space-Efficient, High-Performance Rank and Select Structures on
           Uncompressed Bit Sequences",
  booktitle = "Experimental Algorithms",
  year = "2013",
  publisher = "Springer Berlin Heidelberg",
  address = "Berlin, Heidelberg",
  pages = "151--163",
  isbn = "978-3-642-38527-8",
}

@inproceedings{FlorianRankandSelect,
  author = {Kurpicz, Florian},
  year = {2022},
  title = {Engineering Compact Data Structures for Rank and Select Queries on
           Bit Vectors},
  pages = {257–272},
  eventtitle = {29th International Symposium on String Processing and
                Information Retrieval},
  eventtitleaddon = {SPIRE 2022},
  eventdate = {2022-11-08/2022-11-10},
  venue = {Concepci{\'{o}}n, Chile},
  booktitle = {String Processing and Information Retrieval {\textendash } 29th
               International Symposium, SPIRE 2022, Concepci{\'{o}}n, Chile,
               November 8{\textendash }10, 2022, Proceedings. Ed.: D. Arroyuelo},
  doi = {10.1007/978-3-031-20643-6_19},
  publisher = {Springer International Publishing},
  isbn = {978-3-031-20643-6},
  issn = {0302-9743},
  series = {Lecture Notes in Computer Science},
  language = {english},
  volume = {13617},
}

@phdthesis{CompactPATTrees,
  author = {Clark, David},
  title = {Compact PAT trees},
  year = {1997},
  publisher = "UWSpace",
  url = {http://hdl.handle.net/10012/64},
}

@inproceedings{RankandSelectDict,
  author = {Okanohara, Daisuke and Sadakane, Kunihiko},
  title = {Practical entropy-compressed rank/select dictionary},
  year = {2007},
  publisher = {Society for Industrial and Applied Mathematics},
  address = {USA},
  booktitle = {Proceedings of the Meeting on Algorithm Engineering \&
               Expermiments},
  pages = {60–70},
  numpages = {11},
  location = {New Orleans, Louisiana},
}

@inproceedings{Succincter,
  author = {Patrascu, Mihai},
  title = {Succincter},
  year = {2008},
  isbn = {9780769534367},
  publisher = {IEEE Computer Society},
  address = {USA},
  url = {https://doi.org/10.1109/FOCS.2008.83},
  doi = {10.1109/FOCS.2008.83},
  booktitle = {Proceedings of the 2008 49th Annual IEEE Symposium on Foundations
               of Computer Science},
  pages = {305–313},
  numpages = {9},
  series = {FOCS '08},
}

@inproceedings{GQF,
  author = {Pandey, Prashant and Bender, Michael A. and Johnson, Rob and Patro,
            Rob},
  title = {A General-Purpose Counting Filter: Making Every Bit Count},
  year = {2017},
  isbn = {9781450341974},
  publisher = {ACM},
  address = {New York, NY, USA},
  url = {https://doi.org/10.1145/3035918.3035963},
  doi = {10.1145/3035918.3035963},
  booktitle = {Proceedings of the 2017 ACM International Conference on
               Management of Data},
  pages = {775–787},
  numpages = {13},
  location = {Chicago, Illinois, USA},
  series = {SIGMOD '17},
}

@article{CountMin,
  author = {Cormode, Graham and Muthukrishnan, S.},
  title = {An improved data stream summary: the count-min sketch and its
           applications},
  year = {2005},
  issue_date = {April 2005},
  publisher = {Academic Press, Inc.},
  address = {USA},
  volume = {55},
  number = {1},
  issn = {0196-6774},
  url = {https://doi.org/10.1016/j.jalgor.2003.12.001},
  doi = {10.1016/j.jalgor.2003.12.001},
  journal = {J. Algorithms},
  month = apr,
  pages = {58–75},
  numpages = {18},
}

@article{CountSketch,
  title = {Finding frequent items in data streams},
  journal = {Theoretical Computer Science},
  volume = {312},
  number = {1},
  pages = {3-15},
  year = {2004},
  note = {Automata, Languages and Programming},
  issn = {0304-3975},
  doi = {https://doi.org/10.1016/S0304-3975(03)00400-6},
  url = {https://www.sciencedirect.com/science/article/pii/S0304397503004006},
  author = {Moses Charikar and Kevin Chen and Martin Farach-Colton},
}

@inproceedings{CountMeanMin,
  title = {New Estimation Algorithms for Streaming Data: Count-min Can Do More},
  author = {Fan Deng and Davood Rafiei},
  url = {https://webdocs.cs.ualberta.ca/~drafiei/papers/cmm.pdf},
}

@techreport{MisraGries,
  author = {Misra, Jayadev and Gries, David},
  title = {Finding Repeated Elements},
  year = {1982},
  publisher = {Cornell University},
  address = {USA},
}

@article{StingySketch,
  author = {Li, Haoyu and Chen, Qizhi and Zhang, Yixin and Yang, Tong and Cui,
            Bin},
  title = {Stingy sketch: a sketch framework for accurate and fast frequency
           estimation},
  year = {2022},
  issue_date = {March 2022},
  publisher = {VLDB Endowment},
  volume = {15},
  number = {7},
  issn = {2150-8097},
  url = {https://doi.org/10.14778/3523210.3523220},
  doi = {10.14778/3523210.3523220},
  journal = {Proc. VLDB Endow.},
  month = mar,
  pages = {1426–1438},
  numpages = {13},
}

@article{TailoredSketch,
  author = {Gao, Guoju and Qian, Zhaorong and Huang, He and Sun, Yu-E and Du,
            Yang},
  journal = {IEEE Transactions on Network Science and Engineering},
  title = {TailoredSketch: A Fast and Adaptive Sketch for Efficient Per-Flow
           Size Measurement},
  year = {2025},
  volume = {12},
  number = {1},
  pages = {505-517},
  doi = {10.1109/TNSE.2024.3503904},
}

@inproceedings{SpaceSaving,
  author = {Metwally, Ahmed and Agrawal, Divyakant and El Abbadi, Amr},
  title = {Efficient computation of frequent and top-k elements in data streams},
  year = {2005},
  isbn = {3540242880},
  publisher = {Springer-Verlag},
  address = {Berlin, Heidelberg},
  url = {https://doi.org/10.1007/978-3-540-30570-5_27},
  doi = {10.1007/978-3-540-30570-5_27},
  booktitle = {Proceedings of the 10th International Conference on Database
               Theory},
  pages = {398–412},
  numpages = {15},
  location = {Edinburgh, UK},
  series = {ICDT'05},
}

@inproceedings{UnbiasedSpaceSaving,
  author = {Ting, Daniel},
  title = {Data Sketches for Disaggregated Subset Sum and Frequent Item
           Estimation},
  year = {2018},
  isbn = {9781450347037},
  publisher = {ACM},
  address = {New York, NY, USA},
  url = {https://doi.org/10.1145/3183713.3183759},
  doi = {10.1145/3183713.3183759},
  booktitle = {Proceedings of the 2018 International Conference on Management of
               Data},
  pages = {1129–1140},
  numpages = {12},
  location = {Houston, TX, USA},
  series = {SIGMOD '18},
}

@article{ConvolutionSketch,
  author = {Heddes, Mike and Nunes, Igor and Givargis, Tony and Nicolau, Alex},
  title = {Convolution and Cross-Correlation of Count Sketches Enables Fast
           Cardinality Estimation of Multi-Join Queries},
  year = {2024},
  issue_date = {June 2024},
  publisher = {ACM},
  address = {New York, NY, USA},
  volume = {2},
  number = {3},
  url = {https://doi.org/10.1145/3654932},
  doi = {10.1145/3654932},
  journal = {Proc. ACM Manag. Data},
  month = may,
  articleno = {129},
  numpages = {26},
}

@inproceedings{Compass,
  author = {Izenov, Yesdaulet and Datta, Asoke and Rusu, Florin and Shin, Jun
            Hyung},
  title = {COMPASS: Online Sketch-based Query Optimization for In-Memory
           Databases},
  year = {2021},
  isbn = {9781450383431},
  publisher = {ACM},
  address = {New York, NY, USA},
  url = {https://doi.org/10.1145/3448016.3452840},
  doi = {10.1145/3448016.3452840},
  booktitle = {Proceedings of the 2021 International Conference on Management of
               Data},
  pages = {804–816},
  numpages = {13},
  location = {Virtual Event, China},
  series = {SIGMOD '21},
}

@phdthesis{QueryOptimizationSketches,
  title = {Query Optimization using Sketches in Relational Database Systems},
  author = {Yesdaulet Izenov},
  year = 2023,
  month = {Fall},
  address = {Merced, CA},
  note = {Available at \url{https://escholarship.org/uc/item/5k8334fs}},
  school = {University of California, Merced},
  type = {PhD thesis},
}

@online{DataSketches,
  author = {Rhodes Lee and Lang, Kevin Saydakov, Alaxander and Thaler, Justin
            and Liberty, Edo and Malkin, Jon},
  title = {DataSketches: A Java software library of stochastic streaming
           algorithms},
  year = {2025},
  url = {https://datasketches.apache.org/},
  urldate = {2025-04-15},
  journal = {DataSketches},
}

@inproceedings{BatchUpdateMisraGries,
  author = {Anderson, Daniel and Bevan, Pryce and Lang, Kevin and Liberty, Edo
            and Rhodes, Lee and Thaler, Justin},
  title = {A high-performance algorithm for identifying frequent items in data
           streams},
  year = {2017},
  isbn = {9781450351188},
  publisher = {ACM},
  address = {New York, NY, USA},
  url = {https://doi.org/10.1145/3131365.3131407},
  doi = {10.1145/3131365.3131407},
  booktitle = {Proceedings of the 2017 Internet Measurement Conference},
  pages = {268–282},
  numpages = {15},
  location = {London, United Kingdom},
  series = {IMC '17},
}

@inproceedings{AugmentedSketch,
  author = {Roy, Pratanu and Khan, Arijit and Alonso, Gustavo},
  title = {Augmented Sketch: Faster and More Accurate Stream Processing},
  year = {2016},
  isbn = {9781450335317},
  publisher = {ACM},
  address = {New York, NY, USA},
  url = {https://doi.org/10.1145/2882903.2882948},
  doi = {10.1145/2882903.2882948},
  booktitle = {Proceedings of the 2016 International Conference on Management of
               Data},
  pages = {1449–1463},
  numpages = {15},
  location = {San Francisco, California, USA},
  series = {SIGMOD '16},
}

@inproceedings{BitMatcher,
  author = {Shi, Qilong and Jia, Chengjun and Li, Wenjun and Liu, Zaoxing and
            Yang, Tong and Ji, Jianan and Xie, Gaogang and Zhang, Weizhe and Yu,
            Minlan},
  booktitle = {2024 IEEE 40th International Conference on Data Engineering
               (ICDE)},
  title = {BitMatcher: Bit-level Counter Adjustment for Sketches},
  year = {2024},
  volume = {},
  number = {},
  pages = {4815-4827},
  doi = {10.1109/ICDE60146.2024.00366},
}

@inproceedings{HeavyGuardian,
  author = {Yang, Tong and Gong, Junzhi and Zhang, Haowei and Zou, Lei and Shi,
            Lei and Li, Xiaoming},
  title = {HeavyGuardian: Separate and Guard Hot Items in Data Streams},
  year = {2018},
  isbn = {9781450355520},
  publisher = {ACM},
  address = {New York, NY, USA},
  url = {https://doi.org/10.1145/3219819.3219978},
  doi = {10.1145/3219819.3219978},
  booktitle = {Proceedings of the 24th ACM SIGKDD International Conference on
               Knowledge Discovery \& Data Mining},
  pages = {2584–2593},
  numpages = {10},
  location = {London, United Kingdom},
  series = {KDD '18},
}

@article{HeavyKeeper,
  author = {Yang, Tong and Zhang, Haowei and Li, Jinyang and Gong, Junzhi and
            Uhlig, Steve and Chen, Shigang and Li, Xiaoming},
  title = {HeavyKeeper: An Accurate Algorithm for Finding Top-$k$ Elephant Flows
           },
  year = {2019},
  issue_date = {October 2019},
  publisher = {IEEE Press},
  volume = {27},
  number = {5},
  issn = {1063-6692},
  url = {https://doi.org/10.1109/TNET.2019.2933868},
  doi = {10.1109/TNET.2019.2933868},
  journal = {IEEE/ACM Trans. Netw.},
  month = oct,
  pages = {1845–1858},
  numpages = {14},
}

@article{JigsawSketch,
  author = {Zhang, Boyu and Huang, He and Sun, Yu-E. and Du, Yang and Wang, Dan},
  title = {Jigsaw-Sketch: a fast and accurate algorithm for finding top-k
           elephant flows in high-speed networks},
  journal = {Science China Information Sciences},
  year = {2024},
  month = {Mar},
  day = {20},
  volume = {67},
  number = {4},
  pages = {142101},
  issn = {1869-1919},
  doi = {10.1007/s11432-022-3794-1},
  url = {https://doi.org/10.1007/s11432-022-3794-1},
}

@inproceedings{MVSketch,
  author = {Tang, Lu and Huang, Qun and Lee, Patrick P. C.},
  title = {MV-Sketch: A Fast and Compact Invertible Sketch for Heavy Flow
           Detection in Network Data Streams},
  year = {2019},
  publisher = {IEEE Press},
  url = {https://doi.org/10.1109/INFOCOM.2019.8737499},
  doi = {10.1109/INFOCOM.2019.8737499},
  booktitle = {IEEE INFOCOM 2019 - IEEE Conference on Computer Communications},
  pages = {2026–2034},
  numpages = {9},
  location = {Paris, France},
}

@inproceedings{WavingSketch,
  author = {Li, Jizhou and Li, Zikun and Xu, Yifei and Jiang, Shiqi and Yang,
            Tong and Cui, Bin and Dai, Yafei and Zhang, Gong},
  title = {WavingSketch: An Unbiased and Generic Sketch for Finding Top-k Items
           in Data Streams},
  year = {2020},
  isbn = {9781450379984},
  publisher = {ACM},
  address = {New York, NY, USA},
  url = {https://doi.org/10.1145/3394486.3403208},
  doi = {10.1145/3394486.3403208},
  booktitle = {Proceedings of the 26th ACM SIGKDD International Conference on
               Knowledge Discovery \& Data Mining},
  pages = {1574–1584},
  numpages = {11},
  location = {Virtual Event, CA, USA},
  series = {KDD '20},
}

@inproceedings{ElasticSketch,
  author = {Yang, Tong and Jiang, Jie and Liu, Peng and Huang, Qun and Gong,
            Junzhi and Zhou, Yang and Miao, Rui and Li, Xiaoming and Uhlig, Steve
            },
  title = {Elastic sketch: adaptive and fast network-wide measurements},
  year = {2018},
  isbn = {9781450355674},
  publisher = {ACM},
  address = {New York, NY, USA},
  url = {https://doi.org/10.1145/3230543.3230544},
  doi = {10.1145/3230543.3230544},
  booktitle = {Proceedings of the 2018 Conference of the ACM Special Interest
               Group on Data Communication},
  pages = {561–575},
  numpages = {15},
  location = {Budapest, Hungary},
  series = {SIGCOMM '18},
}

@inproceedings{MicroscopeSketch,
  author = {Wu, Yuhan and Jiang, Shiqi and Dong, Siyuan and Zhong, Zheng and
            Chen, Jiale and Hu, Yutong and Yang, Tong and Uhlig, Steve and Cui,
            Bin},
  title = {MicroscopeSketch: Accurate Sliding Estimation Using Adaptive Zooming},
  year = {2023},
  isbn = {9798400701030},
  publisher = {ACM},
  address = {New York, NY, USA},
  url = {https://doi.org/10.1145/3580305.3599432},
  doi = {10.1145/3580305.3599432},
  booktitle = {Proceedings of the 29th ACM SIGKDD Conference on Knowledge
               Discovery and Data Mining},
  pages = {2660–2671},
  numpages = {12},
  location = {Long Beach, CA, USA},
  series = {KDD '23},
}

@article{PSketch,
  author = {Li, Weihe and Patras, Paul},
  title = {P-Sketch: A Fast and Accurate Sketch for Persistent Item Lookup},
  year = {2023},
  issue_date = {April 2024},
  publisher = {IEEE Press},
  volume = {32},
  number = {2},
  issn = {1063-6692},
  url = {https://doi.org/10.1109/TNET.2023.3306897},
  doi = {10.1109/TNET.2023.3306897},
  journal = {IEEE/ACM Trans. Netw.},
  month = aug,
  pages = {987–1002},
  numpages = {16},
}

@inbook{StreamModels,
  author = {Golab, Lukasz},
  editor = {LIU, LING and {\"O}ZSU, M. TAMER},
  title = {Stream Models},
  bookTitle = {Encyclopedia of Database Systems},
  year = {2009},
  publisher = {Springer US},
  address = {Boston, MA},
  pages = {2834--2836},
  isbn = {978-0-387-39940-9},
  doi = {10.1007/978-0-387-39940-9_370},
  url = {https://doi.org/10.1007/978-0-387-39940-9_370},
}

@article{CMSkmer,
  title = {These are not the k-mers you are looking for: efficient online k-mer
           counting using a probabilistic data structure},
  author = {Zhang, Qingpeng and Pell, Jason and Canino-Koning, Rosangela and
            Howe, Adina Chuang and Brown, C Titus},
  journal = {PloS one},
  volume = {9},
  number = {7},
  pages = {e101271},
  year = {2014},
  publisher = {Public Library of Science San Francisco, USA},
}

@article{SetMinSketch,
  author = {Shibuya, Yoshihiro and Belazzougui, Djamal and Kucherov, Gregory},
  title = {Set-Min sketch: a probabilistic map for power-law distributions with
           application to k-mer annotation},
  elocation-id = {2020.11.14.382713},
  year = {2021},
  doi = {10.1101/2020.11.14.382713},
  publisher = {Cold Spring Harbor Laboratory},
  URL = {https://www.biorxiv.org/content/early/2021/02/25/2020.11.14.382713},
  eprint = {
            https://www.biorxiv.org/content/early/2021/02/25/2020.11.14.382713.full.pdf
            },
  journal = {bioRxiv},
}

@inproceedings{SketchBasedChangeDetection,
  author = {Krishnamurthy, Balachander and Sen, Subhabrata and Zhang, Yin and
            Chen, Yan},
  title = {Sketch-based change detection: methods, evaluation, and applications},
  year = {2003},
  isbn = {1581137737},
  publisher = {ACM},
  address = {New York, NY, USA},
  url = {https://doi.org/10.1145/948205.948236},
  doi = {10.1145/948205.948236},
  booktitle = {Proceedings of the 3rd ACM SIGCOMM Conference on Internet
               Measurement},
  pages = {234–247},
  numpages = {14},
  location = {Miami Beach, FL, USA},
  series = {IMC '03},
}

@article{WhatsNew,
  author = {Cormode, G. and Muthukrishnan, S.},
  journal = {IEEE/ACM Transactions on Networking},
  title = {What's new: finding significant differences in network data streams},
  year = {2005},
  volume = {13},
  number = {6},
  pages = {1219-1232},
  doi = {10.1109/TNET.2005.860096},
}

@inproceedings{TrafficManagementSurvey,
  title = {A survey of sketches in traffic measurement: Design, Optimization,
           Application and Implementation},
  author = {Shangsen Li and Lailong Luo and Deke Guo and Qianzhen Zhang and
            Pengtao Fu},
  year = {2020},
  url = {https://api.semanticscholar.org/CorpusID:236140679},
}

@inproceedings{CoMeT,
  author = {Bostanci, F. and Yüksel, Ismail and Olgun, Ataberk and Kanellopoulos
            , Konstantinos and Tuğrul, Yahya and Yaglikçi, Abdullah Giray and
            Sadrosadati, Mohammad and Mutlu, Onur},
  year = {2024},
  month = {03},
  pages = {593-612},
  title = {CoMeT: Count-Min-Sketch-based Row Tracking to Mitigate RowHammer at
           Low Cost},
  doi = {10.1109/HPCA57654.2024.00050},
}

@inproceedings{ASCS,
  author = {Dai, Zhenwei and Desai, Aditya and Heckel, Reinhard and Shrivastava,
            Anshumali},
  title = {Active Sampling Count Sketch (ASCS) for Online Sparse Estimation of a
           Trillion Scale Covariance Matrix},
  year = {2021},
  isbn = {9781450383431},
  publisher = {ACM},
  address = {New York, NY, USA},
  url = {https://doi.org/10.1145/3448016.3457327},
  doi = {10.1145/3448016.3457327},
  booktitle = {Proceedings of the 2021 International Conference on Management of
               Data},
  pages = {352–364},
  numpages = {13},
  location = {Virtual Event, China},
  series = {SIGMOD '21},
}

@article{Horner,
  title = {IX. A new method of solving numerical equations of all orders, by
           continuous approximation},
  author = {Horner, William George},
  journal = {Philosophical Transactions of the Royal Society of London},
  volume = {109},
  pages = {308--335},
  year = {1819},
  publisher = {The Royal Society},
}

@book{HornerKnuth,
  author = {Knuth, Donald E.},
  title = {The Art of Computer Programming, Volume 2 (3rd Ed.): Seminumerical
           Algorithms},
  year = {1997},
  publisher = {Addison-Wesley Professional},
  address = {USA},
}

@book{HackersDelight,
  author = {Warren, Jr., Henry S.},
  title = {Hacker's Delight},
  edition = {2nd},
  year = {2012},
  publisher = {Addison-Wesley Professional},
  address = {Upper Saddle River, NJ},
  isbn = {978-0321842688},
}

@inproceedings{BitSense,
  author = {Ding, Rui and Yang, Shibo and Chen, Xiang and Huang, Qun},
  title = {BitSense: Universal and Nearly Zero-Error Optimization for Sketch
           Counters with Compressive Sensing},
  year = {2023},
  isbn = {9798400702365},
  publisher = {ACM},
  address = {New York, NY, USA},
  url = {https://doi.org/10.1145/3603269.3604865},
  doi = {10.1145/3603269.3604865},
  booktitle = {Proceedings of the ACM SIGCOMM 2023 Conference},
  pages = {220–238},
  numpages = {19},
  location = {New York, NY, USA},
  series = {ACM SIGCOMM '23},
}

@inproceedings{SALSA,
  author = {Basat, Ran and Einziger, Gil and Mitzenmacher, Michael and Vargaftik
            , Shay},
  year = {2021},
  month = {04},
  pages = {864-875},
  title = {SALSA: Self-Adjusting Lean Streaming Analytics},
  doi = {10.1109/ICDE51399.2021.00080},
}

@article{PyramidSketch,
  author = {Yang, Tong and Zhou, Yang and Jin, Hao and Chen, Shigang and Li,
            Xiaoming},
  title = {Pyramid sketch: a sketch framework for frequency estimation of data
           streams},
  year = {2017},
  issue_date = {August 2017},
  publisher = {VLDB Endowment},
  volume = {10},
  number = {11},
  issn = {2150-8097},
  url = {https://doi.org/10.14778/3137628.3137652},
  doi = {10.14778/3137628.3137652},
  journal = {Proc. VLDB Endow.},
  month = aug,
  pages = {1442–1453},
  numpages = {12},
}

@inproceedings{AdaptiveCounterSplicing,
  author = {Gao, Guoju and Qian, Zhaorong and Huang, He and Du, Yang},
  year = {2023},
  month = {06},
  pages = {1-4},
  title = {An Adaptive Counter-Splicing-Based Sketch for Efficient Per-Flow Size
           Measurement},
  doi = {10.1109/IWQoS57198.2023.10188733},
}

@article{CounterTree,
  author = {Chen, Min and Chen, Shigang and Cai, Zhiping},
  title = {Counter Tree: A Scalable Counter Architecture for Per-Flow Traffic
           Measurement},
  year = {2017},
  issue_date = {April 2017},
  publisher = {IEEE Press},
  volume = {25},
  number = {2},
  issn = {1063-6692},
  url = {https://doi.org/10.1109/TNET.2016.2621159},
  doi = {10.1109/TNET.2016.2621159},
  journal = {IEEE/ACM Trans. Netw.},
  month = apr,
  pages = {1249–1262},
  numpages = {14},
}

@article{TreeSensing,
  author = {Liu, Zirui and Zhang, Yixin and Zhu, Yifan and Zhang, Ruwen and Yang
            , Tong and Xie, Kun and Wang, Sha and Li, Tao and Cui, Bin},
  title = {TreeSensing: Linearly Compressing Sketches with Flexibility},
  year = {2023},
  issue_date = {May 2023},
  publisher = {ACM},
  address = {New York, NY, USA},
  volume = {1},
  number = {1},
  url = {https://doi.org/10.1145/3588910},
  doi = {10.1145/3588910},
  journal = {Proc. ACM Manag. Data},
  month = may,
  articleno = {56},
  numpages = {28},
}

@inproceedings{TightSketch,
  author = {Li, Weihe and Patras, Paul},
  title = {Tight-Sketch: A High-Performance Sketch for Heavy Item-Oriented Data
           Stream Mining with Limited Memory Size},
  year = {2023},
  isbn = {9798400701245},
  publisher = {ACM},
  address = {New York, NY, USA},
  url = {https://doi.org/10.1145/3583780.3615080},
  doi = {10.1145/3583780.3615080},
  booktitle = {Proceedings of the 32nd ACM International Conference on
               Information and Knowledge Management},
  pages = {1328–1337},
  numpages = {10},
  location = {Birmingham, United Kingdom},
  series = {CIKM '23},
}

@inproceedings{CodingSketch,
  author = {Qizhi Chen and Yisen Hong and Yuhan Wu and Tong Yang and Bin Cui},
  title = {CodingSketch: A Hierarchical Sketch with Efficient Encoding and
           Recursive Decoding},
  year = {2024},
  cdate = {1704067200000},
  pages = {1592-1605},
  url = {https://doi.org/10.1109/ICDE60146.2024.00130},
  booktitle = {ICDE},
}

@inbook{ADSMassiveDatasetBook,
  author = {Medjedovic, Dzejla and Tahirovic, Emin and Dedovic, Ines},
  place = {Shelter Island, NY},
  title = {Algorithms and data structures for massive datasets},
  publisher = {Manning Publications Co.},
  chapter = {4},
  year = 2025,
}

@inproceedings{SketchesTrafficManagementSurvey,
  title = {A survey of sketches in traffic measurement: Design, Optimization,
           Application and Implementation},
  author = {Shangsen Li and Lailong Luo and Deke Guo and Qianzhen Zhang and
            Pengtao Fu},
  year = {2020},
  url = {https://api.semanticscholar.org/CorpusID:236140679},
}

@online{PDEP,
  author = {Felix Cloutier},
  title = {PDEP — Parallel Bits Deposit},
  year = {2023},
  url = {https://www.felixcloutier.com/x86/pdep},
  urldate = {2025-08-24},
}

@article{AMS,
  author = {Noga Alon and Yossi Matias and Mario Szegedy},
  title = {The space complexity of approximating the frequency moments},
  journal = {Journal of Computer and System Sciences},
  volume = {58},
  number = {1},
  pages = {137--147},
  year = {1999},
  publisher = {Elsevier},
  doi = {10.1006/jcss.1998.1627},
}

@inbook{SensorDataAnalytics,
  author = {Aggarwal, Charu C.},
  editor = {Aggarwal, Charu C.},
  title = {An Introduction to Sensor Data Analytics},
  bookTitle = {Managing and Mining Sensor Data},
  year = {2013},
  publisher = {Springer US},
  address = {Boston, MA},
  pages = {1--8},
  isbn = {978-1-4614-6309-2},
  doi = {10.1007/978-1-4614-6309-2_1},
  url = {https://doi.org/10.1007/978-1-4614-6309-2_1},
}

@article{RobustAggregationSensorNetworks,
  title = {Robust Aggregation in Sensor Networks},
  author = {George Kollios and John Byers and Jeffrey Considine and Marios
            Hadjieleftheriou and Feifei Li},
  journal = {Bulletin of the IEEE Computer Society Technical Committee on Data
             Engineering},
  year = {2005},
  note = {Special Issue on Data Management in Sensor Networks},
  url = {
         https://www2.cs.arizona.edu/classes/cs645/fall05/cs645-papers/DataAggregation/dataC.pdf
         },
}

@article{RandomizedAdmissionEstimation,
  author = {Ben Basat, Ran and Chen, Xiaoqi and Einziger, Gil and Friedman, Roy
            and Kassner, Yaron},
  title = {Randomized Admission Policy for Efficient Top-k, Frequency, and
           Volume Estimation},
  year = {2019},
  issue_date = {August 2019},
  publisher = {IEEE Press},
  volume = {27},
  number = {4},
  issn = {1063-6692},
  url = {https://doi.org/10.1109/TNET.2019.2918929},
  doi = {10.1109/TNET.2019.2918929},
  journal = {IEEE/ACM Trans. Netw.},
  month = aug,
  pages = {1432–1445},
  numpages = {14},
}

@inproceedings{SketchingStreamsNet,
  author = {Cormode, Graham and Garofalakis, Minos},
  title = {Sketching streams through the net: distributed approximate query
           tracking},
  year = {2005},
  isbn = {1595931546},
  publisher = {VLDB Endowment},
  booktitle = {Proceedings of the 31st International Conference on Very Large
               Data Bases},
  pages = {13–24},
  numpages = {12},
  location = {Trondheim, Norway},
  series = {VLDB '05},
}

@inproceedings{Blink,
  author = {Holterbach, Thomas and Molero, Edgar Costa and Apostolaki, Maria and
            Dainotti, Alberto and Vissicchio, Stefano and Vanbever, Laurent},
  title = {Blink: fast connectivity recovery entirely in the data plane},
  year = {2019},
  isbn = {9781931971492},
  publisher = {USENIX Association},
  address = {USA},
  booktitle = {Proceedings of the 16th USENIX Conference on Networked Systems
               Design and Implementation},
  pages = {161–176},
  numpages = {16},
  location = {Boston, MA, USA},
  series = {NSDI'19},
}

@inproceedings{LOFT,
  author = {Scherrer, Simon and Wu, Che-Yu and Chiang, Yu-Hsi and Rothenberger,
            Benjamin and Asoni, Daniele and Sateesan, Arish and Vliegen, Jo and
            Mentens, Nele and Hsiao, Hsu-Chun and Perrig, Adrian},
  year = {2021},
  month = {09},
  pages = {265-276},
  title = {Low-Rate Overuse Flow Tracer (LOFT): An Efficient and Scalable
           Algorithm for Detecting Overuse Flows},
  doi = {10.1109/SRDS53918.2021.00034},
}

@article{Flink,
  title = {Apache Flink{\texttrademark}: Stream and Batch Processing in a Single
           Engine},
  author = {Paris Carbone and Asterios Katsifodimos and Stephan Ewen and Volker
            Markl and Seif Haridi and Kostas Tzoumas},
  year = {2015},
  language = {English},
  volume = {36},
  pages = {28--38},
  journal = {Bulletin of the IEEE Computer Society Technical Committee on Data
             Engineering},
  publisher = {IEEE},
  number = {4},
}

@inproceedings{AdditiveErrorCounters,
  author = {Basat, Ran Ben and Einziger, Gil and Mitzenmacher, Michael and
            Vargaftik, Shay},
  booktitle = {IEEE INFOCOM 2020 - IEEE Conference on Computer Communications},
  title = {Faster and More Accurate Measurement through Additive-Error Counters},
  year = {2020},
  volume = {},
  number = {},
  pages = {1251-1260},
  doi = {10.1109/INFOCOM41043.2020.9155340},
}

@inproceedings{NetworkAnomalyDetection1,
  author = {Gu, Yu and McCallum, Andrew and Towsley, Don},
  title = {Detecting anomalies in network traffic using maximum entropy
           estimation},
  year = {2005},
  publisher = {USENIX Association},
  address = {USA},
  booktitle = {Proceedings of the 5th ACM SIGCOMM Conference on Internet
               Measurement},
  pages = {32},
  numpages = {1},
  location = {Berkeley, CA},
  series = {IMC '05},
}

@inproceedings{NetworkAnomalyDetection2,
  author = {Wagner, A. and Plattner, B.},
  booktitle = {14th IEEE International Workshops on Enabling Technologies:
               Infrastructure for Collaborative Enterprise (WETICE'05)},
  title = {Entropy based worm and anomaly detection in fast IP networks},
  year = {2005},
  pages = {172-177},
  doi = {10.1109/WETICE.2005.35},
}

@article{CoDDoS,
  author = {Xia, Jiqiang and Tian, Le and Hu, Yuxiang and Li, Ziyong and Sun,
            Penghao and Peng, Jianhua},
  title = {CoDDoS: Detecting and mitigating diverse DDoS attacks with
           programmable switches},
  year = {2025},
  issue_date = {Aug 2025},
  publisher = {Elsevier Science Publishers B. V.},
  address = {NLD},
  volume = {240},
  number = {C},
  issn = {0140-3664},
  url = {https://doi.org/10.1016/j.comcom.2025.108215},
  doi = {10.1016/j.comcom.2025.108215},
  journal = {Comput. Commun.},
  month = aug,
  numpages = {13},
}

@article{Memento,
  author = {Ben Basat, Ran and Einziger, Gil and Keslassy, Isaac and Orda, Ariel
            and Vargaftik, Shay and Waisbard, Erez},
  journal = {IEEE/ACM Transactions on Networking},
  title = {Memento: Making Sliding Windows Efficient for Heavy Hitters},
  year = {2022},
  volume = {30},
  number = {4},
  pages = {1440-1453},
  doi = {10.1109/TNET.2021.3132385},
}

@online{SimonsYahoo,
  author = {Rhodes, Lee},
  title = {Insights from Engineering Sketches for Production and Using Sketches
           at Scale},
  year = {2023},
  url = {https://simons.berkeley.edu/talks/lee-rhodes-yahoo-inc-2023-10-12},
  urldate = {2025-10-10},
}

@misc{Kosarak,
  author = {Bodon, Ferenc},
  title = {Frequent Itemset Mining Dataset Repository, Kosarak},
  year = {2012},
  url = {http://fimi.uantwerpen.be/data/},
  urldate = {2025-10-11},
}

@misc{WebDocs,
  author = {Lucchese, Claudio and Orlando, Salvatore and Perego, Raffaele and
            Silvestri, Fabrizio},
  title = {Frequent Itemset Mining Dataset Repository, Webdocs},
  year = {2012},
  url = {http://fimi.uantwerpen.be/data/},
  urldate = {2025-10-11},
}

@misc{CAIDA2018,
  title = {Anonymized Internet Traces 2018},
  howpublished = {\url{https://catalog.caida.org/dataset/passive_2018_pcap}},
  note = {Dates used: 2025-10-11. Accessed: 2025-10-11.},
  doi = {https://catalog.caida.org/dataset/passive_2018_pcap},
}

@misc{TPCH,
  title = {TPC Benchmark{\texttrademark} H (Decision Support)},
  author = {{Transaction Processing Performance Council (TPC)}},
  howpublished = {\url{https://www.tpc.org/tpch/}},
  note = {Accessed: 2025-12-03},
  year = {2024},
}

@misc{Linux5LevelPaging,
  author = {The kernel development community},
  title = {5-level paging},
  howpublished = {Linux Kernel Documentation},
  year = {2024},
  url = {https://docs.kernel.org/arch/x86/x86_64/5level-paging.html},
  note = {Accessed: 2026-01-19},
}

@techreport{Intel5LevelPaging,
  author = {Intel Corporation},
  title = {5-Level Paging and 5-Level EPT White Paper},
  institution = {Intel Corporation},
  year = {2018},
  month = {May},
  number = {335252-002},
  url = {
         https://www.intel.com/content/www/us/en/content-details/671442/5-level-paging-and-5-level-ept-white-paper.html
         },
  note = {Revision 1.1},
}

@manual{Intel4LevelPaging,
  author = {Intel Corporation},
  title = {Intel® 64 and IA-32 Architectures Software Developer’s Manual, Volume
           3A: System Programming Guide, Part 1},
  year = {2025},
  month = {October},
  note = {Order Number 253668},
  organization = {Intel Corporation},
  url = {
         https://www.intel.com/content/www/us/en/developer/articles/technical/intel-sdm.html
         },
}

@article{FastConcurrentDataSketches,
  author = {Rinberg, Arik and Spiegelman, Alexander and Bortnikov, Edward and
            Hillel, Eshcar and Keidar, Idit and Rhodes, Lee and Serviansky, Hadar
            },
  title = {Fast Concurrent Data Sketches},
  year = {2022},
  issue_date = {June 2022},
  publisher = {ACM},
  address = {New York, NY, USA},
  volume = {9},
  number = {2},
  issn = {2329-4949},
  url = {https://doi.org/10.1145/3512758},
  doi = {10.1145/3512758},
  journal = {ACM Trans. Parallel Comput.},
  month = apr,
  articleno = {6},
  numpages = {35},
}

@article{FPGASketching,
  author = {Kiefer, Martin and Poulakis, Ilias and Zacharatou, Eleni Tzirita and
            Markl, Volker},
  title = {Optimistic Data Parallelism for FPGA-Accelerated Sketching},
  year = {2023},
  issue_date = {January 2023},
  publisher = {VLDB Endowment},
  volume = {16},
  number = {5},
  issn = {2150-8097},
  url = {https://doi.org/10.14778/3579075.3579085},
  doi = {10.14778/3579075.3579085},
  journal = {Proc. VLDB Endow.},
  month = jan,
  pages = {1113–1125},
  numpages = {13},
}

@article{JoinSizeEstimationFilterConditions,
  author = {Vengerov, David and Menck, Andre Cavalheiro and Zait, Mohamed and
            Chakkappen, Sunil P.},
  title = {Join size estimation subject to filter conditions},
  year = {2015},
  issue_date = {August 2015},
  publisher = {VLDB Endowment},
  volume = {8},
  number = {12},
  issn = {2150-8097},
  url = {https://doi.org/10.14778/2824032.2824051},
  doi = {10.14778/2824032.2824051},
  journal = {Proc. VLDB Endow.},
  month = aug,
  pages = {1530–1541},
  numpages = {12},
}

@article{Diva,
  author = {Eslami, Navid and Bercea, Ioana O. and Dayan, Niv},
  title = {Diva: Dynamic Range Filter for Var-Length Keys and Queries},
  year = {2025},
  issue_date = {July 2025},
  publisher = {VLDB Endowment},
  volume = {18},
  number = {11},
  issn = {2150-8097},
  url = {https://doi.org/10.14778/3749646.3749664},
  doi = {10.14778/3749646.3749664},
  journal = {Proc. VLDB Endow.},
  month = jul,
  pages = {3923–3936},
  numpages = {14},
}

@inproceedings{Aeris,
  author = {Chesetti, Yuvaraj and Eslami, Navid and Zhang, Huanchen and Dayan,
            Niv and Pandey, Prashant},
  title = {Aeris Filter: A Strongly and Monotonically Adaptive Range Filter},
  year = {2026},
  issue_date = {June 2026},
  publisher = {ACM},
  address = {New York, NY, USA},
  booktitle = {Proceedings of the 2026 International Conference on Management of
               Data},
  pages = {1670–1684},
  numpages = {15},
  location = {Bangalore, India},
  series = {SIGMOD '26},
}

@inproceedings{OctoSketch,
  author = {Yinda Zhang and Peiqing Chen and Zaoxing Liu},
  title = {{OctoSketch}: Enabling {Real-Time}, Continuous Network Monitoring
           over Multiple Cores},
  booktitle = {21st USENIX Symposium on Networked Systems Design and
               Implementation (NSDI 24)},
  year = {2024},
  isbn = {978-1-939133-39-7},
  address = {Santa Clara, CA},
  pages = {1621--1639},
  url = {https://www.usenix.org/conference/nsdi24/presentation/zhang-yinda},
  publisher = {USENIX Association},
  month = apr,
}

@article{ScoutSketch,
  author = {Gao, Guoju and Ma, Tianyu and Huang, He and Sun, Yu-E and Wang,
            Haibo and Du, Yang and Chen, Shigang},
  title = {Scout Sketch+: Finding Both Promising and Damping Items
           Simultaneously in Data Streams},
  year = {2024},
  issue_date = {Dec. 2024},
  publisher = {IEEE Press},
  volume = {32},
  number = {6},
  issn = {1063-6692},
  url = {https://doi.org/10.1109/TNET.2024.3469196},
  doi = {10.1109/TNET.2024.3469196},
  journal = {IEEE/ACM Trans. Netw.},
  month = oct,
  pages = {5491–5506},
  numpages = {16},
}

@inproceedings{Canopy,
  author = {Kaldor, Jonathan and Mace, Jonathan and Bejda, Micha\l{} and Gao,
            Edison and Kuropatwa, Wiktor and O'Neill, Joe and Ong, Kian Win and
            Schaller, Bill and Shan, Pingjia and Viscomi, Brendan and
            Venkataraman, Vinod and Veeraraghavan, Kaushik and Song, Yee Jiun},
  title = {Canopy: An End-to-End Performance Tracing And Analysis System},
  year = {2017},
  isbn = {9781450350853},
  publisher = {ACM},
  address = {New York, NY, USA},
  url = {https://doi.org/10.1145/3132747.3132749},
  doi = {10.1145/3132747.3132749},
  booktitle = {Proceedings of the 26th Symposium on Operating Systems Principles
               },
  pages = {34–50},
  numpages = {17},
  location = {Shanghai, China},
  series = {SOSP '17},
}

@misc{SublimeArxiv,
  title = {Sublime: Sublinear Error \& Space for Unbounded Skewed Streams},
  author = {Navid Eslami and Ioana O. Bercea and Rasmus Pagh and Niv Dayan},
  year = {2026},
  eprint = {2603.14190},
  archivePrefix = {arXiv},
  primaryClass = {cs.DS},
  url = {https://arxiv.org/abs/2603.14190},
}

@article{Jensen,
  author = {Jensen, J. L. W. V.},
  title = {Sur les fonctions convexes et les inégalités entre les valeurs
           moyennes},
  journal = {Acta Mathematica},
  year = {1906},
  volume = {30},
  pages = {175--193},
  doi = {10.1007/bf02418571},
}

@inproceedings{AGM,
  author = {Ahn, Kook Jin and Guha, Sudipto and McGregor, Andrew},
  title = {Analyzing graph structure via linear measurements},
  year = {2012},
  publisher = {Society for Industrial and Applied Mathematics},
  address = {USA},
  booktitle = {Proceedings of the Twenty-Third Annual ACM-SIAM Symposium on
               Discrete Algorithms},
  pages = {459–467},
  numpages = {9},
  location = {Kyoto, Japan},
  series = {SODA '12},
}

@inproceedings{LpSamplers,
  author = {Jowhari, Hossein and Sa\u{g}lam, Mert and Tardos, G\'{a}bor},
  title = {Tight bounds for Lp samplers, finding duplicates in streams, and
           related problems},
  year = {2011},
  isbn = {9781450306607},
  publisher = {Association for Computing Machinery},
  address = {New York, NY, USA},
  url = {https://doi.org/10.1145/1989284.1989289},
  doi = {10.1145/1989284.1989289},
  booktitle = {Proceedings of the Thirtieth ACM SIGMOD-SIGACT-SIGART Symposium
               on Principles of Database Systems},
  pages = {49–58},
  numpages = {10},
  location = {Athens, Greece},
  series = {PODS '11},
}

@inproceedings{GK,
  author = {Greenwald, Michael and Khanna, Sanjeev},
  title = {Space-efficient online computation of quantile summaries},
  year = {2001},
  isbn = {1581133324},
  publisher = {Association for Computing Machinery},
  address = {New York, NY, USA},
  url = {https://doi.org/10.1145/375663.375670},
  doi = {10.1145/375663.375670},
  booktitle = {Proceedings of the 2001 ACM SIGMOD International Conference on
               Management of Data},
  pages = {58–66},
  numpages = {9},
  location = {Santa Barbara, California, USA},
  series = {SIGMOD '01},
}

@article{KLL,
  title = {Optimal Quantile Approximation in Streams},
  author = {Zohar S. Karnin and Kevin J. Lang and Edo Liberty},
  journal = {2016 IEEE 57th Annual Symposium on Foundations of Computer Science
             (FOCS)},
  year = {2016},
  pages = {71-78},
  url = {https://api.semanticscholar.org/CorpusID:15640305},
}

@article{SplineSketch,
  author = {\L{}ukasiewicz, Aleksander and T\v{e}tek, Jakub and Vesel\'{y},
            Pavel},
  title = {SplineSketch: Even More Accurate Quantiles with Error Guarantees},
  year = {2025},
  issue_date = {December 2025},
  publisher = {Association for Computing Machinery},
  address = {New York, NY, USA},
  volume = {3},
  number = {6},
  url = {https://doi.org/10.1145/3769827},
  doi = {10.1145/3769827},
  journal = {Proc. ACM Manag. Data},
  month = dec,
  articleno = {362},
  numpages = {26},
}

@inproceedings{HLL,
  TITLE = {{HyperLogLog: the analysis of a near-optimal cardinality estimation
           algorithm}},
  AUTHOR = {Flajolet, Philippe and Fusy, Eric and Gandouet, Olivier and Meunier,
            Fr{\'e}d{\'e}ric},
  URL = {https://hal.science/hal-00406166},
  BOOKTITLE = {{Analysis of Algorithms 2007 (AofA07)}},
  ADDRESS = {Juan les pins, France},
  EDITOR = {Philippe Jacquet},
  PAGES = {127--146},
  YEAR = {2007},
  MONTH = Jun,
  HAL_ID = {hal-00406166},
  HAL_VERSION = {v1},
}

@inproceedings{HLLL,
  author = {Karppa, Matti and Pagh, Rasmus},
  title = {HyperLogLogLog: Cardinality Estimation With One Log More},
  year = {2022},
  isbn = {9781450393850},
  publisher = {Association for Computing Machinery},
  address = {New York, NY, USA},
  url = {https://doi.org/10.1145/3534678.3539246},
  doi = {10.1145/3534678.3539246},
  booktitle = {Proceedings of the 28th ACM SIGKDD Conference on Knowledge
               Discovery and Data Mining},
  pages = {753–761},
  numpages = {9},
  location = {Washington DC, USA},
  series = {KDD '22},
}

@article{ULL,
  author = {Ertl, Otmar},
  title = {UltraLogLog: A Practical and More Space-Efficient Alternative to
           HyperLogLog for Approximate Distinct Counting},
  year = {2024},
  issue_date = {March 2024},
  publisher = {VLDB Endowment},
  volume = {17},
  number = {7},
  issn = {2150-8097},
  url = {https://doi.org/10.14778/3654621.3654632},
  doi = {10.14778/3654621.3654632},
  journal = {Proc. VLDB Endow.},
  month = mar,
  pages = {1655–1668},
  numpages = {14},
}

\ifappendix
\clearpage
\newpage

\appendix

\section{Applying \sketch\ to Misra-Gries}\label{sec:mg}
We now apply \sketch\ to Misra-Gries to derive \sketchmg.

\textbf{Misra-Gries (MG).}
By using dedicated counters to track the most frequent keys, MG achieves high
accuracy under skewed workloads. Moreover, by only incrementing a key's
dedicated counter when it is inserted, MG ensures that it never overestimates
the key counts. This property is ideal for finding the most frequent keys, as
described in \Cref{sec:introduction}, since it guarantees that the frequency of
those keys is at least as large as the estimates.

\begin{figure}
    \centering
    \begin{tikzpicture}
        \def\arrayw{4.75}
        \def\arrayh{0.25}
        \def\arraylen{5}
        \def\cellw{0.95}
        \def\elemw{0.65}
        \def\counterw{0.30}
        \def\annotationoffset{0.05}
        \def\bottomymul{2.5}
        \def\casebracewidth{1.5}
        \def\casebracewidthquery{0.55}

        \foreach \j in {1, ..., \arraylen} {
            %\draw[white,pattern=north east lines] (\j*\cellw-\counterw,0.5*\arrayh) rectangle (\j*\cellw-\cellw,-0.5*\arrayh);
            \draw[white,fill=gray!20] (\j*\cellw-\counterw,0.5*\arrayh) rectangle (\j*\cellw-\cellw,-0.5*\arrayh);
            \draw[dotted,thick] (\j*\cellw-\counterw,0.5*\arrayh) -- (\j*\cellw-\counterw,-0.5*\arrayh);
        }
        \foreach \j in {2, ..., \arraylen} {
            \draw (\j*\cellw-\cellw,0.5*\arrayh) -- (\j*\cellw-\cellw,-0.5*\arrayh);
        }
        \draw[black] (0,0.5*\arrayh) rectangle (\arrayw,-0.5*\arrayh);

        \node[inner sep=4pt] (key) at (-2.5*\cellw,2.5*\arrayh) {\small Insert $x$};
        \draw[decorate,decoration={brace,raise=1pt,amplitude=2pt,mirror}] ($(key.east)+(1pt,0.5*\casebracewidth)$) -- ($(key.east)+(1pt,-0.5*\casebracewidth)$);
        \node[inner sep=0pt,anchor=west,align=center] (case_1) at ($(key.east)+(2pt,0.35*\casebracewidth)$) {\footnotesize Exists \hspace*{1pt}};
        \node[inner sep=0pt,anchor=west,align=center] (case_2) at ($(key.east)+(2pt,0.03*\casebracewidth)$) {\footnotesize Does not exist, \\[-4pt] \footnotesize Empty slot};
        \node[inner sep=0pt,anchor=west,align=center] (case_3) at ($(key.east)+(2pt,-0.35*\casebracewidth)$) {\footnotesize Does not exist, \\[-4pt] \footnotesize No empty slot};

        \draw[-stealth] (case_1.east) -| (1.0*\cellw+\elemw+0.5*\counterw,0.5*\arrayh) node[pos=0.25,above=2pt,align=center,inner sep=0pt] {\footnotesize $+1$};
        \draw[-stealth] (case_2.east) -| (\arraylen*\cellw-0.5*\cellw,1.2*\arrayh) node[pos=0.1,above=1pt,align=center,inner sep=0pt] {\footnotesize Store $(x,1)$};
        \draw[decorate,decoration={brace,raise=1pt,amplitude=2pt}] (\arraylen*\cellw-\cellw,0.65*\arrayh) -- (\arraylen*\cellw,0.65*\arrayh);
        \draw (case_3.east) -| (-0.15*\cellw,-\bottomymul*\arrayh);
        \draw (-0.15*\cellw,-\bottomymul*\arrayh) -- (\arraylen*\cellw-0.5*\counterw,-\bottomymul*\arrayh) node[pos=0.1,below=2pt,align=center,inner sep=0pt] {\footnotesize $-1$};
        \foreach \j in {1, ..., \arraylen} {
            \draw[-stealth] (\j*\cellw-0.5*\counterw,-\bottomymul*\arrayh) -- (\j*\cellw-0.5*\counterw,-0.5*\arrayh);
        }

        \node[inner sep=4pt] (query) at (\arrayw+0.05*\cellw,4.25*\arrayh) {\small Query $q$};
        \draw[decorate,decoration={brace,raise=1pt,amplitude=2pt}] ($(query.west)+(-1pt,0.5*\casebracewidthquery)$) -- ($(query.west)+(-1pt,-0.5*\casebracewidthquery)$);
        \node[inner sep=0pt,anchor=east,align=center] (case_1) at ($(query.west)+(-2pt,0.3*\casebracewidthquery)$) {\hspace*{1pt} \footnotesize Exists, 4};
        \node[inner sep=0pt,anchor=east,align=center] (case_2) at ($(query.west)+(-2pt,-0.3*\casebracewidthquery)$) {\footnotesize Otherwise, 0};
        \draw[stealth-] (case_1.west) -| (2.0*\cellw+\elemw+0.5*\counterw,0.5*\arrayh);
        \node[inner sep=0pt] at (2.0*\cellw+0.5*\elemw,0.0*\arrayh) {\footnotesize $q$};
        \node[inner sep=0pt] at (2.0*\cellw+\elemw+0.5*\counterw,0.0*\arrayh) {\footnotesize $4$};

        %\node[inner sep=0pt,anchor=west] at (-3.2*\cellw,5.25*\arrayh) {\textbf{C) MG}};
    \end{tikzpicture}
    \caption{MG uses dedicated counters to track frequent keys.}
    \label{fig:mg_overview}
\end{figure}

\Cref{fig:mg_overview} presents an overview of MG. MG maintains an array of
slots, each storing a key (shown in gray) and a dedicated frequency counter
(shown in white). MG estimates the frequency of each key in the array as the
value of its counter. It returns an estimate of zero for untracked keys. When
inserting a key~$x$, MG first checks if it has a dedicated counter in some
slot. If it does, MG increments~$x$'s counter. If not, and the array has an
empty slot, MG stores~$x$ in it with a count of~$1$. If~$x$ is not in the array
and the array is full, MG decrements all counters and discards any keys whose
counters become zero. Such keys are unlikely to be among the most frequent, and
their slots are better used for tracking more frequent keys. Note that~$x$ is
not stored in the newly emptied slots and is only stored if it is inserted when
some slot is already empty. MG does not support deletions without deteriorating
the error bound~\cite{BatchUpdateMisraGries}. This is because it lacks
information about the untracked keys, which may become the heavy hitters after
the current heavy hitters are deleted, thereby failing to place them in its
array.

Having more slots in the array allows MG to track more keys without discarding
them, leading to more accurate estimates. Since MG is a deterministic data
structure that does not use hashing, unlike CMS and CS, it does not suffer from
hash collisions, and its estimations are within the error bounds with full
confidence. Using an array of~$w$ slots, MG estimates the frequency of any key
with an error of at most~$\frac{N}{w}$, i.e., $|\hat{f}(x)-f(x)| \leq
\frac{N}{w}$ for any key~$x$. Since it uses dedicated counters for the heavy
hitters, it also provides a strong ``\emph{Tail Guarantee},'' stating that the
error for any key is at most~$\min_{1 \leq i \leq w}
\frac{N^{\text{res}(i)}}{w-i+1}$, where $N^{\text{res}(i)}$ is the total
frequency of all keys not among the top~$i$. Note that~$N^{\text{res}(i)}$ is
much smaller than~$N$ when the data is skewed.

It is possible to modify MG to support unbiased estimation by probabilistically
replacing the array's least frequent key newly inserted
keys~\cite{UnbiasedSpaceSaving}. Other variations of MG achieve higher accuracy
in practice, but suffer from two-sided errors in exchange~\cite{SpaceSaving}.

Typical implementations of MG use a hash table to implement the array of keys,
yielding constant-time queries. They also optimize insertions by lazily
decrementing the counters and employing an auxiliary heap to locate and discard
keys with zero counts in~$O(\log w)$ time. Insertions can be sped up to
constant time using a specialized data structure similar to a linked list in
exchange for a higher memory footprint~\cite{SpaceSaving}.

\textbf{Accommodating Skew.}
The counts of the heavy hitters stored in MG may still exhibit skew among
themselves. This leads to memory wastage if the counter values are stored in
uniformly-sized counters, as in the case of CMS and CS. \sketchmg addresses
this issue by applying \counterencodingabbrv\ to the counters. Since
\counterencodingabbrv\ requires the counters to be tightly packed in memory
without any keys in-between them, \sketchmg stores the counter values in a
separate array that mirrors the ordering of the keys within the slots before
applying \counterencodingabbrv.

\textbf{Accommodating Unbounded Data Growth.}
As MG's estimation error is~$N/w$ (or $\min_{1 \leq i \leq w}
N^{\text{res}(i)}/(w-i+1)$ in the tail case), which mirrors that of CMS,
\sketchmg expands based on the growth of the total number of keys~$N$. Since MG
does not hash keys and can place them in any slot within the array, \sketchmg
expands on a slot-by-slot basis. This is in contrast to \sketchcms and
\sketchcs, which must double the size of their arrays to enable consistent
hashing. As such, \sketchmg smoothly increases its number of slots while
ensuring that it doubles every time~$N$ more than doubles. This approach
ensures that the number of counters in the array matches the size
function~$W(N)$ encoding the expansion policy at all times. Moreover, this
incremental approach always maintains at least as many counters as doubling the
number of counters in one go. Therefore, it allows for attaining errors that
grow at a slower rate{\textemdash}by a constant factor.

Because MG does not support deletions, \sketchmg does not need to support
contractions and can forgo storing a record of the sketch's state.

\textbf{Theoretical Analysis.} 
Following a similar argument as that for \sketchcms, we analyze the interaction
of the number of times the count of each is decremented in \sketchmg with the
size function to bound its estimation error:
\begin{restatable}{theorem}{thmmgnormal}
    \label{thm:mg-normal}
    When using a size function~$W(\cdot)$ on a stream with a total key count
    of~$N$, \sketchmg provides an estimation error~$E(N)$ satisfying
    $$
    E(N) \leq  
    \begin{cases}
        O(1) \cdot \epsilon N^{1-\alpha} & \text{if $W(N)=N^{\alpha}/\epsilon$, for $\alpha \in [0,1)$,} \\
        \ln N \cdot \epsilon & \text{if $W(N)=N/\epsilon$.}
    \end{cases}
    $$
    Compared to a fixed-size MG allocated upfront with an array of~$W(N)$
    counters with knowledge of~$N$, the error terms above are higher by at most
    a constant and a logarithmic factor.
\end{restatable}
Compared to the bounds from Theorems~\ref{thm:cms-w} and \ref{thm:cs-w}, the
logarithmic term in the linear case of~$W(N)$ is a natural logarithmic, which
is smaller by a constant factor. Note that setting~$\epsilon \leq 1$ in the
linear case trivially yields an error bound of~0, as \sketchmg will always grow
enough to store the exact counts of the keys without ever decrementing
counters. Interestingly, setting~$\epsilon$ to be even slightly larger than~1
creates the logarithmic overhead of \Cref{thm:mg-normal}.

In addition, MG's tail guarantee generalized to \sketchmg with the same
sublinearity property as that of \Cref{thm:mg-normal}:
\begin{restatable}{theorem}{thmmgres}
    \label{thm:mg-res}
    When using a size function~$W(\cdot)$ on a stream with a total key count
    of~$N$, \sketchmg provides an estimation error~$E(N)$ satisfying
    $$
    E(N) \leq  
    \begin{cases}
        O(1) \cdot \epsilon \parentheses{N^{\text{res}(k)}}^{1-\alpha} + O(k^{1/\alpha-1}) & \text{if $W(N)=N^{\alpha}/\epsilon$} \\
        \ln N^{\text{res}(k)} \cdot \epsilon + \log_2 k  & \text{if $W(N)=N/\epsilon$,}
    \end{cases}
    $$
    where~$N^{res(k)}$ is the total count of the keys not among the top~$k$.
\end{restatable}
The additive terms in the error bounds above come from the possibility of
\sketchmg not having $k$ counters from the onset. This forces it to ``wait
out'' the stream until the array grows enough to contain at least~$k$ counters,
introducing these additive error. These additive terms vanish if \sketchmg is
allocated~$k$ counters from the start.

\section{Upper bounds: \sketchcms, \sketchcs, and \sketchmg}\label{sec:ommitted_proofs_upper}
For simplicity, we will first discuss the analysis for the case in which we
have an insertion-only stream, and later discuss why our arguments carry over
when we also have deletions. In the following, we will also assume that, for
both CS and CMS, the number of arrays we employ is $d=1$. This does not change
the expected error (in the case of CMS) or the variance of the error (in the
case of CS), merely the confidence we have over the estimate.

We will do the analysis assuming a general size function~$W(\cdot)$ and then
discuss how the analysis proceeds for specific instantiantions of~$W(\cdot)$
(linear or sublinear regime). Note that we perform doubling each time the
function~$W(\cdot)$ doubles. In the case of CMS, $W(\cdot)$ is a function of~$N
= \sum f(x)$, the total number of keys in the stream. In the case of CS,
$W(\cdot)$ is a function of~$F = \sum (f(x))^2$, the
squared~$\ell_2${\nobreakdash-}norm of the stream. In either case, we can track
the moments \emph{in time} in which the value of the function doubles.
Formally, instead of thinking of~$W(\cdot)$ as a function of the total key
count~$N$ or the squared~$\ell_2${\nobreakdash-}norm of the stream, we think of
it as a function of time such that, at any point in time~$t$, $W_t$ is the
number of counters the FE sketch uses.

\medskip
\textbf{Division of Time into Epochs.}
Let~$T_0$ be the time at which the FE sketch is instantiated. This sketch
has~$W_0$ counters and is used as long as the length of the stream is at
most~$N_0$. We then define~$T_i$ for $i\geq 1$ to be the first moment in time
when the length of the array doubles, i.e., $W_i = 2\cdot W_{i-1}$. Note that,
in between times~$T_{i}$ and $T_{i+1}$, the length of the array stays the same
(i.e., it is equal to~$W_i$). We call this period of time \emph{Epoch~$i$} and
denote the FE sketch in~this~epoch~by~$\mathcal{S}_{i}$.

Formally, let~$(x,1)$ be an insertion that occured in epoch~$j$ before the
current epoch, i.e., epoch~$i$. During expansions, this update operation is
replicated on all the counters in~$\mathcal{S}_i$ that~$x$ could have hashed
into. Namely, if~$x$ hashed to an index~$\mathcal{J}_x$ in epoch~$j$ for one of
the arrays, then, during epoch~$i$, all of the $2^{j-i}$ counters
$\{0,1\}^{j-i} \circ \mathcal{J}_x$ ingest the insertion of~$x$, where $\circ$
denotes concatenation and $\{0,1\}^{j-i}$ denotes all binary strings of length
$j-i$. From the perspective of a counter in~$\mathcal{S}_i$, what this means is
that it ingests previous updates from elements that don't technically hash into
it. We describe this in more detail in what follows.

\medskip
\textbf{Copying Counters from Previous Epochs.}
Let~$\mathcal{I}$ be a counter in epoch~$i$. This counter ingests the following
operations:

\begin{itemize}
    \item Updates~$(x,1)$ within epoch $i$ for which~$h(x)=\mathcal{I}$. This
        event occurs with probability~$1/W_i$ since there are~$W_i$ counters in
        the array during that epoch. For ease of analysis, we
        define~$\mathcal{C}_{i,\mathcal{I}}(x)$ to be the indicator random
        variable that is~$1$ if this event occurs and $0$ otherwise. Then, we
        have that~$\Pr\left[ \mathcal{C}_{i,\mathcal{I}}(x) = 1 \right] = 1/W_i$.

    \item Updates~$(x,1)$ from epoch~$i-1$ that, during epoch~$i-1$, hashed
        into a counter that matches the lower-order~$i - 1$ bits of
        $\mathcal{I}$'s address in binary. This event occurs with
        probability~$2/W_i$. The factor~$2$ accounts for the counter
        in~$\mathcal{S}_{i-1}$ that was copied during the previous expansion.

    \item In general, updates~$(x,1)$ from epoch $j$ where $j < i$ and for
        which~$h(x)$ matches $\mathcal{I}$ in the lower-order~$j$ bits of its
        binary representation. This event occurs with
        probability~$2^{i-j}/W_i$. Similarly as before, we
        define~$\mathcal{C}_{j,\mathcal{I}}(x)$ to be the indicator random
        variable that is~$1$ if this event occurs and $0$ otherwise. Then,
        $\Pr\left[ \mathcal{C}_{j,\mathcal{I}}(x) = 1 \right] =  2^{i-j}/W_i$.
\end{itemize}

While updates from previous epochs have a higher chance of hashing
into~$\mathcal{I}$ than updates in epoch~$i$, we can nevertheless bound the
influence they have on the total error of the estimate returned by CMS, CS, and
MG.

\medskip
\textbf{Streams with deletions.}
In order to account for deletions, we relate the state of \sketch\ in a stream
with deletions to the state of a different \sketch\ with higher errors on an
insertion-only stream. To begin with, notice that, due to the way we perform
contractions, we can pretend that all previous epochs have only insertions in
them (the state of counters in \sketch\ is smaller than if we only had
insertions previously). We now argue about deletions~$(x,-1)$ that only occur
during the current epoch~$i$ and that have not triggered a contraction. If,
during the same epoch~$i$, there is also an insertion~$(x,1)$, then we can act
as if this insertion-deletion pair did not occur in the stream, i.e., the value
of the counter incurs no change during that epoch~$i$. Any extra deletions
during epoch~$i$ that are not cancelled out by insertions during the same epoch
can only decrement the value of the counter, and they do not negatively affect
the upper bounds we compute. In other words, the value of the counter in
epoch~$i$ is smaller than its value in a stream in which we remove these extra
deletions. Most crucially, by considering this insertion-only stream, we do not
miss any errors the counter accumulates from previous epochs. Namely, during
epoch~$i$, counter~$\mathcal{I}$ might still incur errors from previous
insertions of~$x$ from epochs $j<i$ that would not have hashed to it during
epoch~$i$ (i.e., the event~$\mathcal{C}_{j,\mathcal{I}}(x)$ still holds and we
cannot undo it since we do not know when the last insertion of~$x$ occured).

\subsection{Upper Bounds for \sketchcms}\label{sec:upper_bounds_cms}
In the case of \sketchcms, each counter gathers the sum of frequencies of keys
that hash into it or counters that copied into it. Recall that the total key
count~$N$ in a stream is the total sum of frequencies of that stream. We show
the following:
\begin{lemma}\label{lemma:cms-w}
    Letting~$E(N)$ denote the estimation error of \sketchcms when using a size
    function~$W(\cdot)$ and one array on a stream with a total of~$N$ keys, we
    have at the end of epoch~$i$ that
	$$ \expectation{E(N_{i+1})} \leq  \frac{1}{W_i} \cdot \sum_{j=0}^i  2^{i-j} \cdot (N_{j+1} - N_j), $$
    where~$N_{j}$ for $j\geq 0$ denotes the total number of keys in the stream
    at the start of epoch~$j$. 
\end{lemma}

\begin{proof}
    During epoch~$i$, we have at most~$N_{i+1}-N_{i}$ fresh insertions (i.e.,
    that have not been cancelled out by deletions). Let~$f_i(x)$ denote the
    increase, if any, in the frequency of~$x$ during epoch $i$. We observe that
    $\sum_x f_j(x) = N_{j+1}-N_j$ for all $j\geq 0$. This implies that, in our previous
    notation, $E(N_{i+1}) \leq \mathcal{I}$, which can be expressed as
	$$ E(N_{i+1}) \leq \mathcal{I} =\sum_{j=0}^i   \sum_x \mathcal{C}_{j,\mathcal{I}}(x) \cdot f_j(x). $$
	\noindent
    In other words, the counter~$\mathcal{I}$ gets the frequency~$f_j(x)$ added
    to it if and only if the event~$\mathcal{C}_{j,\mathcal{I}}(x)$ occurs. We
    also have that~$\expectation{\mathcal{C}_{j,\mathcal{I}}(x)} =
    2^{i-j}/W_i$. By linearity of expectation, we then get
	\begin{align*}
		\expectation{\mathcal{I}} &= \sum_{j=0}^i   \sum_x \frac{2^{i-j}}{W_i} \cdot f_j(x) = \frac{1}{W_i} \cdot  \sum_{j=0}^i   \sum_x 2^{i-j} \cdot f_j(x)\\
		&=	\frac{1}{W_i} \cdot  \sum_{j=0}^i   2^{i-j}  \sum_x f_j(x) = \frac{1}{W_i} \cdot  \sum_{j=0}^i  2^{i-j} \cdot  (N_{j+1} - N_j).
	\end{align*}
	The claim follows.
\end{proof}

In other words, the increase in error we get due to copying the counters is
dominated by the term~$\sum_{j=0}^i  2^{i-j} \cdot (N_{j+1} - N_j)$. This is
where the specific size function~$W(\cdot)$ we employ has an effect, as it
translates into bounds on~$N_j$ and their relative sizes. We now consider
various special cases and discuss their respective behavior:
\thmcms*
\begin{proof} ‌\\
    \textbf{The Linear Case of~$W(N)=N/\epsilon$.}
    In this case, the size function doubles whenever the number of elements
    doubles, i.e., $N_{j+1} = 2 \cdot N_j$ for all $j\geq 0$ and, in general,
    $N_j = 2^j \cdot N_0$ for all $j\geq 0$. The additional error term then
    becomes:
	\begin{align*}
		\sum_{j=0}^i  2^{i-j} \cdot  (N_{j+1} - N_j) & =	\sum_{j=0}^i  2^{i-j} \cdot  (2^{j+1}\cdot N_0 - 2^{j}\cdot N_0)\\
		&=  \sum_{j=0}^i  2^{i-j} \cdot  2^j \cdot N_0 = (i+1) \cdot 2^i \cdot N_0\\
		&= (i+1) \cdot N_i\;.
	\end{align*}	
    Plugging in this bound in the formula from \Cref{lemma:cms-w} (i.e.,
    dividing by~$W(N_i)=N_i/\epsilon$), we get the main claim for
    $W(N)=N/\epsilon$.

	\medskip\noindent
    \textbf{The Sublinear Case of~$W(N)=N^{\alpha}/\epsilon$.}
    In this case, the size function doubles more slowly, i.e.,
    whenever~$N_{i+1} = 2^{1/\alpha} \cdot N_i$. Since~$0 \leq \alpha < 1$,
    $N_{i+1}$ is bigger than~$2\cdot N_i$. This difference is why we do not
    have the extra dependency on~$i$ in the error term. Formally, we have
    that~$N_i = 2^{i/\alpha} \cdot N_0$ and thus

	\begin{align*}
        \sum_{j=0}^i  2^{i-j} \cdot  (N_{j+1} - & N_j) = \sum_{j=0}^i  2^{i-j} \cdot  (2^{(j+1)/\alpha}\cdot N_0 - 2^{j/\alpha}\cdot N_0)\\
		&=  \sum_{j=0}^i  2^{i-j} \cdot  2^{j/\alpha} \cdot N_0 \cdot (2^{1/\alpha} -1) \\
		&= 2^i \cdot N_0 \cdot (2^{1/\alpha} -1) \cdot   \sum_{j=0}^i   2^{j/\alpha- j} \\
        &= 2^{(1-1/\alpha) \cdot (i+1) - 1} \cdot N_{i+1} \cdot (2^{1/\alpha} -1) \cdot \sum_{j=0}^i   2^{j/\alpha- j},
        %&= ,
	\end{align*}	
    where in the last step we have used~$N_0 = 2^{(-1/\alpha)\cdot(i+1)} \cdot
    N_{i+1}$. Simplifying the sum of the geometric series on the right
    as~$\frac{2^{(1/\alpha-1)\cdot(i+1)}-1}{2^{1/\alpha-1}-1}$ then yields
	\begin{align*}
        \sum_{j=0}^i  2^{i-j} & \cdot (N_{j+1} - N_j) \\
        &= 2^{(1-1/\alpha) \cdot (i+1) - 1} \cdot N_{i+1} \cdot (2^{1/\alpha} -1) \cdot \frac{2^{(1/\alpha-1) \cdot (i+1)} - 1}{2^{1/\alpha-1} - 1} \\
        &= N_{i+1} \cdot (2^{1/\alpha} -1) \cdot \frac{2^{-1} - 2^{(1-1/\alpha) \cdot (i+1) - 1}}{2^{1/\alpha-1} - 1} \\
        &= N_{i+1} \cdot (2^{1/\alpha} -1) \cdot \frac{1 - 2^{(1-1/\alpha) \cdot (i+1)}}{2^{1/\alpha} - 2} \\
        &= N_{i+1} \cdot (2^{1/\alpha} -1) \cdot \frac{1 - 2^{(\alpha-1) \cdot \log_2 (N_{i+1}/N_0)}}{2^{1/\alpha} - 2} \\
        &\leq N_{i+1} \cdot \frac{2^{1/\alpha} -1}{2^{1/\alpha} - 2}.
	\end{align*}	
    In the penultimate line, we have used the property that~$i+1$ is equal to
    $\log_{2^{1/\alpha}} (N_{i+1}/N_0) = \log_{2^{1/\alpha}} 2 \cdot \log_2
    (N_{i+1}/N_0) = \alpha \cdot \log_2 (N_{i+1}/N_0)$. Moreover, in the last
    line, we have removed the negative exponential to derive the final upper
    bound. The claim follows by plugging this bound in the expected error from
    \Cref{lemma:cms-w}.
\end{proof}

\subsection{Upper Bounds for \sketchcs}\label{sec:upper_bounds_cs}
For \sketchcs, we obtain a similar bound to that of~\Cref{lemma:cms-w}, except
as a function of the squared~$\ell_2${\nobreakdash-}norm of the stream:
\begin{lemma}
    \label{lemma:cs-w} 
    Like CS, \sketchcms provides unbiased estimates. Moreover, letting~$V(F)$
    denote its estimation variance when using a size function~$W(\cdot)$ and
    one array on a stream with squared~$\ell_2${\nobreakdash-}norm $F = \sum_y
    (f(y))^2$, we have at the end of epoch~$i$ that
    $$ V(F_{i+1}) \leq \frac{1}{W_i} \cdot \sum_{j=0}^i 2^{i-j} \cdot (F_{j+1} - F_j), $$
    where~$F_j$ denotes the squared~$\ell_2${\nobreakdash-}norm of the stream
    at the start of epoch~$j$. 
\end{lemma}
\begin{proof}
    We are interested in describing the behavior of the random
    variable~$\hat{f}(z) = \mathcal{I} \cdot \sigma(z)$ for any query~$z$,
    where~$\sigma(z)$ is $z$'s update direction and~$\mathcal{I}$ is the
    counter that~$z$ hashes into. Using the notation we introduced in
    \Cref{sec:upper_bounds_cms}, this counter~$\mathcal{I}$ can be expressed
    as:
    $$ \mathcal{I} = \sum_{j=0}^i \sum_x \mathcal{C}_{j,\mathcal{I}}(x) \cdot f_j(x) \cdot \sigma(x) = \sum_x \mathcal{C}(x), $$
    where~$\mathcal{C}(x) = \sum_{j=0}^i \mathcal{C}_{j,\mathcal{I}}(x) \cdot
    f_j(x) \cdot \sigma(x)$ denotes the contribution of~$x$'s updates to
    counter~$\mathcal{I}$ throughout the epochs.

    We apply a similar analysis to that of a traditional CS to compute the
    expected value of the estimate~$\hat{f}(z)$, showing that~$\mathcal{I}
    \cdot \sigma(z)$ is an unbiased estimator of~$f(z)$, i.e.,
    $\expectation{\mathcal{I} \cdot \sigma(z)} = f(z)$. In the following, we
    denote by~$f_j(x)$ the frequency of key~$x$ within epoch~$j$. By linearity
    of expectation, it holds that
    $$ \expectation{\hat{f}(z)} = \sum_{x \neq z} \expectation{\mathcal{C}(x) \cdot \sigma(z) } + \expectation{\mathcal{C}(z)\sigma(z)} = \sum_{j=0}^i f_j(z) = f(z),$$
    where the last step uses the property that~$\expectation{\sigma(x) \cdot
    \sigma(z) } = 0$ for any $x \neq z$, in conjunction
    with~$\mathcal{C}_{j,\mathcal{I}}(z)=1$ since $h(z)=\mathcal{I}$. This
    proves the first part of the lemma.

    We now turn to bound the variance of this estimator. In particular, we are
    interested in how much bigger it becomes compared to~$F_i/W_i$. Writing the
    definition of the variance, we get that
    \begin{align*}
        \variance{\hat{f}(z)} &= \expectation{\parentheses{\sum_x \mathcal{C}(x) \sigma(z)}^2} - (f(z))^2 \\
        & = \sum_x \expectation{\parentheses{\mathcal{C}(x)}^2} + \sum_x \sum_{y \neq x}\expectation{ \mathcal{C}(x) \mathcal{C}(y)} - (f(z))^2 \\
        & = \sum_x \expectation{\parentheses{\mathcal{C}(x)}^2} - (f(z))^2.
    \end{align*}
    In the third line we have used the property that~$\expectation{\sigma(x)
    \cdot \sigma(y)}=0$ for any $x\neq y$ to
    simplify~$\expectation{\mathcal{C}(x) \mathcal{C}(y)}$ to 0. To
    bound~$\expectation{\parentheses{\mathcal{C}(x)}^2}$, recall that we
    defined~$\mathcal{C}(x) = \sum_{j=0}^i \mathcal{C}_{j,\mathcal{I}}(x) \cdot
    f_j(x) \cdot \sigma(x)$. Then:

    \begin{align*}
        \expectation{\parentheses{\mathcal{C}(x)}^2} &= \expectation{\parentheses{ \sum_{j=0}^i \mathcal{C}_{j,\mathcal{I}}(x) \cdot f_j(x) }^2}\\
    \end{align*}
    \begin{align}
        & = \sum_{j=0}^i \expectation{ \parentheses{ \mathcal{C}_{j,\mathcal{I}}(x) \cdot f_j(x) }^2} \notag \\
        & \;\; + 2\cdot \sum_{j=0}^{i-1}  \sum_{k=j+1}^i   \expectation{ \mathcal{C}_{j,\mathcal{I}}(x) \cdot f_j(x) \cdot  \mathcal{C}_{k,\mathcal{I}}(x) \cdot f_k(x) } \notag \\
        & = \sum_{j=0}^i (f_j(x))^2 \cdot \expectation{ \mathcal{C}_{j,\mathcal{I}}(x)} \notag \\
        & \;\; +  2\cdot \sum_{j=0}^{i-1}  \sum_{k=j+1}^i   f_j(x)  \cdot f_k(x) \cdot   \expectation{ \mathcal{C}_{j,\mathcal{I}}(x) \cdot   \mathcal{C}_{k,\mathcal{I}}(x) }. \label{eq:sketchcs_var}
    \end{align}
    Here, we have used the property~$(\mathcal{C}_{j,\mathcal{I}}(x))^2 =
    \mathcal{C}_{j,\mathcal{I}}(x)$ (which follows
    from~$\mathcal{C}_{j,\mathcal{I}}(x)$ being an indicator random variable)
    to simplify the second line. 

    To bound~$\expectation{ \mathcal{C}_{j,\mathcal{I}}(x)}$, we use the
    definition of the epochs to conclude that~$\expectation{ 
    \mathcal{C}_{j,\mathcal{I}}(x)}  \leq 2^{i-j}/W_i$. For $\expectation{ 
    \mathcal{C}_{j,\mathcal{I}}(x) \cdot \mathcal{C}_{k,\mathcal{I}}(x) }$,
    we note that the product of the two indicator random variables is $1$ only
    when both of them are $1$. Yet, the two random variables are not
    independent: if $\mathcal{C}_{k,\mathcal{I}}(x) =1$, then $
    \mathcal{C}_{j,\mathcal{I}}(x)=1$ for any $k\geq j$. Namely, if during
    epoch $k$, $x$ hashed into a counter whose address was a prefix
    of~$\mathcal{I}$'s address, then the same would have happened during
    epoch~$j \leq k$. This implies that~$\mathcal{C}_{j,\mathcal{I}}(x) \cdot
    \mathcal{C}_{k,\mathcal{I}}(x) = \mathcal{C}_{k,\mathcal{I}}(x)$, thereby
    simplifying the sum to
    \begin{align*}
        \sum_{j=0}^{i-1}  \sum_{k=j+1}^i   f_j(x)  \cdot f_k(x) \cdot & \expectation{ \mathcal{C}_{j,\mathcal{I}}(x) \cdot \mathcal{C}_{k,\mathcal{I}}(x) }  \\ 
        & = \sum_{k=1}^i \sum_{j=0}^{k-1}  f_j(x)  \cdot f_k(x) \cdot   \expectation{ \mathcal{C}_{k,\mathcal{I}}(x) }\\
        & = \sum_{k=1}^i   f_k(x) \cdot   \expectation{ \mathcal{C}_{k,\mathcal{I}}(x) } \cdot \sum_{j=0}^{k-1}  f_j(x)\\
        & = \sum_{k=1}^i   f_k(x) \cdot   \expectation{ \mathcal{C}_{k,\mathcal{I}}(x) } \cdot f_{\leq k-1}(x).
    \end{align*}
    Here, for simplicity, we have denoted~$\sum_{j=0}^{i} f_i(x)$, i.e., the
    sum of frequencies accumulated by~$x$ through epochs~0 to $i$, as $f_{\leq
    {i}}(x)$. Note that we set~$f_{\leq -1}(x)=0$. Plugging the above into
    \Cref{eq:sketchcs_var} yields
    \begin{align*}
        \expectation{\parentheses{\mathcal{C}(x)}^2} &= \sum_{j=0}^i \parentheses{(f_j(x))^2  + 2 f_j(x) \cdot f_{\leq j-1}(x)} \cdot \expectation{ \mathcal{C}_{j,\mathcal{I}}(x)}  \\
        & \mkern-18mu = \sum_{j=0}^i \parentheses{(f_{\leq j}(x) + f_{\leq j-1}(x))^2 - (f_{\leq j-1}(x))^2} \cdot \expectation{ \mathcal{C}_{j,\mathcal{I}}(x)}. \\
        & \mkern-18mu = \sum_{j=0}^i \parentheses{(f_{\leq j}(x))^2  - (f_{\leq j-1}(x))^2} \cdot \expectation{ \mathcal{C}_{j,\mathcal{I}}(x)}.
    \end{align*}
    While the term~${f_{\leq j}(x) }^2  - {f_{\leq j-1}(x) }^2$ might seem
    curious, we note that~$\sum_x{f_{\leq j}(x) }^2  - {f_{\leq j-1}(x) }^2 =
    F_{j+1} - F_j$ denotes the increase in the
    squared~$\ell_2${\nobreakdash-}norm of the stream vector as we move from
    epoch~$j$ to epoch $j+1$. 

    Putting everything together, we get that
    \begin{align*}
        \variance{\hat{f}(z)} &= \sum_x \expectation{\parentheses{\mathcal{C}(x)}^2} - (f(z))^2
    \end{align*}
    \begin{align*}
        & = \sum_x \sum_{j=0}^i \parentheses{(f_{\leq j}(x))^2  - (f_{\leq j-1}(x))^2} \cdot \expectation{ \mathcal{C}_{j,\mathcal{I}}(x)} - (f(z))^2 \\
        & \leq \sum_{j=0}^{i} \parentheses{F_{j+1} - F_j} \cdot \frac{2^{i-j}}{W_i},
    \end{align*}
    proving the claim.
\end{proof}

At this point, the reasoning behind defining the size function to based on the
squared~$\ell_2${\nobreakdash-}norm of the stream becomes more clear. We now
apply \Cref{lemma:cs-w} to derive bounds on the variance depending on the size
function:
\thmcs*
We note that the calculations for proving \Cref{thm:cs-w} are very similar to
those carried out for \sketchcms in \Cref{thm:cms-w}, with the difference that
we now see a dependency on the squared~$\ell_2${\nobreakdash-}norm~$F$ rather
than the total key count~$N$. We include these derivations for~completeness:
\begin{proof}
    We start with the original FE sketch~$\mathcal{S}_0$, which is used to
    process the stream until its squared~$\ell_2${\nobreakdash-}norm reaches at
    most some~$F_0$.

    \textbf{The Linear Case of~$W(F)=F/\epsilon$.}
    In this case, the size function doubles whenever the
    squared~$\ell_2${\nobreakdash-}norm doubles, i.e., $F_{j+1} = 2
    \cdot F_j$ for all $j\geq 0$ and, in general, $F_j = 2^j \cdot F_0$ for all
    $j\geq 0$. We therefore get that
    \begin{align*}
        V(F) &\leq  \sum_{j=0}^{i}  \parentheses{F_{j+1} - F_j} \cdot \frac{2^{i-j}}{W_i} \\
        & \leq \frac{1}{W_i} \cdot \parentheses{ F_0 \cdot 2^{i} + \sum_{j=1}^{i}  \parentheses{2^{j+1} - 2^j} \cdot F_0 \cdot 2^{i-j}} \\
        & \leq \frac{1}{W_i} \cdot  \parentheses{ F_0 \cdot 2^{i} +  \sum_{j=1}^{i}  2^j \cdot F_0 \cdot 2^{i-j}} \\
        & \leq \frac{1}{W_i} \cdot  \parentheses{ F_0 \cdot 2^{i} +  \sum_{j=1}^{i}  F_0 \cdot 2^i} \\
        & \leq (i+1) \cdot \frac{F_i}{W_i}.
    \end{align*}

\medskip\noindent
\textbf{The Sublinear Case of~$W(F)=F^{\alpha}/\epsilon$.}
    In this case, the size function doubles more slowly, i.e.,
    whenever~$F_{i+1} = 2^{1/\alpha} \cdot F_i$. Since~$0 \leq \alpha < 1$,
    $F_{i+1}$ is bigger than~$2F_i$. Plugging this into the statement
    of~\Cref{lemma:cs-w}, we have that~$F_i = 2^{i/(1-\alpha)} \cdot F_0$ and
    thus
    \begin{align*}
        \frac{1}{W_i} \cdot \sum_{j=0}^i  2^{i-j} \cdot & (F_{j+1} - F_j) \\
        & = \frac{1}{W_i} \cdot \sum_{j=0}^i  2^{i-j} \cdot \parentheses{2^{(j+1)/(1-\alpha)}\cdot F_0 - 2^{j/(1-\alpha)} \cdot F_0} \\
        & = \frac{1}{W_i} \cdot \sum_{j=0}^i  2^{i-j} \cdot 2^{j/(1-\alpha)} \cdot F_0 \cdot \parentheses{2^{1/(1-\alpha)}-1} \\
        & = \frac{1}{W_i} \cdot 2^i \cdot F_0 \cdot \parentheses{2^{1/(1-\alpha)}-1} \cdot \sum_{j=0}^i 2^{j/(1-\alpha) - j} \\
        & = O(F_{i+1}/W_i).
    \end{align*}	
\end{proof}

\subsection{Upper Bounds for \sketchmg}\label{sec:upper_bounds_mg_normal}
In the case of \sketchmg, we consider a slightly more general division of time
into epochs. We start with a sketch that uses~$W(N_0)$ counters until the
stream length reaches~$N_0$. At this point in time, the sketch sets the number
of counters to~$W(N_1)$. We leave the relationship between~$W(N_0)$ and
$W(N_1)$ undefined for now. The sketch uses~$W(N_1)$ counters until the stream
length reaches~$N_1$, and we call this period of time epoch~1. Once~$N_2$ is
reached, the sketch switches to~$W(N_2)$ counters (starting epoch~2) and so on.
We then prove the following lemma:

\begin{lemma}
    \label{lemma:mg-w}
    Letting~$E(N)$ denote the estimation error of \sketchmg when using a size
    function~$W(\cdot)$ on a stream with a total of~$N$ keys, we have at the
    end of epoch $i$ that
    $$ E(N_i) \leq \sum_{j=0}^i \frac{N_j - N_{j-1}}{W(N_j)+1},$$
    where~$N_{j}$ for $j\geq 0$ denotes the total number of keys in the stream
    at the end of epoch~$j$.
\end{lemma}
\begin{proof}
    We begin by first defining some notation. Let~$C_i$ denote the sum of the
    counter stored in the sketch at the end of epoch~$i$ (with $C_{-1}=0$).
    Note that $C_i \geq 0$, since we only decrement counters when they are all
    nonzero. We also denote by~$d_i$ the total number of keys that decrement
    the sketch's counts during epoch~$i$. By design, the following relationship
    holds for each~$i\geq 1$:
    $$ C_i = C_{i-1} + N_i - N_{i-1} - d_i \cdot (W(N_i)+1). $$

    To see why this is true, consider the case in which~$i=1$. At the beginning
    of the first epoch, the counts stored in the sketch add up to~$C_0$. During
    the epoch, we see~$N_1-N_0$ fresh insertions, each of which either
    increases the total count of the sketch by~1, or decrements~$W(N_1)$
    counts. The number of keys among these~$N_1-N_0$ that make such decrements
    is exactly~$d_1$. In other words, there are~$N_1-N_0-d_1$ keys that
    increase the sketch counts by~1, and $d_1$ keys that decrease the sketch
    count by~$W(N_1)$. The total changes are therefore~$N_1-N_0-d_1 - d_1 \cdot
    W(N_1)$, which is how we get the first expression. It follows from the
    general expression that 
    $$ d_i = \frac{N_i - N_{i-1} + C_i - C_{i-1}}{W(N_i)+1}. $$

    Now, for any specific key~$x$, the maximum error that we can make in its
    count is if we decrement its count every time we a key decrements the
    sketch's counts (i.e., this worst-case is achieved when the sketch stores
    all occurences of~$x$ from the beginning and then keeps decrementing one
    from its count every time it seems an unstored key). Formally, if we are at
    the end of epoch~$i$, then the error is bounded as
    $$ E(N_i) \leq f(x) - \hat{f}(x) \leq \sum_{j=0}^i d_j.$$
    \noindent
    Plugging in our expression for~$d_j$ and using $W_j=W(N_j)$ for ease of
    notation yields 
    \begin{align*}
        \sum_{j=0}^i d_j & = \sum_{j=0}^i  \frac{N_j - N_{j-1} + C_{j-1}-C_j}{W_j+1} \\
        &= \sum_{j=0}^i \frac{N_j - N_{j-1}}{W_j+1} + \sum_{j=0}^i \frac{C_{j-1}-C_j}{W_j+1} \\
        &=\sum_{j=0}^i \frac{N_j - N_{j-1}}{W_j+1} + \sum_{j=0}^{i-1} C_j \cdot \parentheses{-\frac{1}{W_j+1}+\frac{1}{W_{j+1}+1}} - \frac{C_i}{W_j+1}\\
        &=\sum_{j=0}^i \frac{N_j - N_{j-1}}{W_j+1} - \sum_{j=0}^{i-1} C_j \cdot \frac{W_{j+1} - W_j}{(W_j+1)(W_{j+1}+1)} - \frac{C_i}{W_i+1}\\
        &\leq \sum_{j=0}^i  \frac{N_j - N_{j-1}}{W_j+1}  \;,
    \end{align*}
    where the last inequality holds because~$W_{j+1}>W_j$, so we are ignoring
    negative factors in the sum.
\end{proof}

In addition to~\Cref{lemma:mg-w}, we require the following technical claim to
prove \Cref{thm:mg-normal}:
\begin{claim}
    \label{claim:technical}
    For any constant~$\alpha<1$, it holds that 
    $$ \sum_{j=0}^i \frac{2^{j/\alpha} - 2^{(j-1)/\alpha}}{2^{j}} = O(1) \cdot \frac{2^{i/\alpha}}{2^i}. $$
\end{claim}
\begin{proof} 
    We note that
    \begin{align*}
        \frac{2^{j/\alpha} - 2^{(j-1)/\alpha} } {2^j} &= (1-2^{-1/\alpha}) \cdot \frac{2^{j/\alpha} } {2^j} \\
        & = (1-2^{-1/\alpha}) \cdot  \frac{2^{i/\alpha}}{2^i } \cdot \frac{2^{(j-i)/\alpha}}{2^{j-i}} \\
        & =  (1-2^{-1/\alpha}) \cdot \frac{1}{2^{(i-j)\cdot (1/\alpha-1)}} \cdot \frac{2^{i/\alpha}}{2^i }.
    \end{align*}
    Summing these terms yields
    \begin{align*}
        \sum_{j=0}^i \frac{2^{j/\alpha} - 2^{(j-1)/\alpha} } {2^j} & \leq (1-2^{-1/\alpha}) \cdot \frac{2^{i/\alpha}}{2^i} \cdot \sum_{j=0}^i \frac{1}{2^{(i-j)\cdot (1/\alpha-1)}} \\
        &= O(1) \cdot \frac{2^{i/\alpha}}{2^i},
    \end{align*}
    since~$1/\alpha \geq 1$ and the sum~$\sum_{j=0}^i \frac{1}{2^{(i-j) \cdot
    (1/\alpha-1)}}$ converges to a constant.
\end{proof}

\noindent
We are now ready to prove \Cref{thm:mg-normal}:
\thmmgnormal*
\begin{proof} ‌\\
    \textbf{The Linear Case of~$W(N)=N/\epsilon$.} 
    In this case, we use the fact that in \sketchmg we increase the counters
    one by one as opposed to doubling them. In other words, we consider
    divisions into epochs for which~$W(N_{i+1})=W(N_i)+1$ and hence
    $N_{i+1}=N_i+\epsilon$. Plugging this into~\Cref{lemma:mg-w}, we get
    at the end of epoch~$i$ that 
	\begin{align*}
		E(N_i)&\leq \sum_{j=0}^i  \frac{N_j - N_{j-1}}{W(N_j)+1} = \frac{N_0}{W(N_0)} + \sum_{j=1}^i  \frac{\epsilon}{W(N_0)+j} \\
		&= \epsilon + \epsilon \cdot \sum_{j=1}^i  \frac{1}{W(N_0)+j} \leq \epsilon + \epsilon \cdot \int_{W(N_0)+1}^{W(N_0)+i} \frac{1}{x} \,dx\\
		& \leq \epsilon + \epsilon \cdot \ln\parentheses{\frac{W(N_0)+i}{W(N_0)+1}} \leq \epsilon \cdot (1+\ln i) \leq \epsilon \cdot \ln N.
	\end{align*}
    Here, the last line follows from~$i$ being the number of times we expand by
    one counter until we reach a stream length of~$N$, implying that~$i \leq
    N/\epsilon+1$. 

    \medskip\noindent
    \textbf{The Sublinear Case of~$W(N)=N^{\alpha}/\epsilon$.}
    In this case, it is easier to think of divisions into time in which we
    double the number of counters instead of increasing them one by one. That
    is, we define the epochs such that~$W(N_{i+1}) = 2\cdot W(N_i)$ and hence
    $N_{i+1} = 2^{1/\alpha} \cdot N_i$. Noting that~\Cref{lemma:mg-w} holds
    irrespective of the specific definition of the epochs then yields
    \begin{align*}
        \sum_{j=0}^i \frac{N_j - N_{j-1}}{W(N_j)+1} & = \sum_{j=0}^i \frac{2^{j/\alpha} \cdot N_0 - 2^{(j-1)/\alpha} \cdot N_0} {2^j \cdot W(N_0)} \\
        & = \sum_{j=0}^i  \frac{2^{j/\alpha}- 2^{(j-1)/\alpha}} {2^j} \cdot \frac{N_0}{W(N_0)}.
    \end{align*}
    Finally, by applying Claim~\ref{claim:technical}, we get
    \begin{align*}
        \sum_{j=0}^i \frac{N_j - N_{j-1}}{W(N_j)+1} & \leq  O(1) \cdot \frac{2^{i/\alpha}}{2^i} \cdot \frac{N_0}{W(N_0)} = O(1) \cdot \frac{N_i}{W(N_i)}.
    \end{align*}
    proving the desired result.
\end{proof}

\subsection{Tail Bounds for \sketchmg}\label{sec:upper_bounds_mg_res}
We prove the tail error guarantees for \sketchmg following the same general
argument structure of \Cref{sec:upper_bounds_mg_normal}.

\begin{lemma}\label{lemma:mg-res-w}
    Letting~$E(N)$ denote the error of \sketchmg when using a size
    function~$W(\cdot)$ on a stream of total of~$N$ keys, we have at the end of
    epoch~$i$ that
    $$ E(N_i) \leq \sum_{j=0}^i \frac{N^{\text{res}(k)}_j-N^{\text{res}(k)}_{j-1}}{W(N_j)+1-k}, $$
    where~$N^{\text{res}(k)}_j$ for $j\geq 0$ is the total count of the keys
    not among the top~$k$ at the end of epoch~$j$.
\end{lemma}
\begin{proof}
    We start similarly as in the proof of~\Cref{lemma:mg-w} and get that, if
    $d_i$ is the number of deletions to the sketch that occur during some epoch
    $i \geq 0$, then it holds that
    $$ C_i = C_{i-1} + N_i - N_{i-1} - d_i \cdot (W(N_i)+1). $$
    Here, $C_i$ denotes the sum of the sketch's at the end of epoch~$i$ (with
    $C_{-1}=0$). We generalize the above to
    $$ C_i - C_{i-1} \geq  N^{(k)}_i - N^{(k)}_{i-1} - d_i \cdot k,$$
    where~$N^{(k)}_i $ denotes total count of the top~$k$ keys at the end of
    epoch~$i$ (the top $k$ with respect to the whole stream). In other words,
    the change in the sketch counts is at least the change in the sketch counts
    restricted to the top~$k$ keys. This is because the total count restriced
    to these~$k$ keys increases by~$N^{(k)}_i - N^{(k)}_{i-1}$ and each of
    these~$k$ keys gets its counter decremented at most~$d_i$ times.
    Combining the equation and the inequality above and rearranging, we get
    that
    $$ d_i \leq \frac{N_i - N_{i-1} -N^{(k)}_i + N^{(k)}_{i-1}}{W(N_i)+1-k} =  \frac{N^{\text{res}(k)}_i- N^{\text{res}(k)}_{i-1}}{W(N_i)+1-k}. $$
    Summing these terms proves the desired result.
\end{proof}

We now proceed to prove the full tail bound of \Cref{thm:mg-res} following the
same proof structure as that of \Cref{thm:mg-normal}:
\thmmgres*
\begin{proof}
    Note that the additive terms in the bounds of \Cref{thm:mg-res} come from
    the application of \Cref{thm:mg-normal} to the portion of the stream that
    must be ``waited out'' by \sketchmg to create the~$k${\nobreakdash-}th
    counter for the first time. We therefore focus on deriving the remaining
    part of the error bound, which corresponds to the error resulting from
    processing the stream with an instance of \sketchmg initialized with at
    least~$k$ counters.

    \textbf{The Linear Case of~$W(N)=N/\epsilon$.} 
    In this case, $W(N_{i+1})=W(N_i)+1$ and $N_{i+1}=N_i +
    \epsilon$. We consider an epoch~$i$ and define a series of
    integers~$i_\ell$ as follows. We start by defining~$i_1$ as the largest
    integer such that
    $$ \sum_{j=1}^{i_1} \parentheses{N^{\text{res}(k)}_j- N^{\text{res}(k)}_{j-1}} \leq \epsilon. $$
    In words, we group consecutive terms of the form~$N^{\text{res}(k)}_j-
    N^{\text{res}(k)}_{j-1}$ until we are about to reach~$\epsilon$. Note that
    $$ \sum_{j=1}^{i_1+1} \parentheses{N^{\text{res}(k)}_j- N^{\text{res}(k)}_{j-1}} \geq \epsilon, $$
    as the maximality of~$i_1$ implies that the extra
    term~$N^{\text{res}(k)}_{i_1+1}- N^{\text{res}(k)}_{i_1}$ causes the sum to
    grow beyond~$\epsilon$. We also define~$i_2 > i_1$ in a similar fashion as
    the largest index such that 
    $$ \sum_{j=i_1+1}^{i_2} \parentheses{N^{\text{res}(k)}_j - N^{\text{res}(k)}_{j-1}} \leq \epsilon. $$
    Like before, we stop right before we gather a sum of~$\epsilon$.
    Moreover, as before, including the extra
    term~$N_{i_2+1}^{\text{res}(k)}-N_{i_2}^{\text{res}(k)}$ yields
    $$ \sum_{j=i_{1}+1}^{i_2+1} \parentheses{N^{\text{res}(k)}_j- N^{\text{res}(k)}_{j-1}} \geq \epsilon. $$
    We continue as such for~$\ell$ steps until we have that:
    $$ \sum_{j=i_{\ell+1}}^{i} \parentheses{N^{\text{res}(k)}_j- N^{\text{res}(k)}_{j-1}} \leq \epsilon. $$

    We now define~$i_0=0$ and upper bound~$\ell$ as
    \begin{align*}
        N^{\text{res}(k)}_i  & \geq \sum_{j=1}^{i_{\ell}} \parentheses{N^{\text{res}(k)}_j- N^{\text{res}(k)}_{j-1}}\\
        &\geq \frac{1}{2} \cdot \sum_{s=0}^{\ell} \sum_{j=i_{s}+1}^{i_{s+1}+1} \parentheses{N^{\text{res}(k)}_j- N^{\text{res}(k)}_{j-1}} \\
        &\geq \frac{1}{2} \cdot \sum_{s=0}^{\ell}  \epsilon = \epsilon \cdot \ell/2,
    \end{align*}
    i.e., $\ell \leq 2 N^{\text{res}(k)}_i/\epsilon$. Now we also observe that
    for every~$0 \leq s \leq \ell$ it holds that
    \begin{align*}
        \sum_{j=i_s+1}^{i_{s+1}}\frac{N^{\text{res}(k)}_j- N^{\text{res}(k)}_{j-1}}{W(N_j)+1-k} &\leq \sum_{j=i_s+1}^{i_{s+1}} \frac{N^{\text{res}(k)}_j- N^{\text{res}(k)}_{j-1}}{W(N_0)+s+1-k} \\
        & \leq \frac{\epsilon}{W(N_0)+s+1-k}.
    \end{align*}
    This is because~$W(N_j)+1-k = W(N_0)+j+1-k \geq W(N_0)+s+1-k$, since $j
    \geq i_s+1$ by the bounds of the summation and~$i_s+1 \geq s$ by the way we
    parametrize the~$i_s$ while defining them. Putting everything together
    yields
    \begin{align*}
        \sum_{j=0}^i  \frac{N^{\text{res}(k)}_j- N^{\text{res}(k)}_{j-1}}{W(N_j)+1-k} &\leq \frac{N^{\text{res}(k)}_0} {W(N_0)+1-k} + \sum_{s=0}^{\ell}  \frac{\epsilon}{W(N_0)+s+1-k} \\
        & \leq \epsilon + \epsilon \cdot \sum_{s=0}^{\ell}  \frac{1}{W(N_0)+s+1-k} \\
        & \leq \epsilon + \epsilon \cdot \int_{W(N_0)+1-k}^{W(N_0)+\ell+1-k}  \frac{1}{x} dx \\
        & \leq \epsilon + \epsilon \cdot \ln \ell \leq 2 \epsilon + \epsilon \cdot \ln N^{\text{res}(k)}_i.
    \end{align*}

    \medskip\noindent
    \textbf{The Sublinear Case of~$W(N)=N^{\alpha}/\epsilon$.}
    In this case, we switch to the notion of epochs where the number of
    counters doubles as opposed to increasing by one. That is, we define the
    epochs such that~$W(N_{i+1}) = 2\cdot W(N_i)$ and hence~$N_{i+1} =
    2^{1/\alpha} \cdot N_i$. We further define~$\gamma(i)$ as the smallest
    integer such that~$2^{\gamma(i)/\alpha}\cdot N^{\text{res}(k)}_0 \geq
    N^{\text{res}(k)}_i$, with~$\gamma(0)=0$. It immediately follows that
    $$ 2^{\gamma(i)/\alpha}\cdot N^{\text{res}(k)}_0 \leq 2^{1/\alpha} \cdot N^{\text{res}(k)}_i. $$
    We use $\gamma(i)$ to define a hypothetical upper bound on the error
    incurred if we were to double our counters as a function
    of~$N^{\text{res}(k)}$ rather than $N$, which is strictly worse for
    accuracy, yet creates a similar scenario to the proof of
    \Cref{thm:mg-normal} and yields the bounds \Cref{thm:mg-res}. We express
    this upper bound as
    $$ E'_i = \sum_{\ell=0}^{\gamma(i)} \frac{2^{\ell/\alpha} \cdot N^{\text{res}(k)}_0 - 2^{(\ell-1)/\alpha} \cdot N^{\text{res}(k)}_0}{2^{\ell} \cdot W\parentheses{N^{\text{res}(k)}_0 }-k+1}. $$
    We will relate the actual error of \sketchmg to~$E'$ to prove the desired
    upper bound. To this end, we make the following claim that for all~$i \geq
    0$:
    \begin{claim}
        \label{claim:ms-major}
        $$ \sum_{j=0}^i  \frac{N^{\text{res}(k)}_j- N^{\text{res}(k)}_{j-1}}{W(N_j)+1-k} \leq 4\cdot E'_i. $$
    \end{claim}	

    Before we prove Claim~\ref{claim:ms-major}, we would first like to discuss
    in more detail how it implies our bounds. Namely, we first note that each
    denominator can be substituted by~$2^{\ell} \cdot
    W\parentheses{N^{\text{res}(k)}_0}$. Then:
    \begin{align*}
        E'_i &\leq  \sum_{\ell=0}^{\gamma(i)} \frac{2^{\ell/\alpha} \cdot N^{\text{res}(k)}_0 - 2^{(\ell-1)/\alpha} \cdot N^{\text{res}(k)}_0}{2^{\ell} \cdot W\parentheses{N^{\text{res}(k)}_0 }} \\
        & =  \sum_{\ell=0}^{\gamma(i)} \frac{2^{\ell/\alpha} - 2^{(\ell-1)/\alpha} }{2^{\ell}}  \cdot \frac{N^{\text{res}(k)}_0}{W\parentheses{N^{\text{res}(k)}_0 }}\\
        & = O(1) \cdot  \frac{2^{\gamma(i)/\alpha} \cdot N^{\text{res}(k)}_0}{2^{\gamma(i)} \cdot W\parentheses{N^{\text{res}(k)}_0 }}, \\
        %	&\leq O(1) \cdot  \frac{2^{1/\alpha} \cdot N^{\text{res}(k)}_i}{2^{\gamma(i)} \cdot W\parentheses{N^{\text{res}(k)}_0 }}
    \end{align*}
    where the last equality holds by Claim~\ref{claim:technical}. We now note
    that, by definition, the numerator is at most~$2^{1/\alpha} \cdot
    N^{\text{res}(k)}_i$, and the denominator satisfies the following
    inequality:
    \begin{align*}
        2^{\gamma(i)} \cdot W\parentheses{N^{\text{res}(k)}_0 } &= 	\parentheses{2^{\gamma(i)/\alpha} \cdot N^{\text{res}(k)}_0}^\alpha /\epsilon \\
        & \geq \parentheses{N^{\text{res}(k)}_i}^\alpha /\epsilon=W\parentheses{N^{\text{res}(k)}_i}.
    \end{align*}
    Therefore, we have that
    $$ E'_i = O(1) \cdot \frac{ N^{\text{res}(k)}_i}{ W\parentheses{N^{\text{res}(k)}_i }}. $$
    In conjunction with~\Cref{lemma:mg-res-w}, we get the desired result. We
    now prove Claim~\ref{claim:ms-major}:
    \begin{proof}[Proof of Claim~\ref{claim:ms-major}]
        For $i=0$, note that $\gamma(0)=0$ and so:
        $$ \frac{N^{\text{res}(k)}_0}{W(N_0)+1-k} = E'_0$$

        \noindent
        For every $i \geq 1$, we prove that
        $$ \frac{N^{\text{res}(k)}_{i}- N^{\text{res}(k)}_{i-1}}{W(N_i)+1-k} \leq 2 \cdot  \sum_{\ell=\gamma(i-1)}^{\gamma(i)} \frac{2^{\ell/\alpha} \cdot N^{\text{res}(k)}_0 - 2^{(\ell-1)/\alpha} \cdot N^{\text{res}(k)}_0}{2^{\ell} \cdot W\parentheses{N^{\text{res}(k)}_0 }-k+1}. $$
        \noindent
        To this end, we first upper bound the denominators on the right-hand
        side as
        $$ 2^{\gamma(i)} \cdot W\parentheses{N^{\text{res}(k)}_0 }-k+1 \leq 2 \cdot \parentheses{{W(N_i)+1-k}}. $$
        \noindent
        This is because:
        \begin{align*}
            2^{\gamma(i)} \cdot W\parentheses{N^{\text{res}(k)}_0 } &=	2^{\gamma(i)} \cdot \parentheses{N^{\text{res}(k)}_0 }^{\alpha} /\epsilon = \parentheses{2^{\gamma(i)/\alpha} \cdot N^{\text{res}(k)}_0}^\alpha /\epsilon\\
            & \leq \parentheses{2^{1/\alpha} \cdot N^{\text{res}(k)}_i}^\alpha /\epsilon\\
            & \leq \parentheses{2^{1/\alpha} \cdot N_i}^\alpha /\epsilon= 2  \cdot W(N_i).
        \end{align*}
        \noindent
        For the numerators of the right-hand side, we observe that
        \begin{align*}
            \sum_{\ell=\gamma(i-1)}^{\gamma(i)} 2^{\ell/\alpha}  - 2^{(\ell-1)/\alpha} 
            & = \parentheses{2^{\gamma(i)/\alpha} - 2^{(\gamma(i-1)-1)/\alpha} }.
        \end{align*}

        Since, by the definition of~$\gamma(i)$, we have
        that~$2^{\gamma(i)/\alpha}\cdot N^{\text{res}(k)}_0 \leq 2^{1/\alpha}
        \cdot N^{\text{res}(k)}_i$ and $2^{(\gamma(i-1)-1)/\alpha}\cdot
        N^{\text{res}(k)}_0 \geq N^{\text{res}(k)}_{i-1}$, the
        numerators must satisfy:
        $$ \sum_{\ell=\gamma(i-1)}^{\gamma(i)} \parentheses{2^{\ell/\alpha} \cdot N^{\text{res}(k)}_0 - 2^{(\ell-1)/\alpha} \cdot N^{\text{res}(k)}_0} \geq N^{\text{res}(k)}_{i}- N^{\text{res}(k)}_{i-1}. $$ 

        Putting everything together yields
        \begin{align*}
            \frac{N^{\text{res}(k)}_{i}- N^{\text{res}(k)}_{i-1}}{W(N_i)+1-k} & \leq 
            \sum_{\ell=\gamma(i-1)}^{\gamma(i)} \frac{2^{\ell/\alpha} \cdot N^{\text{res}(k)}_0 - 2^{(\ell-1)/\alpha} \cdot N^{\text{res}(k)}_0}{W(N_i)+1-k}\\
            &\leq 2 \cdot  \sum_{\ell=\gamma(i-1)}^{\gamma(i)} \frac{2^{\ell/\alpha} \cdot N^{\text{res}(k)}_0 - 2^{(\ell-1)/\alpha} \cdot N^{\text{res}(k)}_0}{2^{\ell} \cdot W\parentheses{N^{\text{res}(k)}_0 }-k+1} \;,
        \end{align*}
        which is what we wanted to prove for every specific epoch~$i$. Summing
        these terms over all epochs and taking the double-counted terms into
        account then proves the claim.
    \end{proof}	
    As such, the theorem follows.
\end{proof}

\section{Lower Bound: Proof of Theorem~\ref{thm:sketch_expandability}}\label{sec:ommitted_proofs_lower}
\thmsketchexpandabilitylabel*
\begin{proof}
    Our proof follows a contradiction argument similar to the proof of
    Theorem~3.1 in~\cite{PaghExpandability}. That is, we consider an FE
    sketch~$\mathcal{S}$ that, given a key~$x$, computes an
    estimate~$\hat{f}(x)$ such that $|\mathbb{E}[\hat{f}(x)] - f(x)| \leq E(N)$
    and $\mathrm{Var}(\hat{f}(x)) \leq V(N) = \delta \cdot (E(N))^2$. Note
    that, by Chebyshev's inequality, this FE sketch guarantees an estimation
    error of~$2E(N)$ with a confidence of~$1-\delta$. We show that if such a
    sketch is ``too space-efficient,'' it can be used to compress many
    sequences of~$n=N/(2E(N))$ keys~$\mathbf{x}=x_1,\dots,x_n$ coming from a
    large universe~$U$ of size~$u \gg n$ to a size smaller than their entropy,
    which must be impossible. Our argument holds regardless of the time
    complexity of~$\mathcal{S}$'s operations. For simplicity, we prove this
    theorem assuming that~$\mathbb{E}[\hat{f}(x)]=f(x)$, which corresponds
    to~$\mathcal{S}$ having an estimation error of~$E(N)$ with a confidence
    of~$1-\delta$. Here, we compress sequences of~$n=N/E(N)$ keys. This proof
    is generalizable to the full statement of the theorem.

    Consider the process of iterating through~$\mathbf{x}$'s keys. We
    partition~$\mathbf{x}$ into contiguous subsequences~$\mathbf{x}_1,
    \mathbf{x}_2, \dots, \mathbf{x}_\ell$, with $\mathbf{x}_i$ having a length
    of~$2^i \cdot \sqrt{n}$ and $\ell=\log_2 n/2$. While iterating through
    these subsequences, we reinsert each key enough times into~$\mathcal{S}$
    such that by the time we are done with~$\mathbf{x}_1,\dots,\mathbf{x}_i$
    for any $i$, we have inserted each key a total~$E(N_i)$ times, where
    $N_i=E(N_i) \cdot \sum_{j=1}^i |\mathbf{x}_j|$. We denote the state
    of~$\mathcal{S}$ after processing the subsequences~$\mathbf{x}_1, \dots,
    \mathbf{x}_i$ as $\mathcal{S}_i$. We denote $\mathcal{S}_i$'s estimate for
    the count of key~$x$ as $\hat{f}_i(x)$ while denoting the actual count
    of~$x$ in $\mathbf{x}_1, \dots, \mathbf{x}_i$ as~$f_i(x)$. 

    We first show that, the estimates of one of the~$\mathcal{S}_i$s can be
    extended to an estimator that has a ``small'' second moment of at
    most~$V/\ell = 2 \cdot V/\log_2 n$. We will use these estimates to
    compress the sequence~$\mathbf{x}$. To this end, we apply the law of total
    variance to get, for any key~$x$,
    \begin{align*}
        \mathrm{Var} \left( \hat{f}_\ell(x) \right) & = \mathbb{E}\left[ \mathrm{Var}\left( \hat{f}_{\ell}(x) \;\middle|\; \hat{f}_{\ell - 1}(x) \right) \right]
                                                      + \mathrm{Var}\left( \mathbb{E}\left[ \hat{f}_{\ell}(x) \;\middle|\; \hat{f}_{\ell - 1}(x) \right] \right).
    \end{align*}
    Applying the law of total variance once more to the right summand
    simplifies it to
    \begin{align*}
        \mathrm{Var}\left( \mathbb{E}\left[ \hat{f}_\ell(x) \;\middle|\; \hat{f}_{\ell - 1}(x) \right] \right) & = \mathbb{E}\left[ \mathrm{Var}\left( \mathbb{E}\left[ \hat{f}_{\ell}(x) \;\middle|\; \hat{f}_{\ell - 1}(x) \right] \;\middle|\; \hat{f}_{\ell - 2}(x) \right) \right] \\
                                                      & \; \; + \mathrm{Var}\left( \mathbb{E}\left[ \mathbb{E}\left[ \hat{f}_{\ell}(x) \;\middle|\; \hat{f}_{\ell - 1}(x) \right] \;\middle|\; \hat{f}_{\ell - 2}(x) \right] \right) \\
                                                      & = \mathbb{E}\left[ \mathrm{Var}\left( \mathbb{E}\left[ \hat{f}_{\ell}(x) \;\middle|\; \hat{f}_{\ell - 1}(x) \right] \;\middle|\; \hat{f}_{\ell - 2}(x) \right) \right] \\
                                                      & \; \; + \mathrm{Var}\left( \mathbb{E}\left[ \hat{f}_\ell(x) \;\middle|\; \hat{f}_{\ell - 2}(x) \right] \right).
    \end{align*}
    Here, we have simplified the inner expectation in the second summand using
    the law of total expectation. Recursively applying the law of total
    variance this term yields, for any~$i$, that
    \begin{align*}
        \mathrm{Var}\left( \mathbb{E}\left[ \hat{f}_\ell(x) \;\middle|\; \hat{f}_i(x) \right] \right) & = \mathbb{E}\left[ \mathrm{Var}\left( \mathbb{E}\left[ \hat{f}_{\ell}(x) \;\middle|\; \hat{f}_i(x) \right] \;\middle|\; \hat{f}_{i - 1}(x) \right) \right] \\
                                                      & \; + \mathrm{Var}\left( \mathbb{E}\left[ \hat{f}_\ell(x) \;\middle|\; \hat{f}_{i - 1}(x) \right] \right),
    \end{align*}
    with~$\mathrm{Var}\left( \mathbb{E}\left[ \hat{f}_\ell(x) \;\middle|\; \hat{f}_0(x) \right] \right)=0$ 
    due to~$\hat{f}_0(x)$ and 
    thus~$\mathbb{E}\left[ \hat{f}_\ell(x) \;\middle|\; \hat{f}_0(x) \right]$
    being constants. Now, we utilize the definition of the variance to simplify
    the expectations of the variances as
    \begin{align*}
        & \mathbb{E}\left[ \mathrm{Var}\left( \mathbb{E}\left[ \hat{f}_\ell(x) \;\middle|\; \hat{f}_i(x) \right] \;\middle|\; \hat{f}_{i - 1}(x) \right) \right] \\
        & = \mathbb{E}\left[ \mathbb{E}\left[ \left( \mathbb{E}\left[ \hat{f}_\ell(x) \;\middle|\; \hat{f}_i(x) \right] - \mathbb{E}\left[ \mathbb{E}\left[ \hat{f}_\ell(x) \;\middle|\; \hat{f}_i(x) \right] \right] \right)^2 \;\middle|\; \hat{f}_{i - 1}(x) \right] \right] \\
        & = \mathbb{E}\left[ \left( \mathbb{E}\left[ \hat{f}_\ell(x) \;\middle|\; \hat{f}_i(x) \right] - \mathbb{E}\left[ \mathbb{E}\left[ \hat{f}_\ell(x) \;\middle|\; \hat{f}_i(x) \right] \right] \right)^2 \right] \\
        & = \mathbb{E}\left[ \left( \mathbb{E}\left[ \hat{f}_\ell(x) \;\middle|\; \hat{f}_i(x) \right] - \mathbb{E}\left[ \hat{f}_\ell(x) \right] \right)^2 \right] \\
        & = \mathbb{E}\left[ \left( \mathbb{E}\left[ \hat{f}_\ell(x) \;\middle|\; \hat{f}_i(x) \right] - f_\ell(x) \right)^2 \right] = \hat{e}_i(x).
    \end{align*}
    The third and fourth lines apply the law of total expectation to simplify
    the double expectations, and the fifth line utilizes the assumption
    of~$\mathbb{E}[\hat{f}_\ell(x)] = f_\ell(x)$. Intuitively, $\hat{e}_i(x)$
    is a measure of how much the conditional frequency
    estimate~$\mathbb{E}\left[ \hat{f}_\ell(x) \;\middle|\; \hat{f}_i(x) \right]$ 
    can deviate from~$f_\ell(x)$. Summing these terms yields
    \begin{align}
        V(N) \geq \mathrm{Var}\left( \hat{f}_\ell(x) \right) & = \mathrm{Var}\left( \mathbb{E}\left[ \hat{f}_\ell(x) \;\middle|\; \hat{f}_{\ell - 1}(x) \right] \right) + \sum_{i=1}^{\ell - 1} \hat{e}_i(x) \notag \\ 
        & = \sum_{i=1}^\ell \hat{e}_i(x). \label{eq:total_variance}
    \end{align}
    Here we have simplified the term~$\mathrm{Var}\left( \mathbb{E}\left[ \hat{f}_\ell(x) \;\middle|\; \hat{f}_{\ell - 1}(x) \right] \right)$ 
    to $\hat{e}_\ell$ by expressing~$\hat{f}_\ell(x)$ as 
    $\mathbb{E}\left[ \hat{f}_\ell(x) \;\middle|\; \hat{f}_\ell(x) \right]$ and
    simplifying similarly to the other terms. 

    \Cref{eq:total_variance} implies that, for at least half of the~$i$s, it
    holds that $\hat{e}_{i^*}(x) \leq 2 \cdot V(N)/\ell$. By an averaging
    argument, this further implies that there is some~$i^*$ such that
    $\hat{e}_{i^*}(x) \leq 2 \cdot V(N)/\ell$ holds for at least half of the
    elements~$x$ from the universe~$U$, which we denote by~$U'$ with $|U'| = u'
    \geq u/2 \gg n$. 

    We now focus on sequences~$\mathbf{x}$ consisting of the elements of~$U'$.
    We leverage the property of low~$\hat{e}_{i^*}(x)$ to compress the
    subsequence~$\mathbf{x}_i$. Our encoding proceeds as follows:
    \begin{enumerate}[label=(\roman*)]
        \item We write the explicit representations of the keys in all
            subsequences~$\mathbf{x}_i$, where $i > i^*$. The total cost for
            writing these representations is~$(n-\sum_{j=1}^{i^*}
            |\mathbf{x}_j|) \cdot \log_2 u'$ bits.
        \item We write the state of the FE sketch~$\mathcal{S}_{i^*}$. This
            takes~$|\mathcal{S}_{i^*}|$ bits, which is the quantity we want to
            bound. \\
            Since we know the sequences~$\mathbf{x}_{i^*+1}, \dots,
            \mathbf{x}_\ell$ and the state of the FE
            sketch~$\mathcal{S}_{i^*}$, and because we do not limit the running
            time, we can compute~$g(x)=\mathbb{E}\left[ \hat{f}_\ell(x)
            \;\middle|\; \hat{f}_{i^*}(x) \right]$ by setting the remaining
            randomness of~$\mathcal{S}_{i^*}$ to all possible combinations,
            processing $\mathbf{x}_{i^*+1}, \dots, \mathbf{x}_\ell$, and
            averaging the estimates~$\hat{f}_\ell(x)$. Note that here we are
            assuming that~$\mathcal{S}$ is a linear FE sketch, which allows us
            to scale its state to increase the counts of the keys
            in~$\mathbf{x}_1,\dots,\mathbf{x}_{i^*}$ to the count of~$E(N)$ of
            the entire stream without scanning them. Recall that 
            $$ \;\;\;\;\;\;\;\;\;\;\; 2 \cdot V(N)/\ell \geq \hat{e}_{i^*}(x) = \mathbb{E}\left[ \left( \mathbb{E}\left[ \hat{f}_\ell(x) \;\middle|\; \hat{f}_{i^*}(x) \right] - f_\ell(x) \right)^2 \right]. $$
            This implies that the estimate~$g(x)$ is unlikely to deviate more
            than~$E(N)$ from $f_\ell(x)$. That is, when~$g(x)/E(N)$ is smaller
            than the number of times~$x$ appears in~$\mathbf{x}$ outside
            of~$\mathbf{x}_1,\dots,\mathbf{x}_{i^*}$ by more than~$1/2$, $x$
            most probably does not appear
            in~$\mathbf{x}_1,\dots,\mathbf{x}_{i^*}$. This allows for reducing
            the number of bits required to encode each element
            in~$\mathbf{x}_1,\dots,\mathbf{x}_{i^*}$ in our compression scheme.
        \item Yet, since~$g(x)$ can suffer from severe underestimations, one
            cannot rely purely on its value to compress and decompress the
            subsequences~$\mathbf{x}_1,\dots,\mathbf{x}_{i^*}$. To address this
            issue, we store a bit for each key
            in~$\mathbf{x}_1,\dots,\mathbf{x}_{i^*}$ indicating whether~$g(x)$
            underestimates its count by in the entire stream by more
            than~$E(N)/2$. Each bit is~1 if the corresponding key has a large
            underestimation error, and~0 otherwise. This takes a total
            of~$\sum_{j=1}^{i^*} |\mathbf{x}_j|=(2^{i^*+1}-1) \cdot \sqrt{n}$
            bits.
        \item For each bit set to 1 in the previous step, we store the exact
            representation of the corresponding key using~$\log_2 u'$ bits. All
            keys in~$\mathbf{x}_1,\dots,\mathbf{x}_{i^*}$ associated with a~0
            bit either have overestimation errors or an absolute error smaller
            than~$E(N)/2$. As such, they only need to be disambiguated from the
            set of keys~$x$ where $g(x) \geq E(N)/2$.
        \item For each bit set to 0 in the third step, we store the index of
            the corresponding key in~$\mathbf{x}_{i^*}$ from the set of all
            keys~$y$ where $g(y)/E(N)$ is larger than the number of times~$y$
            appears in the rest of the sequence by at least~$1/2$. The cost per
            each key is the logarithm of the size of this set.
    \end{enumerate}

    We now fix the internal randomness of~$\mathcal{S}_{i^*}$ to convert it to
    a deterministic FE sketch giving rise to a~$g(\cdot)$ that, for many
    sequences~$\mathbf{x}$, only has estimation errors larger than~$E(N)/2$ on
    a $64\delta/\log_2 n$ fraction of the elements in~$U'$. We denote the
    fraction of such keys for a given sequence~$\mathbf{x}$
    as~$\mu(\mathbf{x})$. Notice how~$\mu(\mathbf{x})$ is relative size of the
    set we use to encode each key in the fifth step of the encoding. We also
    show that one can pick the internal randomness such that, in addition the
    aforementioned condition, $g(\cdot)$ errs more than~$E(N)$ on a fraction of
    at most~$64\delta/\log_2 n$ of the keys in the special
    subsequence~$\mathbf{x}_{i^*}$. We denote this fraction
    by~$\eta(\mathbf{x})$. Notice also how~$\eta(\mathbf{x})$ is fraction of
    the bits set to 1 in the third step of the encoding, which dictates the
    number of keys we store explicitly in the fourth step. 

    We prove the above statement following a similar argument to that of
    Section~3.1 in~\cite{PaghExpandability}:
    \begin{lemma}
        \label{lemma:derandomization}
        We can fix the internal randomness of the FE sketch~$\mathcal{S}_{i^*}$
        such that, for at least half of the possible sequences~$\mathbf{x}$,
        i.e., ${u'}^n/2$ many of them, $\mu(\mathbf{x}) \leq 64\delta/\log_2 n$
        and $\eta(\mathbf{x}) \leq 64\delta/\log_2 n$.
    \end{lemma}
    \begin{proof}
        By Markov's inequality, the~$g(x)$ resulting from the randomized
        version of~$\mathcal{S}_{i^*}$ has errors bounded by~$E(N)/2$ with a
        confidence of at least~$1-8\delta/\ell=1-16\delta/\log_2 n$ for any $x
        \in U'$ when processing the subsequence~$\mathbf{x}_{i^*}$. Thus,
        denoting the internal random bits of~$\mathcal{S}_{i^*}$ by $r$, for
        any~$\mathbf{x} \in {U'}^n$, we have that~$\mathbb{E}_r\left[ 
        \mu(\mathbf{x}) \right] \leq 16\delta/\log_2 n$ and 
        $\mathbb{E}_r\left[ \eta(\mathbf{x}) \right] \leq 16\delta/\log_2 n$.
        Therefore, once again by Markov's inequality, we have
        $$ \Pr_r\left[ \mu(\mathbf{x}) \geq 64\delta/\log_2 n \right] \leq \frac{1}{4}, \;\;\;\; \Pr_r\left[ \eta(\mathbf{x}) \geq 64\delta/\log_2 n \right] \leq \frac{1}{4}. $$
        Thus, by a union bound, we have that 
        $$ \Pr_r\left[ \mu(\mathbf{x}) \geq 64\delta/\log_2 n \;\vee\; \eta(\mathbf{x}) \geq 64\delta/\log_2 n \right] \leq \frac{1}{2}. $$
        Taking the complement of the above yields 
        $$ \Pr_r\left[ \mu(\mathbf{x}) < 64\delta/\log_2 n \;\wedge\; \eta(\mathbf{x}) < 64\delta/\log_2 n \right] \geq \frac{1}{2}. $$
        That is, each~$\mathbf{x}$ satisfies~$\mu(\mathbf{x}) < 64\delta/\log_2
        n$ and $\eta(\mathbf{x}) < 64\delta/\log_2 n$ with a probability of at
        least~$1/2$. An averaging argument implies that there is some choice
        for the internal randomness of~$\mathcal{S}_{i^*}$, i.e., $r$, that
        satisfies~$\mu(\mathbf{x}) < 64\delta/\log_2 n$ and $\eta(\mathbf{x}) <
        64\delta/\log_2 n$ at the same time for at least half of all~${u'}^n$
        sequences.
    \end{proof}

    We thus fix the internal randomness of~$\mathcal{S}_{i^*}$ to the string
    indicated by \Cref{lemma:derandomization} and use the estimation
    function~$g(\cdot)$ that arises from it to encode the an arbitrary sequence
    from the ${u'}^n/2$ sequences in \Cref{lemma:derandomization}. The first,
    second, and third steps of the encoding take
    up~$(n-(2^{i^*+1}-1)\cdot\sqrt{n}) \cdot \log_2 u'$, $|\mathcal{S}_{i^*}|$,
    and $(2^{i^*+1}-1)\cdot\sqrt{n}$ bits, respectively. Since we have
    $\eta(\mathbf{x}) \leq 64\delta/\log_2 n$ for all sequencees $\mathbf{x}$
    that we consider, the number of 1s we store in the bitmap in the third step
    is at most~$\frac{64\delta}{\log_2 n} \cdot (2^{i^*+1}-1)\cdot\sqrt{n}$ bits. We
    explicitly store the keys that correspond to these bits in the fourth step,
    taking~$\frac{64\delta}{\log_2 n} \cdot (2^{i^*+1}-1)\cdot\sqrt{n} \cdot
    \log_2 u'$ bits. Finally, for the keys corresponding to the~0 bits in the
    bitmap from the third step, since~$\mu(\mathbf{x})$ determines the size of
    the set we index in the last step, it takes
    { \small
    \begin{align*}
        & \left( 1 - \frac{64\delta}{\log_2 n} \right) \cdot (2^{i^*+1}-1)\cdot\sqrt{n} \cdot \log_2 \left( \mu(\mathbf{x}) \cdot u' \right) \\
        & \;\;\;\; = \left( 1 - \frac{64\delta}{\log_2 n} \right) \cdot (2^{i^*+1}-1)\cdot\sqrt{n} \cdot \left[ \log_2 u' - \log_2\log_2 n - \log_2 \frac{1}{\delta} + 4 \right]
    \end{align*}
    }
    bits to encode them. The sum of these costs cannot be smaller than the
    entropy of encoding a random sequence from these~${u'}^n/2$ sequences,
    i.e., $n \cdot \log_2 u' - 1$ bits. That is, summing up the costs of the
    encoding, the following inequality holds:
    \begin{align*}
        |&\mathcal{S}_{i^*}| + n \cdot \log_2 u' \\ 
        & + (2^{i^*+1}-1) \cdot \sqrt{n} \cdot \left( 1 - \left( 1-\frac{64\delta}{\log_2 n} \right) \cdot \left[ \log_2\log_2 n + \log_2 \frac{1}{\delta} - 4 \right] \right) \\
        & \geq n \cdot \log_2 u' - 1.
    \end{align*}
    Simplifying the above inequality, we get that $|\mathcal{S}_{i^*}|$ must be
    at least
    \begin{align*}
        (2^{i^*+1}-1) \sqrt{n} \cdot \left( \left( 1-\frac{64\delta}{\log_2 n} \right) \cdot \left[ \log_2\log_2 n + \log_2 \frac{1}{\delta} - 4 \right] - 1 \right) - 1,
    \end{align*}
    which is equivalent to
    \begin{align*}
        |\mathcal{S}_{i^*}| \geq \sum_{j=1}^{i^*} |\mathbf{x}_j| \cdot \left( \log_2\log_2 n + \log_2 \frac{1}{\delta} - \Theta(1) \right),
    \end{align*}
    for~$n$ large enough or~$\delta$ small enough. Letting~$n' =
    \sum_{j=1}^{i^*} |\mathbf{x}_j|$ and noting that~$\log_2\log_2 n \geq
    \log_2\log_2 n'$, we have that 
    \begin{align*}
        |\mathcal{S}_{i^*}| \geq n' \cdot \left( \log_2\log_2 n' + \log_2 \frac{1}{\delta} - \Theta(1) \right).
    \end{align*}
    Finally, recalling that the stream processed by~$\mathcal{S}_{i^*}$ had a
    length of~$N_{i^*}=n' \cdot E(N_{i^*})$ and rewriting $n'$ as
    $N_{i^*}/E(N_{i^*})$ yields
    \begin{align*}
        |\mathcal{S}_{i^*}| \geq N_{i^*}/E(N_{i^*}) \cdot \left[ \log_2\log_2 (N_{i^*}/E(N_{i^*})) + \log_2 \frac{1}{\delta} - \Theta(1) \right],
    \end{align*}
    which proves the desired result.
\end{proof}
\fi

\end{document}